\documentclass[letterpaper]{article} \usepackage{hyperref}
\usepackage{icml21-arxiv}

\usepackage[verbose=true,letterpaper]{geometry}
\AtBeginDocument{
  \newgeometry{
    textheight=9in,
    textwidth=6in,
    top=1in,
    headheight=12pt,
    headsep=25pt,
    footskip=30pt
  }
}
\usepackage{times}
\usepackage{helvet}
\usepackage{courier}
\usepackage{graphicx}
\usepackage{graphicx}
\usepackage{natbib}
\let\cite\citep

\usepackage{caption}
\newcommand{\parencite}{\citep}
\newcommand{\textcite}{\citet}

\icmltitlerunning{Efficient Deviation Types and Learning}

\usepackage{xcolor}
\definecolor{mydarkblue}{rgb}{0,0.08,0.45}
\hypersetup{ pdflang=en-US,
    pdftitle={Efficient Deviation Types and Learning for Hindsight Rationality in Extensive-Form Games},
    pdfauthor={Dustin Morrill,Ryan D'Orazio,Marc Lanctot,James R. Wright,Michael Bowling,Amy R. Greenwald},
    pdfsubject={Proceedings of the International Conference on Machine Learning 2021},
    pdfkeywords={ICML,[Game Theory and Economic Paradigms] Equilibrium,[Game Theory and Economic Paradigms] Imperfect Information,[Multiagent Systems] Multiagent Learning,[Machine Learning] Online Learning & Bandits},
    pdfborder=0 0 0,
    pdfpagemode=UseNone,
    colorlinks=true,
    linkcolor=mydarkblue,
    citecolor=mydarkblue,
    filecolor=mydarkblue,
    urlcolor=mydarkblue,
    pdfview=FitH}

\def\correctionPaperReferenceKey/{edl2021arxivCorrections}

\usepackage{mathtools}

\DeclarePairedDelimiter{\abs}{\lvert}{\rvert}
\DeclarePairedDelimiter{\norm}{\lVert}{\rVert}
\DeclarePairedDelimiter{\subex}{(}{)}
\DeclarePairedDelimiter{\subblock}{[}{]}
\DeclarePairedDelimiter{\tuple}{(}{)}
\DeclarePairedDelimiter{\set}{\{}{\}}

\usepackage{bbold}
\usepackage{bm}

\newcommand{\reals}{\mathbb{R}}

\newcommand{\Simplex}{\Delta}
\newcommand{\simplex}{\Simplex}
\newcommand{\like}{\widetilde}

\newcommand{\bs}[1]{\bm{#1}}

\newcommand{\expectation}{\mathbb{E}}
\newcommand{\E}{\expectation}

\newcommand{\as}{\doteq}

\newcommand{\zeros}{\bs{0}}

\newcommand{\bigO}[1]{\operatorname{\mathcal{O}}\subex{#1}}

\newcommand{\ip}[2]{\langle #1, \, #2 \rangle}
\newcommand{\ind}[1]{\mathbb{1}\set*{#1}}
\newcommand{\given}{\,|\,}

\newcommand{\where}{\;|\;}

\newcommand{\PureStratSet}{S}
\newcommand{\pureStrat}{s}
\newcommand{\pureProfile}{\pureStrat}

\newcommand{\reward}{r}
\newcommand{\utility}{u}
\newcommand{\policy}{\pi}
\newcommand{\StrategySet}{\Pi}
\newcommand{\strategy}{\policy}
\newcommand{\strat}{\strategy}

\newcommand{\recDist}{\mu}

\newcommand{\Actions}{\mathcal{A}}

\newcommand{\regret}{\rho}

\newcommand{\supReward}{U}
\newcommand{\maxReward}{\supReward}
\newcommand{\grad}{\nabla}

\newcommand{\EXT}{\textsc{ex}}
\newcommand{\IN}{\textsc{in}}
\newcommand{\INT}{\IN}
\newcommand{\SWAP}{\textsc{sw}}

\newcommand{\infoSet}{I}
\newcommand{\InfoSets}{\mathcal{I}}
\newcommand{\reachProb}{P}
\newcommand{\chance}{c}

\newcommand{\Histories}{\mathcal{H}}
\newcommand{\TerminalHistories}{\mathcal{Z}}

\newcommand{\playerChoice}{\mathcal{P}}
\usepackage{amssymb}
\newcommand{\emptyHistory}{\varnothing}

\newcommand{\termValue}{\reward}
\newcommand{\cfIv}{v}

\newcommand{\cfv}{\cfIv}

\usepackage{upgreek}

\newcommand{\immStrat}{\sigma}

\newcommand{\DevSet}{\Phi}
\newcommand{\dev}{\phi}

\newcommand{\COUNTERFACTUAL}{\textsc{cf}}
\newcommand{\CF}{\COUNTERFACTUAL}

\newcommand{\BHV}{\IN}

\newcommand{\TARGET}{\odot}
\newcommand{\TRIGGER}{\text{!}}

\newcommand{\tSelectionFn}{w}
\newcommand{\TSelectionSet}{W}

\newcommand{\rmOperator}{L}

\newcommand{\maxActivation}{\omega}

\newcommand{\parent}{\mathbb{p}}

\newcommand{\parentAction}{\mathbb{a}}
\newcommand{\arbHistory}{\mathbb{h}}
\newcommand{\infoSetOf}{\mathbb{I}}
 
\usepackage{xspace}

\makeatletter
\DeclareRobustCommand\onedot{\futurelet\@let@token\@onedot}
\def\@onedot{\ifx\@let@token.\else.\null\fi\xspace}

\def\eg/{\emph{e.g}\onedot} \def\Eg/{\emph{E.g}\onedot}
\def\ie/{\emph{i.e}\onedot} \def\Ie/{\emph{I.e}\onedot}
\def\cf/{\emph{c.f}\onedot} \def\Cf/{\emph{C.f}\onedot}
\def\vs/{\emph{vs}\onedot} \def\Vs/{\emph{Vs}\onedot}
\def\etc/{\emph{etc}\onedot}
\def\wrt/{w.r.t\onedot} \def\dof/{d.o.f\onedot}
\def\etal/{\emph{et al}\onedot}
\def\viceversa/{\emph{vice-versa}}
\def\ow/{\emph{o.w}\onedot}
\def\whp/{w.h.p\onedot}
\def\apriori/{\emph{a priori}} \def\Apriori/{\emph{A priori}}
\def\ala/{\`{a} la}

\def\naive/{na\"{\i}ve} \def\Naive/{Na\"{\i}ve}
\def\rmPlus/{regret matching\textsuperscript{+}}
\def\rrmPlus/{RRM\textsuperscript{+}}
\def\rcfrPlus/{RCFR\textsuperscript{+}}
\def\cfrPlus/{CFR\textsuperscript{+}}
\@ifdefinable{\Politex/}{\def\Politex/{\textsc{Politex}}}

\def\NashConv/{\textsc{NashConv}}
\def\NashConvAUC/{$\overline{\textsc{NashConv}}$}

\makeatother

\newcommand{\textbxf}[1]{{\fontseries{b}\selectfont #1}}

\def\efceFootnote/{3}
 \definecolor{offWhite}{RGB}{240,240,240}
\definecolor{grey}{RGB}{180,180,180}
\definecolor{darkgreen}{RGB}{0,125,0}
\definecolor{lime}{RGB}{255,200,0}

\definecolor{amiiBlue}{RGB}{16,72,118}
\definecolor{amiiPink}{RGB}{241,97,119}
\definecolor{amiiYellow}{RGB}{248,209,109}
\definecolor{amiiPurple}{RGB}{123,105,145}

\colorlet{yes}{cyan!50!white}
\colorlet{newYes}{cyan!75!white}
\colorlet{no}{red!50!white}
\colorlet{newNo}{red!75!white}

\usepackage{amsthm}
\makeatletter
\ifdefined\theorem
\else
  \newtheorem{theorem}{Theorem}
\fi
\ifdefined\lemma
\else
  \newtheorem{lemma}{Lemma}
\fi
\ifdefined\corollary
\else
  \newtheorem{corollary}{Corollary}
\fi
\ifdefined\definition
\else
  \newtheorem{definition}{Definition}
\fi
\ifdefined\proposition
\else
  \newtheorem{proposition}{Proposition}
\fi
\makeatother

\usepackage{booktabs}
\usepackage{multirow}
\usepackage{pbox}
\usepackage{array}

\def\stretchBox{\pbox{\linewidth}}

\usepackage{threeparttable}
\usepackage{tabularx}

\DeclareMathSymbol{\negative}{\mathbin}{AMSa}{"39}

\makeatletter
\ifdefined\algorithmic
\else
  \usepackage{algorithm}
  \usepackage{algorithmic}
\fi
\makeatother
\newcommand{\State}{\STATE}
\newcommand{\For}{\FOR}
\newcommand{\EndFor}{\ENDFOR}
\newcommand{\Return}{\OUTPUT}
\newcommand{\If}{\IF}
\newcommand{\Else}{\ELSE}
\newcommand{\EndIf}{\ENDIF}
\newcommand{\CommentChar}{\#}
\newcommand{\Indent}[1][1]{\hspace{#1em}}
\newcommand{\Input}{\State \textbf{Input:}}
\newcommand{\LComment}[1]{\State \CommentChar\ #1}

\makeatletter
\newcommand\safeIncCounter[1]{\@ifundefined{c@#1}{\newcounter{#1}\stepcounter{#1}}{\stepcounter{#1}}}
\makeatother

\newcounter{resetCounter}

\usepackage{xargs}

\def\gTwoFive/{g\textsubscript{2, 5, $\uparrow$}}
\def\gThreeFour/{g\textsubscript{3, 4, $\uparrow$}}

\usepackage{nomencl}
\makenomenclature

\makeatletter
\ifdefined\isCorrectionsPaper
  \def\correctionBlockDescription/{The changes to the original paper pertaining to this correction can be found within labeled and numbered ``correction'' blocks.
    A brief description of the specific change within a given correction block can be found at the start of that block.}
\else
  \def\correctionBlockDescription/{The changes to the original paper pertaining to this correction are highlighted with labeled and numbered ``correction'' blocks in the separate report \textcite{\correctionPaperReferenceKey/}.}
\fi
\makeatother

\newcommand{\edlIntroCorrectionsFootnote}{
  \footnote{This paper is changed from its original version published as \textcite{edl2021}.
    That version asserts that counterfactual and partial sequence deviations subsume external and causal deviations, respectively, without qualification.
    This assertion is incorrect, as \textcite{macQueen2022counterexampleSlackMessages}'s counterexample shows.
    This example and its consequences are described in the corrected version of \textcite{hsr2020arxiv}.
    Instead, observable sequential rationality is required to ensure that counterfactual deviations and partial sequence deviations subsume external deviations and causal deviations, respectively.
    \correctionBlockDescription/
  }
}
 \usepackage{tikz}
\usetikzlibrary{
    shapes.geometric,
    arrows,
    arrows.meta,
    graphs,
    graphs.standard,
    calc,
    positioning,
    intersections,
    chains,
    matrix,
    shapes.misc,
    backgrounds,
    overlay-beamer-styles,
    decorations.pathmorphing,
    decorations.pathreplacing,
    shapes.symbols,
    trees
}

\def\cmToEmFactor{2.3710630158366}

\tikzset{
  influenceArrow/.style={>=Triangle, -{>[scale=#1]}},
  influenceArrow/.default=0.4
}

\tikzstyle{devColor} = [draw=red]
\tikzstyle{dev} = [devColor, very thick]

\tikzstyle{followColor} = [draw=black]
\tikzstyle{follow} = [followColor, very thick]

\tikzstyle{infoColor} = [draw=cyan]
\tikzstyle{info} = [infoColor, very thick]
\tikzstyle{infoArrow} = [influenceArrow, thin, infoColor]

\tikzstyle{alt} = [draw=grey]
\tikzstyle{zeroProb} = [draw=grey]
\tikzstyle{rec} = [draw=black, densely dashed, very thick]
\def\stateMinimumRadius{0.18cm}
\tikzstyle{state} = [circle, draw=black, minimum size=2*\stateMinimumRadius, inner sep=0.5mm, fill=white]
\tikzstyle{util} = [inner sep=1mm]
\tikzstyle{behaveLabel} = [text width=2cm]
\tikzstyle{actionArrow} = [thick, >=Stealth, -{>[scale=0.7]}]
\tikzstyle{showPoint} = [shape=circle, fill=black, minimum size=0.3em]

 \tikzset{
  seqPath/.style args={segment length #1 amplitude #2}{
    very thick,
    line join=round,
    decorate,
    decoration={zigzag, segment length=#1, amplitude=#2}},
  seqPath/.default={segment length 0.2cm amplitude 1mm}
}

\def\actLen{\cmToEmFactor*0.5em}
\def\seqLen{\actLen}

\tikzset{
    absDecisionDag/.style args={size #1}{very thick, inner sep=0em, fill=white, regular polygon, regular polygon sides=3, minimum size=#1+#1/3.4641},
    absDecisionDag/.default={size \seqLen}
}
\newcommandx{\inTrans}[3][1=\actLen, 2=very thick]{
  \def\inTransActLen{#1}
  \def\inTransActionStyle{#2}
  \def\inTransRoot{#3}

\coordinate(\inTransRoot/dev) at ($(\inTransRoot)-(0,\inTransActLen)$);
  \coordinate(\inTransRoot/info) at ($(\inTransRoot/dev)+(-\inTransActLen,0.05em)$);
  \coordinate(\inTransRoot/influenceSource)
    at ($(\inTransRoot)!0.8!(\inTransRoot/info)+(0.1em, 0)$);
  \coordinate(\inTransRoot/south)
    at ($(\inTransRoot/dev)!0.5!(\inTransRoot/info)$);
  \coordinate(\inTransRoot/north)
    at (\inTransRoot -| \inTransRoot/south);
  \coordinate(\inTransRoot/center)
    at ($(\inTransRoot/north)!0.5!(\inTransRoot/south)$);

  \draw[\inTransActionStyle, rounded corners]
    (\inTransRoot/info)
    edge[info]
    ($(\inTransRoot)-(0,0.05em)$)
    (\inTransRoot)
    edge[dev]
    (\inTransRoot/dev);
}

\newcommandx{\seqInTrans}[3][1=follow, 2=\seqLen]{
  \def\seqInTransStyle{#1}
  \def\seqInTransSeqLen{#2}
  \def\seqInTransRoot{#3}

\coordinate(\seqInTransRoot/seq) at ($(\seqInTransRoot.south)-(0,\seqInTransSeqLen)$);
  \inTrans{\seqInTransRoot/seq}

  \draw[seqPath, \seqInTransStyle] (\seqInTransRoot.north) -- (\seqInTransRoot/seq);
}

\tikzstyle{tree} = [absDecisionDag={size \seqLen}]

\newcommand{\inDevDiagram}[1]{
  \def\root{#1}
  \node(\root/recState)
    [tree, info, anchor=north]
    at ($(\root.south)+(-\cmToEmFactor*0.3em,0)$)
    {};

  \node(\root/devState)
    [tree, dev, anchor=north]
    at ($(\root.south)+(\cmToEmFactor*0.3em,0)$)
    {};
  \draw[infoArrow, bend left=50]
    (\root/recState)
    edge[infoColor]
    (\root/devState);

\coordinate(\root/south) at (\root.north |- \root/recState.south);
  \coordinate(\root/center) at ($(\root.north)!0.5!(\root/south)$);
  \coordinate(\root/west) at (\root/center -| \root/recState.west);
}
\newcommand{\behavDevDiagram}[1]{
  \def\behavDevDiagramRoot{#1}

  \coordinate(\behavDevDiagramRoot/_) at (\behavDevDiagramRoot);
  \inTrans{\behavDevDiagramRoot/_};
  \inTrans{\behavDevDiagramRoot/_/dev};
  \inTrans{\behavDevDiagramRoot/_/dev/dev};

\coordinate(\behavDevDiagramRoot/north)
    at (\behavDevDiagramRoot -| \behavDevDiagramRoot/_/south);
  \coordinate(\behavDevDiagramRoot/south)
    at (\behavDevDiagramRoot/_/dev/dev/dev -| \behavDevDiagramRoot/north);
  \coordinate(\behavDevDiagramRoot/center)
    at ($(\behavDevDiagramRoot/north)!0.5!(\behavDevDiagramRoot/south)$);
  \coordinate(\behavDevDiagramRoot/east) at (\behavDevDiagramRoot/center -| \behavDevDiagramRoot);
}
\newcommand{\tipsDevDiagram}[1]{
  \def\tipsDevDiagramRoot{#1}

  \seqInTrans{\tipsDevDiagramRoot}
  \seqInTrans[dev]{\tipsDevDiagramRoot/seq/dev}
  \node(\tipsDevDiagramRoot/followTree)
    [tree, follow, anchor=north]
    at (\tipsDevDiagramRoot/seq/dev/seq/dev) {};

\coordinate(\tipsDevDiagramRoot/north)
    at (\tipsDevDiagramRoot);
  \coordinate(\tipsDevDiagramRoot/south)
    at (\tipsDevDiagramRoot/north |- \tipsDevDiagramRoot/followTree.south);
  \coordinate(\tipsDevDiagramRoot/center)
    at ($(\tipsDevDiagramRoot/north)!0.5!(\tipsDevDiagramRoot/south)$);
  \coordinate(\tipsDevDiagramRoot/west)
    at (\tipsDevDiagramRoot/center -| \tipsDevDiagramRoot/seq/info);
}
\newcommand{\cspsDevDiagram}[1]{
  \def\cspsDevDiagramRoot{#1}

  \seqInTrans{\cspsDevDiagramRoot}
  \node(\cspsDevDiagramRoot/followTree)
    [tree, follow, anchor=north]
    at ($(\cspsDevDiagramRoot/seq/dev)-(0,\seqLen)$)
    {};
  \draw[seqPath, dev]
    (\cspsDevDiagramRoot/seq/dev)
    --
    (\cspsDevDiagramRoot/followTree);

\coordinate(\cspsDevDiagramRoot/north)
    at (\cspsDevDiagramRoot);
  \coordinate(\cspsDevDiagramRoot/south)
    at (\cspsDevDiagramRoot/north |- \cspsDevDiagramRoot/followTree.south);
  \coordinate(\cspsDevDiagramRoot/center)
    at ($(\cspsDevDiagramRoot/north)!0.5!(\cspsDevDiagramRoot/south)$);
}
\newcommand{\cfpsDevDiagram}[1]{
  \def\cfpsDevDiagramRoot{#1}

  \coordinate(\cfpsDevDiagramRoot/seq) at ($(\cfpsDevDiagramRoot)-(0, \seqLen)$);
  \seqInTrans[dev]{\cfpsDevDiagramRoot/seq};
  \node(\cfpsDevDiagramRoot/followTree)
    [tree, follow, anchor=north]
    at (\cfpsDevDiagramRoot/seq/seq/dev)
    {};
  \draw[seqPath, follow]
    (\cfpsDevDiagramRoot)
    --
    (\cfpsDevDiagramRoot/seq);

\coordinate(\cfpsDevDiagramRoot/north)
    at (\cfpsDevDiagramRoot);
  \coordinate(\cfpsDevDiagramRoot/south)
    at (\cfpsDevDiagramRoot/north |- \cfpsDevDiagramRoot/followTree.south);
  \coordinate(\cfpsDevDiagramRoot/center)
    at ($(\cfpsDevDiagramRoot/north)!0.5!(\cfpsDevDiagramRoot/south)$);
}
\newcommand{\inCsDevDiagram}[1]{
  \def\root{#1}
  \seqInTrans{\root};

  \node(\root/devTree) [tree, devColor, anchor=north] at (\root/seq/dev) {};

\coordinate(\root/south) at (\root/devTree.south);
  \coordinate(\root/center) at ($(\root.north)!0.5!(\root/south)$);
  \coordinate(\root/north east) at (\root.north -| \root/devTree.south east);
}
\newcommand{\inCfDevDiagram}[1]{
  \def\root{#1}
  \seqInTrans[dev]{\root};
  \node(\root/followTree)
    [tree, follow, anchor=north]
    at (\root/seq/dev)
    {};

\coordinate(\root/south) at (\root/followTree.south);
  \coordinate(\root/center) at ($(\root.north)!0.5!(\root/south)$);
}
\newcommand{\inActDevDiagram}[1]{
  \def\root{#1}
  \seqInTrans{\root};
  \node(\root/followTree)
    [tree, follow, anchor=north]
    at (\root/seq/dev)
    {};

\coordinate(\root/south) at (\root/followTree.south);
  \coordinate(\root/center) at ($(\root.north)!0.5!(\root/south)$);
}
\newcommand{\bpsDevDiagram}[1]{
  \def\bpsDevDiagramRoot{#1}

  \coordinate(\bpsDevDiagramRoot/seq) at ($(\bpsDevDiagramRoot)-(0, \seqLen)$);

  \node(\bpsDevDiagramRoot/followTree)
    [tree, follow, anchor=north]
    at ($(\bpsDevDiagramRoot/seq)-(0, \seqLen)$)
    {};
  \draw[seqPath, follow]
    (\bpsDevDiagramRoot)
    --
    (\bpsDevDiagramRoot/seq);
  \draw[seqPath, dev]
    (\bpsDevDiagramRoot/seq)
    --
    (\bpsDevDiagramRoot/followTree);

\coordinate(\bpsDevDiagramRoot/north)
    at (\bpsDevDiagramRoot);
  \coordinate(\bpsDevDiagramRoot/south)
    at (\bpsDevDiagramRoot/north |- \bpsDevDiagramRoot/followTree.south);
  \coordinate(\bpsDevDiagramRoot/center)
    at ($(\bpsDevDiagramRoot/north)!0.5!(\bpsDevDiagramRoot/south)$);
}
 \usepackage[textwidth=2in,colorinlistoftodos,prependcaption,textsize=footnotesize]{todonotes}
\expandafter\newif\csname ifGin@setpagesize\endcsname
\newcommand{\todonote}[4][inline]{\safeIncCounter{#2NoteCounter}
  \todo[color=offWhite,bordercolor=#3,linecolor=#3,#1]{\textbf{\uppercase{#2}$_{\arabic{#2NoteCounter}}$:}~#4}}

\usepackage[normalem]{ulem}

\newcommand{\replaced}[3]{\def\counterPrefix{#1}
  \def\arrowMarker{#2}
  \def\replacedText{#3}
  \todo[color=offWhite,bordercolor=red,inline]{$\bs{\arrowMarker}$ \textbf{Replaced (\arabic{Replaced\counterPrefix{}NoteCounter})} \replacedText }}

\newcommand{\replacedStart}[2]{\def\user{#1}
  \def\text{#2}
  \safeIncCounter{Replaced#1NoteCounter}\replaced{\user}{\downarrow}{\text}}
\newcommand{\replacedEnd}[1]{\def\user{#1}
  \replaced{\user}{\uparrow}{}}

\newcommand{\issue}[3]{\todo[color=black,inline]{\textcolor{white}{$\bs{#2}$ \textbf{Issue \##3} (Part \arabic{Issue#3NoteCounter}) #1}}}

\newcommand{\issueChangeStart}[2][]{\safeIncCounter{Issue#2NoteCounter}\issue{#1}{\downarrow}{#2}}
\newcommand{\issueChangeEnd}[2][]{\issue{#1}{\uparrow}{#2}}

\newcommand{\correction}[2]{\def\arrowMarker{#1}
  \def\correctionLabel{#2}
  \todo[color=offWhite,bordercolor=red,inline]{$\bs{\arrowMarker}$ \textbf{Correction (\arabic{CorrectionNoteCounter})} \correctionLabel }}
\newcommand{\correctionStart}[1]{\def\correctionStartLabel{#1}
  \safeIncCounter{CorrectionNoteCounter}
  \correction{\downarrow}{\correctionStartLabel}}
\newcommand{\correctionEnd}{\correction{\uparrow}{}}
 \renewcommand{\todonote}[4][inline]{\ignorespaces}
\renewcommand{\issueChangeStart}[2][]{\ignorespaces}
\renewcommand{\issueChangeEnd}[2][]{\ignorespaces}
\renewcommand{\replacedStart}[2][]{\ignorespaces}
\renewcommand{\replacedEnd}[1][]{\ignorespaces}
\renewcommand{\correctionStart}[2][]{\ignorespaces}
\renewcommand{\correctionEnd}[1][]{\ignorespaces}
 
\usepackage{nicefrac}
\usepackage[capitalize,noabbrev]{cleveref}

\begin{document}
\icmltitle{Efficient Deviation Types and Learning for Hindsight Rationality in Extensive-Form Games}

\begin{icmlauthorlist}
  \icmlauthor{Dustin Morrill}{alberta}
  \icmlauthor{Ryan D'Orazio}{montreal}
  \icmlauthor{Marc Lanctot}{deepmind}\\
  \icmlauthor{James R.\ Wright}{alberta}
  \icmlauthor{Michael Bowling}{alberta,deepmind}
  \icmlauthor{Amy R.\ Greenwald}{brown}
\end{icmlauthorlist}

\icmlaffiliation{alberta}{Department of Computing Science, University of Alberta; Alberta Machine Intelligence Institute, Edmonton, Alberta, Canada}
\icmlaffiliation{montreal}{DIRO, Universit\'{e} de Montr\'{e}al; Mila, Montr\'{e}al, Qu\'{e}bec, Canada}
\icmlaffiliation{deepmind}{DeepMind}
\icmlaffiliation{brown}{Computer Science Department, Brown University, Providence, Rhode Island, United States}

\icmlcorrespondingauthor{Dustin Morrill}{morrill@ualberta.ca}

\icmlkeywords{ICML,[Game Theory and Economic Paradigms] Equilibrium,[Game Theory and Economic Paradigms] Imperfect Information,[Multiagent Systems] Multiagent Learning,[Machine Learning] Online Learning & Bandits}

\vskip 0.3in
\printAffiliationsAndNotice{}

\begin{abstract}
  Hindsight rationality is an approach to playing general-sum games that prescribes no-regret learning dynamics for individual agents with respect to a set of deviations, and further describes jointly rational behavior among multiple agents with mediated equilibria.
To develop hindsight rational learning in sequential decision-making settings, we formalize behavioral deviations as a general class of deviations that respect the structure of extensive-form games.
Integrating the idea of time selection into counterfactual regret minimization (CFR), we introduce the extensive-form regret minimization (EFR) algorithm that achieves hindsight rationality for any given set of behavioral deviations with computation that scales closely with the complexity of the set.
We identify behavioral deviation subsets, the partial sequence deviation types, that subsume previously studied types and lead to efficient EFR instances in games with moderate lengths.
In addition, we present a thorough empirical analysis of EFR instantiated with different deviation types in benchmark games, where we find that stronger types typically induce better performance. \end{abstract}

\section{Introduction}
We seek more effective algorithms for playing multi-player, general-sum extensive-form games (EFGs).
The hindsight rationality framework~\parencite{hsr2020} suggests a game playing approach that prescribes no-regret dynamics and describes jointly rational behavior with mediated equilibria~\parencite{aumann1974ce}.
Rationality within this framework is measured by regret in hindsight relative to strategy transformations, also called deviations, rather than as prospective optimality with respect to beliefs.
Each deviation transforms the learner's behavior into a competitor that the learner must surpass, so a richer set of deviations pushes the learner to perform better.

While larger deviation sets containing more sophisticated deviations produce stronger competitors, they also tend to raise computational and storage requirements.
For example, there is one external deviation (constant strategy transformation) for each strategy in a set of $n$, but there are $n^2$ internal deviations~\parencite{foster1999int-regret} that transform one particular strategy into another, and the latter is fundamentally stronger.
Though even achieving hindsight rationality with respect to external deviations appears intractable because the number of strategies in an EFG grows exponentially with the size of the game.

The counterfactual regret minimization (CFR)~\parencite{cfr} algorithm makes use of the EFG structure to be efficiently hindsight rational for external deviations.
Modifications to CFR by \textcite{celli2020noregret} and \textcite{hsr2020} are efficiently hindsight rational for other types of deviations as well.
\correctionStart{Mention observable sequential rationality.}
We generalize these algorithms as \emph{extensive-form regret minimization} (\emph{EFR}), a simple and extensible algorithm that is observably sequentially hindsight rational~\parencite{hsr2020} for any given deviation set where each deviation can be decomposed into action transformations at each decision point.\edlIntroCorrectionsFootnote{}It is generally intractable to run EFR with all of these \emph{behavioral deviations} so we identify four subsets that lead to efficient EFR instantiations that are observably sequentially hindsight rational for all previously studied tractable deviation types (external, causal~\parencite{forges2002efce,von2008efce-complexity,gordonGreenwaldMarks2008,Dudik09causal-dev,farina2020efcce}, action~\parencite{von2008efce-complexity,hsr2020}, and counterfactual~\parencite{hsr2020}) simultaneously.
\correctionEnd
We provide EFR instantiations and sublinear regret bounds for each of these new \emph{partial sequence deviation} types.

We present a thorough empirical analysis of EFR's performance with different deviation types in benchmark games from OpenSpiel~\parencite{LanctotEtAl2019OpenSpiel}.
Stronger deviation types typically lead to better performance, and EFR with the strongest type of partial sequence deviation often performs nearly as well as that with all behavioral deviations, in games where the latter is tractable.
 
\section{Background}
This work will continuously reference decision making from both the macroscopic, \emph{normal-form} view, and the microscopic, \emph{extensive-form} view.
We first describe the normal-form view, which models simultaneous decision making, before extending it with the extensive-form view, which models sequential decision making.

\subsection{The Normal-Form View}
At the macro-scale, players in a game choose \emph{strategies} that jointly determine the \emph{utility} for each player.
We assume a bounded utility function
$\utility_i : \TerminalHistories \to [-\maxReward, \maxReward]$
for each player $i$ on a finite set of outcomes, $\TerminalHistories$.
Each player has a finite set of \emph{pure strategies}, $\pureStrat_i \in \PureStratSet_i$, describing their decision space.
A set of results for entirely random events, \eg/, die rolls, is denoted $\PureStratSet_{\chance}$.
A \emph{pure strategy profile},
$\pureStrat \in \PureStratSet = \PureStratSet_{\chance} \times \bigtimes_{i=1}^N \PureStratSet_i$,
is an assignment of pure strategies to each player, and each strategy profile corresponds to a unique outcome
$z \in \TerminalHistories$
determined by the \emph{reach function} $\reachProb(z; \pureStrat) \in \set{0, 1}$.

A \emph{mixed strategy}, $\strat_i \in \StrategySet_i = \simplex^{\abs{\PureStratSet_i}}$, is a probability distribution over pure strategies.
In general, we assume that strategies are mixed where pure strategies are point masses.
The probability of a chance outcome, $\pureStrat_{\chance} \in \PureStratSet_{\chance}$, is determined by the ``chance player'' who plays the fixed strategy $\strat_{\chance}$.
A \emph{mixed strategy profile},
$\strat \in \StrategySet = \set{\strat_{\chance}} \times \bigtimes_{i=1}^N \StrategySet_i$,
is an assignment of mixed strategies to each player.
The probability of sampling a pure strategy profile, $\pureStrat$, is the product of sampling each pure strategy individually, \ie/,
$\strat(\pureStrat) = \strat_c(\pureStrat_c) \prod_{i = 1}^N \strat_i(\pureStrat_i)$.
For convenience, we denote the tuple of mixed strategies for all players except $i$ as $\strat_{-i} \in \StrategySet_{-i} = \set{\strat_{\chance}} \times \bigtimes_{j\neq i} \StrategySet_j$.
We overload the reach function to represent the probability of realizing outcome $z$ according to mixed profile $\strat$, \ie/,
$\reachProb(z; \strat) = \E_{\pureStrat \sim \strat}[\reachProb(z; \pureStrat)]$,
allowing us to express player $i$'s expected utility as
$\utility_i(\strat_i, \strat_{-i}) \as \utility_i(\strat)
  = \E_{z \sim \reachProb(\cdot; \strat)}[\utility_i(z)]$.

The \emph{regret} for playing strategy $\strat_i$ instead of deviating to an alternative strategy $\strat'_i$ is their difference in expected utility
$\utility_i(\strat'_i, \strat_{-i}) - \utility_i(\strat)$.
We construct alternative strategies by transforming $\strat_i$.
Let $\DevSet^{\SWAP}_{\mathcal{X}} = \{\dev : \mathcal{X} \to \mathcal{X}\}$ be the set of transformations to and from a given finite set $\mathcal{X}$.
The pure strategy transformations in $\DevSet^{\SWAP}_{\PureStratSet_i}$ are known as \emph{swap deviations}~\parencite{greenwald2003general}.
Given a mixed strategy $\strat_i$, the transformed mixed strategy under deviation
$\dev \in \DevSet^{\SWAP}_{\PureStratSet_i}$
is the pushforward measure of $\strat_i$, denoted as $\dev(\strat_i)$ and defined by
$[\dev\strat_i](\pureStrat'_i) = \sum_{\pureStrat_i \in \dev^{-1}(\pureStrat'_i)} \strat_i(\pureStrat_i)$ for all $\pureStrat'_i \in \PureStratSet_i$,
where $\dev^{-1} : \pureStrat'_i \mapsto \set{ \pureStrat_i \where \dev(\pureStrat_i) = \pureStrat_i' }$ is the pre-image of $\dev$.
The regret for playing strategy $\strat_i$ instead of deviating according to $\dev$ is then
$\regret(\dev; \strat) = \utility_i(\dev(\strat_i), \strat_{-i}) - \utility_i(\strat)$.

In an online learning setting, a learner repeatedly plays a game with unknown, dynamic, possibly adversarial players.
On each round $1 \le t \le T$, the learner who acts as player $i$ chooses a strategy, $\strat_i^t$, simultaneously with the other players who in aggregate choose $\strat^t_{-i}$.
The learner is evaluated on their strategies, $\tuple{\strat_i^t}_{t = 1}^T$, against a deviation, $\dev$, with the cumulative regret
$\regret^{1:T}(\dev) \as \sum_{t = 1}^T \regret(\dev; \strat^t)$.
A learner is rational in hindsight with respect to a set of deviations,
$\DevSet \subseteq \DevSet^{\SWAP}_{\PureStratSet_i}$,
if the maximum positive regret,
$\max_{\dev \in \DevSet} \big(\regret^{1:T}(\dev)\big)^+$ where $\cdot^+ = \max\set{\cdot, 0}$,
is zero.
A \emph{no-regret} or \emph{hindsight rational} algorithm ensures that average maximum positive regret vanishes as $T \to \infty$.

The \emph{empirical distribution of play},
$\recDist^T \in \simplex^{\abs{\PureStratSet}}$,
is the distribution that summarizes online correlated play, \ie/,
$\recDist^T(\pureProfile) = \frac{1}{T} \sum_{t = 1}^T \strat^t(\pureProfile)$,
for all pure strategy profiles, $\pureProfile$.
The distribution $\recDist^T$ can be viewed as a source of ``strategy recommendations'' distributed to players by a neutral ``mediator''.
The incentive for player $i$ to deviate from the mediator's recommendations, sampled from $\recDist^T$, to behavior chosen by $\dev$ is then player $i$'s average regret
$\E_{\pureStrat \sim \recDist^T}[ \regret(\dev; \pureStrat) ]
  = \frac{1}{T} \regret^{1:T}(\dev)$.
Jointly hindsight rational play converges toward a \emph{mediated equilibrium}~\parencite{aumann1974ce} where no player has an incentive to deviate from the mediator's recommendations, since hindsight rational players ensure that their average regret vanishes.

The deviation set influences what behaviors are considered rational and the difficulty of ensuring hindsight rationality.
For example, the \emph{external deviations}, $\DevSet^{\EXT}_{\PureStratSet_i} = \set{ \dev^{\to \pureStrat_i^{\TARGET}} : \pureStrat_i \mapsto \pureStrat_i^{\TARGET} }_{\pureStrat_i^{\TARGET} \in \PureStratSet_i}$, are the constant strategy transformations, a set which is generally limited in strategic power compared to the full set of swap deviations.
However, it is generally intractable to directly minimize regret with respect to the external deviations in sequential decision-making settings because the set of pure strategies grows exponentially with the number of decision points.

\subsection{The Extensive-Form View}

\textbf{Actions, histories, and information sets.}
An \emph{extensive-form game} (\emph{EFG}) models player behavior as a sequence of decisions.
Outcomes, here called \emph{terminal histories}, are constructed incrementally from the empty history, $\emptyHistory \in \Histories$.
At any history $h$, one player determined by the \emph{player function} $\playerChoice : \Histories \setminus \TerminalHistories \to \set{1, \dots, N} \cup \set{\chance}$ plays an \emph{action}, $a \in \Actions(h)$, from a finite set, which advances the current history to $ha$.
We write $h \sqsubset ha$ to denote that $h$ is a predecessor of $ha$.
We denote the maximum number of actions at any history as $n_{\Actions}$.

Histories are partitioned into \emph{information sets} to model imperfect information, \eg/, private cards.
The player to act in each history $h \in \infoSet$ in information set $\infoSet \in \InfoSets$ must do so knowing only that the current history is in $\infoSet$.
The unique information set that contains a given history is returned by $\infoSetOf$ (``blackboard I'') and
an arbitrary history of a given information set is returned by $\arbHistory$ (``blackboard h'').
Naturally, the action sets of each history within an information set must coincide, so we overload $\Actions(\infoSet) = \Actions(\arbHistory(\infoSet))$.

Each player $i$ has their own \emph{information partition}, denoted $\InfoSets_i$.
We restrict ourselves to \emph{perfect-recall} information partitions that ensure players never forget the information sets they encounter during play and their information set transition graphs are forests (not trees since other players may act first).
We write $\infoSet \prec \infoSet'$ to denote that $\infoSet$ is a predecessor of $\infoSet'$ and $\parent(\infoSet')$ to reference the unique parent (immediate predecessor) of $\infoSet'$.
Let $d_{\infoSet}$ be the number of $\infoSet$'s predecessors representing $\infoSet$'s depth and $d_* = \max_{\infoSet \in \InfoSets_i} d_{\infoSet}$ be the depth of player $i$'s deepest information set.
We use $a_h^{\to \infoSet'}$ or $a_{\infoSet}^{\to \infoSet'}$ to reference the unique action required to play from $h \in \infoSet$ to a successor history in $\infoSet' \succ \infoSet$.

\textbf{Strategies and reach probabilities.}
From the extensive-form view, a pure strategy is an assignment of actions to each of a player's information sets, \ie/, $\pureStrat_i(\infoSet)$ is the action that player $i$ plays in information set $\infoSet$ according to pure strategy $\pureStrat_i$.
A natural generalization is to randomize at each information set, leading to the notion of a \emph{behavioral strategy}~\parencite{Kuhn53}.
A behavioral strategy is defined by an assignment of \emph{immediate strategies}, $\strat_i(\infoSet) \in \simplex^{\abs{\Actions(\infoSet)}}$, to each of player $i$'s information sets, where $\strat_i(a \given \infoSet)$ is the probability that $i$ plays action $a$ in $\infoSet$.
Perfect recall ensures \emph{realization equivalence} between the set of mixed and behavioral strategies
where there is always a behavioral strategy that applies the same weight to each terminal history as a mixed strategy and \viceversa/.
Thus, we treat mixed and behavioral strategies (and by extension pure strategies) as interchangeable representations.

Since histories are action sequences and behavioral strategies define conditional action probabilities, the probability of reaching a history under a profile is the joint action probability that follows from the chain rule of probability.
We overload $\reachProb(h; \strat)$ to return the probability of a non-terminal history $h$.
Furthermore, we can look at the joint probability of actions played by just one player or a subset of players, denoted, for example, as $\reachProb(h; \strat_i)$ or $\reachProb(h; \strat_{-i})$.
We can use this and perfect recall to define the probability that player $i$ plays to their information set $\infoSet \in \InfoSets_i$ as
$\reachProb(\arbHistory(\infoSet); \strat_i)$.
Additionally, we can exclude actions taken before some initial history $h$ to get the probability of playing from $h$ to history $h'$, written as $\reachProb(h, h'; \cdot)$, where it is $1$ if $h = h'$ and $0$ if $h \not\sqsubseteq h'$.
 
\correctionStart{Move section up to here since we need counterfactual values to define observable sequential rationality. This move also required adding a bit more background on CFR and a couple of forward references.}
\subsection{Counterfactual Regret Minimization}
\label{sec:cfrBackground}
\emph{Counterfactual regret minimization} (\emph{CFR})~\parencite{cfr} is based on information-set-local hindsight rational learning, where each immediate strategy is evaluated based on \emph{counterfactual value}.
Given a strategy profile, $\strat$, the counterfactual value for taking $a$ in information set $\infoSet$ is the expected utility for player $i$, assuming they play to reach $\infoSet$ before playing $\strat_i$ thereafter and that the other players play according to $\strat_{-i}$ throughout, \ie/,
\[
  \cfv_{\infoSet}(a; \strat)
    = \sum_{\substack{
      h \in \infoSet,\\
      z \in \TerminalHistories}}
        \reachProb(h; \strat_{-i})
        \underbrace{\reachProb(ha, z; \strat) \utility_i(z)}_{\text{Future value given $ha$.}}.
\]
The counterfactual value allows each of CFR's local learners to be evaluated at each information set in relative isolation according to \emph{immediate counterfactual regret}, which is the extra counterfactual value achieved by choosing a given action instead of following $\strat_i$ at $\infoSet$, \ie/,
$\regret^{\CF}_{\infoSet}(a; \strat)
    = \cfv_{\infoSet}(a; \strat) - \E_{a' \sim \strat_i(\infoSet)} \cfv_{\infoSet}(a'; \strat)$.

CFR was originally shown to be hindsight rational for the external deviations~\parencite{cfr}, and is now better understood as achieving observable sequential hindsight rationality with respect to the counterfactual deviations~\parencite{hsr2020}.
Observable sequential rationality will be discussed in \cref{sec:osr} and counterfactual deviations will be defined in \cref{sec:behavDevs}.
\correctionEnd

\subsection{Extensive-Form Correlated Equilibrium}

\emph{Extensive-form correlated equilibrium} (\emph{EFCE}) is defined by Definition 2.2 of \textcite{von2008efce-complexity} as a mediated equilibrium with respect to deviations that are constructed according to the play of a \emph{deviation player}.
At the beginning of the game, the mediator samples a pure strategy profile (strategy recommendations), $\pureStrat$, and the game plays out according to this profile until it is player $i$'s turn to act.
Player $i$'s decision at this information set $\infoSet$ is determined by the deviation player who observes $\pureStrat_i(\infoSet)$, which is the action recommended to player $i$ by the mediator at $\infoSet$, and then chooses an action by either following this recommendation or deviating to a different action.
After choosing an action and waiting for the other players to move according to their recommended strategies, the deviation player arrives at $i$'s next information set.
Knowing the actions that were previously recommended to $i$, they again choose to follow the next recommendation or to deviate from it.
This process continues until the game ends.

The number of different states that the deviation player's memory could be in upon reaching information set $\infoSet$ at depth $d_{\infoSet}$ is $n_{\Actions}^{d_{\infoSet}}$ corresponding to the number of action combinations across $\infoSet$'s predecessors.
One way to avoid this exponential growth is to assume that recommended strategies are \emph{reduced}, that is, they do not assign actions to information sets that could not be reached according to actions assigned to previous information sets.
Thus, the action recommendation that the deviation player would normally observe after a previous deviation does not exist to observe.
This assumption effectively forces the deviation player to behave according to an \emph{informed causal deviation}~\parencite{gordonGreenwaldMarks2008,Dudik09causal-dev} defined by a ``trigger'' action and information set pair, along with a strategy to play after triggering, and the number of possible memory states grows linearly with depth.
Defining EFCE as a mediated equilibrium with respect to informed causal deviations allows them to be computed efficiently, which has led to this becoming the conventional definition of EFCE.

\section{Behavioral Deviations}
\label{sec:behavDevs}

Instead of achieving tractability by limiting the amount of information present in strategy recommendations, what if we intentionally hide information from the deviation player?
At each information set, $\infoSet$, we now provide the deviation player with three options: (i) follow the action recommendation at information set $\infoSet$, $\pureStrat_i(\infoSet)$, sight unseen, (ii) choose a new action without ever seeing $\pureStrat_i(\infoSet)$, or (iii) observe $\pureStrat_i(\infoSet)$ and then choose an action.

If
$\Actions_* = \bigcup_{\infoSet \in \InfoSets_i} \Actions(\infoSet)$
is the union of player $i$'s action sets, then we can describe the deviation player's memory,
$g \in G_i \subseteq (\set{*} \cup \Actions_*)^{d_*}$,
as a string that begins empty and gains a character after each of player $i$'s actions.
The recommendation, $\pureStrat_i(\infoSet)$, at information set $\infoSet$ where option one or three is chosen must be revealed to the deviation player, either as a consequence of play (option one) or as a prerequisite (option three), thus resulting in the next memory state $g\pureStrat_i(\infoSet)$.
Otherwise, the next memory state is formed by appending the ``$*$'' character to indicate that $\pureStrat_i(\infoSet)$ remains hidden.
Limiting the options available to the deviation player thus limits the number of memory states that they can realize.
Given a memory state $g$, there is only one realizable child memory state at the next information set if the deviation player is allowed either option one or two, or two memory states if both options one and two are allowed.
If all three options are allowed, the number of realizable child memory states at the next information set is equal to the number of actions at the current information set plus one.

Formally, these three options are executed at each information set $\infoSet$ with an \emph{action transformation}, $\dev_{\infoSet} : \Actions(\infoSet) \to \Actions(\infoSet)$, chosen from one of three sets: (i) the singleton containing the \emph{identity transformation}, $\set{\dev^1 : a \mapsto a}$, (ii) the external transformations,
$\DevSet_{\Actions(\infoSet)}^{\EXT}$,
or (iii) the \emph{internal transformations}~\parencite{foster1999int-regret}
\[\DevSet_{\Actions(\infoSet)}^{\IN}
  = \set*{
    \dev^{a^{\TRIGGER} \to a^{\TARGET}} : a \mapsto \begin{cases}
      a^{\TARGET} &\mbox{if } a = a^{\TRIGGER}\\
      a &\mbox{o.w.}
    \end{cases}
  }_{a^{\TRIGGER}, a^{\TARGET} \in \Actions(\infoSet)}.\]
While internal transformations can only swap one action with another, there is no loss in generality because every multi-action swap can be a represented as the combination of single swaps~\parencite{Dudik09causal-dev,greenwald2003general}.
Thus, any strategy sequence that can be improved upon by a swap deviation can also be improved upon by at least one internal deviation.

A complete assignment of action transformations to each information set and realizable
memory state
represents a complete strategy for the deviation player.
We call such an assignment a \emph{behavioral deviation} in analogy with behavioral strategies and denote them as $\DevSet^{\IN}_{\InfoSets_i}$ since the behavioral deviations are a natural analog of the internal transformations in EFGs.

All previously described EFG deviation types can be represented as sets of behavioral deviations:

\textbf{\textcite{von2008efce-complexity}'s deviations.}
Any strategy that \textcite{von2008efce-complexity}'s deviation player could employ\footnote{Where the deviation player makes only single-action swaps, which, again, is a simplification made without a loss in generality~\parencite{Dudik09causal-dev,greenwald2003general}.} is an assignment of internal transformations to every information set and memory state so the set of all such behavioral deviations represents all possible deviation player strategies.
A mediated equilibrium with respect to the behavioral deviations could thus perhaps be called a ``full strategy EFCE'', though ``behavioral correlated equilibrium'' may lead to less confusion with the conventional EFCE definition.

\textbf{Causal deviations.}
An informed causal deviation is defined by trigger information set $\infoSet^{\TRIGGER}$, trigger action $a^{\TRIGGER}$, and strategy $\strat'_i$. The following behavioral deviation reproduces any such deviation: assign (i) the internal transformation $\dev^{a^{\TRIGGER} \to a^{\TARGET}}$ to the sole memory state at $\infoSet^{\TRIGGER}$, (ii) external transformations to all successors $\infoSet' \succ \infoSet^{\TRIGGER}$ where $a^{\TRIGGER}$ is in the deviation player's memory to reproduce $\strat'_i$, and (iii) identity transformations to every other information set and memory state. The analogous \emph{blind causal deviation}~\parencite{farina2020efcce} always triggers in $\infoSet^{\TRIGGER}$, which is reproduced with the same behavioral deviation except that the external transformation $\dev^{\to \strat'_i(\infoSet^{\TRIGGER})}$ is assigned to $\infoSet^{\TRIGGER}$.

\textbf{Action deviations.}
An \emph{action deviation}~\parencite{von2008efce-complexity} modifies the immediate strategy at $\infoSet^{\TRIGGER}$, $\strat_i(\infoSet^{\TRIGGER})$, only, either conditioning on $\strat_i(\infoSet^{\TRIGGER})$ (an informed action deviation) or not (a blind action deviation~\parencite{hsr2020}), so any such deviation is reproduced by assigning either an internal or external transformation to the sole memory state at $\infoSet^{\TRIGGER}$, respectively, and identity transformations elsewhere.

\textbf{Counterfactual deviations.}
A \emph{counterfactual deviation}~\parencite{hsr2020} plays to reach a given ``target'' information set, $\infoSet^{\TARGET}$, and transforms the immediate strategy there so any such deviation is reproduced by assigning (i) external transformations to all of the information sets leading up to $\infoSet^{\TARGET}$, (ii) an external or internal transformation to the sole memory state at $\infoSet^{\TARGET}$ (for the blind and informed variant, respectively), and (iii) identity transformations elsewhere.

Phrasing these deviation types as behavioral deviations allows us to identify complexity differences between these deviation types by counting the number of realizable memory states they admit.
Across all action or counterfactual deviations, there is always exactly one memory state at each information set to which a non-identity transformation is assigned.
Thus, a hindsight rational algorithm need only ensure its strategy cannot be improved by applying a single action transformation at each information set.
Under the causal deviations, in contrast, the number of memory states realizable at information set $\infoSet$ is at least the number of $\infoSet$'s predecessors
since there is at least one causal deviation that triggers at each of them and plays to $\infoSet$.
This makes causal deviations more costly to compete with and gives them strategic power, though notably not enough to subsume either the action or counterfactual deviations~\parencite{hsr2020}.
Are there sets of behavioral deviations that subsume the causal, action, and counterfactual deviations without being much more costly than the causal deviations?
 
\correctionStart{Add section to discuss observable sequential rationality.}
\section{Observable Sequential Rationality}
\label{sec:osr}

The approach that this work takes to navigate tradeoffs between strategic power and computational efficiency is to use observable sequential rationality (OSR)~\parencite{hsr2020} to elevate the strength of simple deviation types.

We overload counterfactual value to apply to an entire strategy, not just an action, as
\[
  \cfv^{\CF}_{\infoSet}(\strat_i; \strat_{-i})
    = \sum_{\substack{
      h \in \infoSet,\\
      z \in \TerminalHistories}}
        \reachProb(h; \strat_{-i})
        \underbrace{\reachProb(h, z; \strat) \utility_i(z)}_{\text{Future expected value of $\strat_i$.}}.
\]
and define the reach-probability-weighted counterfactual value as
$\cfv_{\infoSet}(\strat_i; \strat_{-i})
  = \reachProb\subex*{\arbHistory(\infoSet); \strat_i} \cfv^{\CF}_{\infoSet}(\strat_i; \strat_{-i})$.
From this value function, we get a general notion of full regret, which can be used to define OSR.
\begin{definition}
  \label{def:generalFullRegret}
  Define the \emph{full regret} of deviation $\dev \in \DevSet^{\SWAP}_{\PureStratSet_i}$ under strategy profile $\strat$ from information set $\infoSet$ as the deviation reach-probability-weighted difference in counterfactual value
  $\regret_{\infoSet}(\dev; \strat)
    = \cfv_{\infoSet}\subex*{\dev(\strat_i); \strat_{-i}} - \cfv_{\infoSet}\subex*{\dev_{\prec \infoSet}(\strat_i); \strat_{-i}}$,
  where $\dev_{\prec \infoSet}$ is the deviation that applies $\dev$ only before $\infoSet$, \ie/,
  $[\dev_{\prec \infoSet} \pureStrat](\bar{\infoSet}) = [\dev \pureStrat](\bar{\infoSet})$
  if $\bar{\infoSet} \prec \infoSet$ and $\pureStrat(\bar{\infoSet})$ otherwise.
\end{definition}
\begin{definition}
  \label{def:obs-seq-rationality}
  A recommendation distribution,
  $\recDist \in \simplex^{\abs{\PureStratSet}}$,
  is OSR for player $i$ with respect to a set of deviations,
  $\DevSet \subseteq \DevSet^{\SWAP}_{\PureStratSet_i}$,
  if the maximum benefit for every deviation,
  $\dev \in \DevSet$,
  according to the reach-probability-weighted counterfactual value from every information set,
  $\infoSet \in \InfoSets_i$,
  is non-positive,
  \begin{align*}
    \E_{\pureStrat \sim \recDist}\subblock*{
      \regret_{\infoSet}(\dev; \pureStrat)
    }
    \le 0.
  \end{align*}
\end{definition}
The hindsight analogue to \cref{def:obs-seq-rationality} follows.
\begin{definition}
  \label{def:obsSeqHindsightRationality}
  Player $i$ is \emph{observably sequentially} (\emph{OS}) \emph{hindsight rational} for deviation set $\DevSet \subseteq \DevSet^{\SWAP}_{\PureStratSet}$ if they choose a strategy $\strat_i^t$ on each round $t$ so that their average full regret vanishes at each information set $\infoSet$, \ie/,
  $\lim_{T \to \infty}
      \frac{1}{T} \sum_{t = 1}^T \regret_{\infoSet}(\dev; \strat^t) \le 0$,
  with respect to the behavior of the other players, $\tuple{ \strat^t_{-i} }_{t = 1}^T$.
  The positive part of their maximum average full regret across information sets is their \emph{OSR gap}.
\end{definition}

To show how the strength of a deviation type can be elevated with OSR, we need to define the ``single-target'' deviations derived from a given deviation set.
\begin{definition}
  The set of \emph{single-target deviations} generated from an arbitrary set of deviations $\DevSet \subseteq \DevSet^{\SWAP}_{\PureStratSet_i}$ is
  \[\DevSet_{\preceq \TARGET}
    = \set{
        \dev' \where
          \forall \pureStrat_i, \,
          [\dev' \pureStrat_i](\bar{\infoSet}) = [\dev \pureStrat_i](\bar{\infoSet}) \mbox{ if } \bar{\infoSet} \preceq \infoSet
          \mbox{ and } \pureStrat_i(\bar{\infoSet}) \mbox{ o.w.}
      }_{\dev \in \DevSet, \, \infoSet \in \InfoSets_i}.
  \]
  $\DevSet_{\preceq \TARGET}$ is the set of deviations constructed from $\DevSet$ that only deviate along a single path of information sets up to a ``target'' information set and behave identically to the input strategy at all other information sets.
  For example, the counterfactual deviations are the single-target deviations constructed from the external deviations so we can denote the set of counterfactual deviations as $\DevSet_{\PureStratSet_i, \preceq \TARGET}^{\EXT}$.
\end{definition}

The set of single-target deviations is special because it captures all of the ways that the input strategy could be modified along any sequence of information sets without including the combinations of these modifications across multiple sequences.
The set of single-target deviations can therefore be much smaller than its generating set while preserving the generating set's capacity to express different behavior modifications.

The next result formalizes the intuition that the set of single-target deviations preserves the essential expressive capacity of its generating set by proving that there is no beneficial deviation in an arbitrary set of deviations $\DevSet$ if OSR is achieved with respect to $\DevSet_{\preceq \TARGET}$.
This result is a generalization of Theorem 3 from \textcite{hsr2020arxiv} for all deviation sets.
\begin{theorem}
  \label{thm:singleTargetDevElevationWithOsr}
  If the full regret for player $i$ at each information set $\infoSet$ for the sequence of $T$ strategy profiles,
$\tuple{ \strat^t }_{t = 1}^T$,
with respect to each single-target deviation
$\dev \in \DevSet_{\preceq \TARGET}$,
is
$d_{\infoSet}(\dev)f(T) \ge 0$,
where
$d_{\infoSet}(\dev)$
is the number of non-identity action transformations
$\dev$
applies from $\infoSet$ to the end of the game, then $i$'s OSR gap with respect to
$\DevSet_{\preceq \TARGET}$
and
$\DevSet \subseteq \DevSet^{\SWAP}_{\PureStratSet_i}$
is no more than
$\abs{\InfoSets_i} f(T)$. \end{theorem}
\correctionEnd

\section{Partial Sequence Deviations}

Notice that the causal, action, and counterfactual deviations are composed of contiguous blocks of the same type of action transformation.
We can therefore understand these deviations as having distinct phases.
The \emph{correlation phase} is an initial sequence of identity transformations, where ``correlation'' references the fact that the identity transformation preserves any correlation that player $i$'s behavior has with those of the other players.
There are causal and action deviations with a correlation phase, but no counterfactual deviation exhibits such behavior.
All of these deviation types permit a \emph{de-correlation phase} that modifies the input strategy with external transformations, breaking correlation.
Finally, the \emph{re-correlation phase} is where identity transformations follow a de-correlation phase, but it is only present in action and counterfactual deviations.
The informed variant of each deviation type separates these phases with a single internal transformation, which both modifies the strategy and preserves correlation.
The action deviation type is the only one that permits all three phases, but the de-correlation phase is limited to a single action transformation.

Why not permit all three phases at arbitrary lengths to subsume the causal, action, and counterfactual deviations?
We now introduce four types of \emph{partial sequence deviations} based on exactly this idea, where each phase spans a ``partial sequence'' through the game.

The \emph{blind partial sequence} (\emph{BPS}) deviation has all three phases and lacks any internal transformations.
\correctionStart{Correct discussion.}
Just as there is always a blind counterfactual deviations that reproduces the behavior of an external deviation along a single path of information sets, there is always a BPS deviation that reproduces the behavior of a blind causal deviation along a single path of information sets.
More precisely, the set of BPS deviations is the set of single-target deviations generated from the set of blind causal deviations.
Just as there are exponentially fewer counterfactual deviations than external deviations, there are exponentially fewer BPS deviations than blind causal deviations.
There are
$d_* n_{\Actions} \abs{\InfoSets_i}$
BPS deviations compared with
$\bigO{n_{\Actions}^{\abs{\InfoSets_i}}\abs{\InfoSets_i}}$
blind causal deviations.
Combined with OSR via \cref{thm:singleTargetDevElevationWithOsr}, the BPS deviations capture the same strategic power with an exponential reduction in complexity.
Even better, the set of BPS deviations includes the sets of blind action and blind counterfactual deviations.
The empirical distribution of play of learners that are hindsight rational for BPS deviations thus converge towards what we could call a \emph{BPS correlated equilibrium}.
An OS BPS correlated equilibrium (OS-BPSCE) is in the intersection of the OS versions of three equilibrium sets: extensive-form coarse-correlated equilibrium (EFCCE)~\parencite{farina2020efcce}, agent-form coarse-correlated equilibrium (AFCCE)~\parencite{hsr2020}, and counterfactual coarse-correlated equilibrium (CFCCE)~\parencite{hsr2020}.
\correctionEnd

\correctionStart{Add a note to connect back to single-target deviation terminology.}
In general, re-correlation is strategically useful~\parencite{hsr2020} and adding it to a deviation type (transforming it into a single-target deviation type) \emph{decreases} its complexity!
\correctionEnd
\correctionStart{Reword and remove for consistency with above changes.}
While this observation may be new in its generality, \textcite{cfr} implicitly uses this property of deviations in EFGs and specifically the fact that the set of blind counterfactual deviations is the set of single-target deviations generated from the set of external deviations.
\correctionEnd

There are three versions of informed partial sequence deviations due to the asymmetry between informed causal and informed counterfactual deviations.
A \emph{causal partial sequence} (\emph{CSPS}) deviation uses an internal transformation at the end of the correlation phase while a \emph{counterfactual partial sequence} (\emph{CFPS}) deviation uses an internal transformation at the start of the re-correlation phase.
A \emph{twice informed partial sequence} (\emph{TIPS}) deviation uses internal transformations at both positions, making it the strongest of our partial sequence deviation types.
\correctionStart{Add qualifications.}
The set of CSPS deviations is the set of single-target deviations generated from the set of informed causal deviations, and therefore subsumes the informed causal deviations, when used with OSR, while being exponentially smaller.
TIPS achieves our initial goal as it subsumes the informed causal, informed action, and informed counterfactual deviations at the cost of an $n_{\Actions}$ factor compared to CSPS or CFPS, when used with OSR.
Each type of informed partial sequence deviation corresponds to a new equilibrium concept and a new OS equilibrium concept in the intersection of previously studied equilibrium concepts.
\correctionEnd

\correctionStart{Fix diagram.}
\def\efrExtraSection/{E}
Table \efrExtraSection/.1 in Appendix \efrExtraSection/ gives a formal definition of each deviation type derived from behavioral deviations and \cref{fig:dev-diagram} gives a visualization of each type along with their relationships.
\correctionEnd
\correctionStart{Add a qualification to the behavioral deviation label in referenced table and remove now unnecessary footnote.}
The number of deviations contained within each deviation type is listed in \cref{tab:numRegrets}.
\correctionEnd

\tikzstyle{myAction} = [very thick]
\tikzstyle{mySeqPath} = [seqPath, very thick]

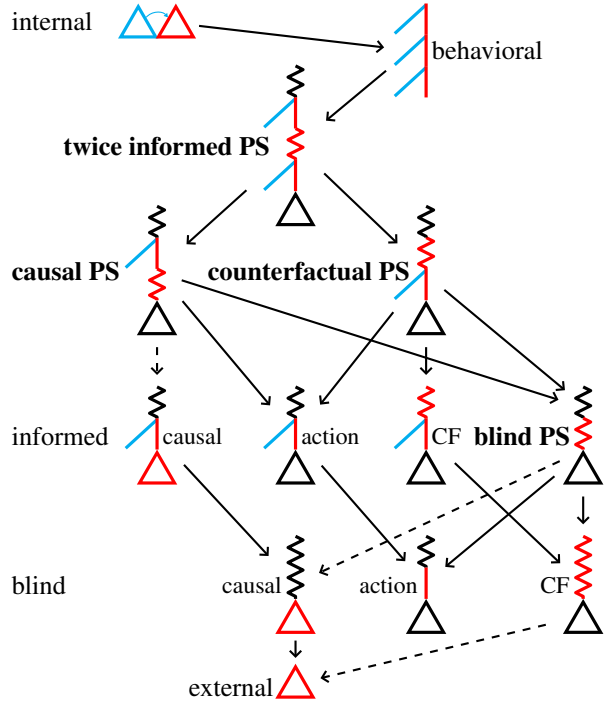
\begin{figure}[tb]
\centering
\begin{tikzpicture}[inner sep=0]
  \tikzstyle{arrow} = [thick, >=Straight Barb, -{>[scale=0.7]}]
  \tikzstyle{column 3} = [column sep=3em]
  \tikzstyle{osrArrow} = [arrow, dashed]

\matrix(fig) [column sep=0, row sep=\cmToEmFactor*0.5em] {
    \coordinate(inRoot);
    \inDevDiagram{inRoot};
    &
    &
    \coordinate(bhvRoot);
    \behavDevDiagram{bhvRoot};
    & \\[-2.5em]

    &
\coordinate(tipsRoot);
    \tipsDevDiagram{tipsRoot};
    &
    & \\[-2em]

\coordinate(cspsRoot);
    \cspsDevDiagram{cspsRoot};
    &
    &
\coordinate(cfpsRoot);
    \cfpsDevDiagram{cfpsRoot};
    & \\[\cmToEmFactor*0.3em]

\coordinate(iCauseRoot);
    \inCsDevDiagram{iCauseRoot};
    \node(label)
      [anchor=west, xshift=0.2em]
      at (iCauseRoot/center)
      {\small causal};
    &

\coordinate(iActRoot);
    \inActDevDiagram{iActRoot};
    \node(label)
      [anchor=west, xshift=0.2em]
      at (iActRoot/center)
      {\small action};
    &

\coordinate(iCfRoot);
    \inCfDevDiagram{iCfRoot};
    \node(label)
      [anchor=west, xshift=0.2em]
      at (iCfRoot/center)
      {\small CF};
    &
    \coordinate(bpsRoot);
    \bpsDevDiagram{bpsRoot}; \\[\cmToEmFactor*0.3em]

    &
\coordinate(bCauseRoot);
    \node(bCauseRoot/devTree) [tree, dev, anchor=north] at ($(bCauseRoot.south)+(0,-\seqLen-\actLen)$) {};
    \draw[mySeqPath, follow] (bCauseRoot) -- (bCauseRoot/devTree.north);

\coordinate(bCauseRoot/south) at (bCauseRoot/devTree.south);
    \coordinate(bCauseRoot/center) at ($(bCauseRoot)!0.5!(bCauseRoot/south)$);

    \node(label)
      [anchor=east, xshift=-0.5em]
      at (bCauseRoot/center)
      {\small causal};
    &
\coordinate(bActRoot);
    \coordinate(bActRoot/seq)
      at ($(bActRoot.north)+(0,-\seqLen)$);
    \draw[mySeqPath, follow] (bActRoot.north) -- (bActRoot/seq.north);
    \node(bActRoot/followTree)
      [tree, follow, anchor=north]
      at ($(bActRoot/seq.south)+(0,-\actLen)$) {};
    \draw[myAction, dev] (bActRoot/seq.north) -- (bActRoot/followTree.north);

\coordinate(bActRoot/south) at (bActRoot/followTree.south);
    \coordinate(bActRoot/center) at ($(bActRoot)!0.5!(bActRoot/south)$);

    \node(label)
      [anchor=east, xshift=-0.3em]
      at (bActRoot/center)
      {\small action};
    &
    \coordinate(bCfRoot);
    \node(bCfRoot/followTree)
      [tree, follow, anchor=north]
      at ($(bCfRoot.south)+(0,-\seqLen-\actLen)$) {};
    \draw[mySeqPath, dev] (bCfRoot.north) -- (bCfRoot/followTree.north);

\coordinate(bCfRoot/south) at (bCfRoot/followTree.south);
    \coordinate(bCfRoot/center) at ($(bCfRoot)!0.5!(bCfRoot/south)$);

    \node(label)
      [anchor=east, xshift=-0.5em]
      at (bCfRoot/center)
      {\small CF}; \\

    &
    \node(exRoot) [tree, dev] {};
    & \\
  };
  \node(inLabel)
    [anchor=east, xshift=-1.2em]
    at (inRoot/west)
    {internal};
  \node(bhvLabel)
    [anchor=west, xshift=0.2em]
    at (bhvRoot/east)
    {behavioral};
  \node(informedLabel -| inLabel.west)
    [anchor=west]
    at (iCauseRoot/center -| inLabel.west)
    {informed};
  \node(blindLabel)
    [anchor=west]
    at (bCauseRoot/center -| inLabel.west)
    {blind};
  \node(exLabel)
    [xshift=-0.3em, anchor=east]
    at (exRoot.west)
    {external};

  \node(tipsLabel)
    [anchor=east, xshift=-\cmToEmFactor*0.4em]
    at (tipsRoot/center)
    {\textbxf{twice informed PS}};
  \node(cspsLabel)
    [anchor=west]
    at (cspsRoot/center -| inLabel.west)
    {\textbxf{causal PS}};
  \node(cfpsLabel)
    [anchor=east, xshift=-0.6em]
    at (cfpsRoot/center)
    {\textbxf{counterfactual PS}};
  \node(bpsLabel)
    [anchor=east, xshift=-\cmToEmFactor*0.2em]
    at (bpsRoot/center)
    {\textbxf{blind PS}};

  \draw[arrow]
    ($(inRoot/center)!1.6em!(bhvRoot/center)$)
    --
    ($(bhvRoot/center)!1em!(inRoot/center)$);
  \draw[arrow]
    ($(bhvRoot/center)!1.2em!(tipsRoot/center)$)
    --
    ($(tipsRoot/center)!1.5em!(bhvRoot/center)$);
  \draw[arrow]
    ($(tipsRoot/center)!2.5em!(cspsRoot/center)$)
    --
    ($(cspsRoot/center)!1.5em!(tipsRoot/center)$);
  \draw[arrow]
    ($(tipsRoot/center)!1.5em!(cfpsRoot/center)$)
    --
    ($(cfpsRoot/center)!1.5em!(tipsRoot/center)$);

  \draw[osrArrow]
    ($(cspsRoot/center)!2.3em+\seqLen/2!(iCauseRoot/center)$)
    --
    ($(iCauseRoot/center)!2.3em!(cspsRoot/center)$);
  \draw[arrow]
    ($(cfpsRoot/center)!2.3em+\seqLen/2!(iCfRoot/center)$)
    --
    ($(iCfRoot/center)!2.3em!(cfpsRoot/center)$);

  \draw[arrow]
    ($(cspsRoot/center)!1.5em!(iActRoot/center)$)
    --
    ($(iActRoot/center)!1.5em!(cspsRoot/center)$);
  \draw[arrow]
    ($(cspsRoot/center)!1em!(bpsRoot/center)$)
    --
    ($(bpsRoot/center)!1em!(cspsRoot/center)+(0, 1em)$);

  \draw[arrow]
    ($(cfpsRoot/center)!2em!(iActRoot/center)$)
    --
    ($(iActRoot/center)!1.5em!(cfpsRoot/center)$);
  \draw[arrow]
    ($(cfpsRoot/center)!1em!(bpsRoot/center)$)
    --
    ($(bpsRoot/center)!1em!(cfpsRoot/center)+(0, 1em)$);

  \draw[arrow]
    ($(iCauseRoot/center)!1.5em!(bCauseRoot/center)$)
    --
    ($(bCauseRoot/center)!1.5em!(iCauseRoot/center)$);

  \draw[arrow]
    ($(iActRoot/center)!1.5em!(bActRoot/center)$)
    --
    ($(bActRoot/center)!1.2em!(iActRoot/center)$);
  \draw[arrow]
    ($(iCfRoot/center)!1.5em!(bCfRoot/center)$)
    --
    ($(bCfRoot/center)!1.2em!(iCfRoot/center)$);

  \draw[osrArrow]
    ($(bpsRoot/center)!1em!(bCauseRoot/center)-(0, 0.5em)$)
    --
    ($(bCauseRoot/center)!1em!(bpsRoot/center)$);
  \draw[arrow]
    ($(bpsRoot/center)!1.5em!(bActRoot/center)-(0, 0.5em)$)
    --
    ($(bActRoot/center)!1em!(bpsRoot/center)$);

  \draw[arrow]
    ($(bpsRoot/center)!2.3em!(bCfRoot/center)$)
    --
    ($(bCfRoot/center)!2.3em!(bpsRoot/center)$);

  \draw[arrow]
    ($(bCauseRoot/center)!2.1em!(exRoot)$)
    --
    ($(exRoot)!1.1em!(bCauseRoot/center)$);
  \draw[osrArrow]
    ($(bCfRoot/center)!1.5em!(exRoot)-(0, 1em)$)
    --
    ($(exRoot)!1em!(bCfRoot/center)$);
\end{tikzpicture}
    \caption{A summary of the deviation landscape in EFGs.
    Each pictogram is an abstract representation of a prototypical deviation.
    Games play out from top to bottom.
    Straight lines represent action transformations, zigzags are transformation sequences, and triangles are transformations of entire decision trees.
    Identity transformations are colored black; internal transformations have a cyan component representing the trigger action or strategy and a red component representing the target action or strategy; and external transformations only have a red component.
    Arrows denote ordering from a stronger to a weaker deviation type (and therefore a subset to superset equilibrium relationship), the dashed arrow denotes that this relationship holds only under OSR.}
  \label{fig:dev-diagram}
\end{figure}
 \def\numActions{n_{\Actions}}
\def\DevSetInMinusI{\DevSet^{\INT}_{\Actions(\infoSet)} \setminus \set{\dev^1}}
\def\DevSetI{\DevSet^{\INT}_{\Actions(\infoSet)}}
\def\DevSetE{\DevSet^{\EXT}_{\Actions(\infoSet)}}
\def\numDevInMinusI{\numActions^2 - \numActions}
\def\lightrule{\specialrule{0.0001pt}{1mm}{1mm}}
\begin{table}[tb]
  \centering
  \begin{threeparttable}
  \caption{A rough accounting of (i) realizable memory states, (ii) action transformations, and (iii) the total number of deviations showing dominant terms.
    Columns (i) and (ii) are with respect to a single information set.}
  \label{tab:numRegrets}
  \small
  \begin{tabularx}{\linewidth}{Xlll}
    \toprule
    type
      & \stretchBox{(i)}
      & \stretchBox{(ii)}
      & \stretchBox{(iii)}\\
    \midrule
    internal
      & N/A
      & N/A
      & $\numActions^{2 \abs{\InfoSets_i}}$\\[0.3em]
    single-target behavioral
      & $\numActions^{d_*}$
      & $\numActions^2$
      & $\numActions^{d_* + 2} \abs{\InfoSets_i}$\\
    \lightrule
    TIPS
      & $d_* \numActions$
      & $\numActions^2$
      & $d_* \numActions^3 \abs{\InfoSets_i}$\\[0.3em]
    CSPS
      & $d_* \numActions$
      & $\numActions$\tnote{$\dagger$}
      & $d_* \numActions^2 \abs{\InfoSets_i}$\\[0.3em]
    CFPS
      & $d_*$
      & $\numActions^2$
      & $d_* \numActions^2 \abs{\InfoSets_i}$\\[0.3em]
    BPS
      & $d_*$
      & $\numActions$
      & $d_* \numActions \abs{\InfoSets_i}$\\
    \lightrule
    informed causal
      & $d_*$
      & N/A
      & $\numActions^{\abs{\InfoSets_i} + 1} \abs{\InfoSets_i}$\\[0.3em]
    informed action
      & $1$
      & $\numActions^2$
      & $\numActions^2 \abs{\InfoSets_i}$\\[0.3em]
    informed CF
      & $1$
      & $\numActions^2$
      & $\numActions^2 \abs{\InfoSets_i}$\\
    \lightrule
    blind causal
      & $d_*$
      & N/A
      & $\numActions^{\abs{\InfoSets_i}} \abs{\InfoSets_i}$\\[0.3em]
    blind action
      & $1$
      & $\numActions$
      & $\numActions \abs{\InfoSets_i}$\\[0.3em]
    blind CF
      & $1$
      & $\numActions$
      & $\numActions \abs{\InfoSets_i}$\\
    \lightrule
    external
      & N/A
      & N/A
      & $\numActions^{\abs{\InfoSets_i}}$\\
\bottomrule
  \end{tabularx}
  \begin{tablenotes}
    \item[$\dagger$] One memory state at each information set is associated with the set of internal transformations which contains $\bigO{\numActions^2}$ transformations, but this is dominated by the number of external transformations associated with every other memory state in non-root information sets.
  \end{tablenotes}
  \end{threeparttable}
\end{table}

\section{Extensive-Form Regret Minimization}
\label{sec:efr}
\correctionStart{Add OS qualification. The previous statement is still true but this statement is stronger.}
We now develop \emph{extensive-form regret minimization} (\emph{EFR}), a general and extensible algorithm that is OS hindsight rational for any given set of behavioral deviations.
\correctionEnd
Its computational requirements and regret bound scale closely with the number of realizable memory states.

\correctionStart{Move content up and add discussion of upcoming material.}
CFR variants for different deviation types have already been derived by modifying the counterfactual values used to train CFR's information-set-local learners.
\textcite{celli2020noregret} define \emph{laminar subtree trigger regret} (\emph{immediate trigger regret} in our terminology), as the regret under counterfactual values weighted by the probability that player $i$ plays to a given predecessor and plays a particular action there.
Their \emph{ICFR} modification of \emph{pure CFR}~\parencite{gibson2014phd}\footnote{Pure CFR purifies the learner's strategy on each round by sampling actions at each information set.} is hindsight rational for informed causal deviations\footnote{Actually, ICFR is OS hindsight rational for CSPS deviations as well, but of course this was previously not understood.}.
\textcite{hsr2020} also observe that simply weighting the counterfactual regret at $\infoSet$ by the probability that player $i$ plays to $\infoSet$ modifies CFR so that it is hindsight rational for blind action deviations.

We derive EFR by generalizing these ideas to weighted counterfactual regrets with dynamic weights determined by deviation reach probabilities on each round.
Before we can construct EFR, we need to discuss regret minimization with time selection functions from the normal-form perspective.
Along the way, we will derive a new algorithm, time selection regret matching, to minimize regret in this scenario, which will allow us to use regret matching in EFR in much the same way as it is typically used in CFR.
\correctionEnd

\subsection{Time Selection}
The key insight that leads to EFR is that each of the deviation player's memory states corresponds to a different weighting function, which reduces the problem of minimizing immediate regret with respect to all weightings simultaneously to \emph{time selection regret minimization}~\parencite{blum2007time-selection}.
In a time selection problem, there is a finite set of $M(\dev)$ time selection functions,
$\TSelectionSet(\dev) = \set{t \mapsto \tSelectionFn^t_j \in [0, 1]}_{j = 1}^{M(\dev)}$,
for each deviation $\dev \in \DevSet \subseteq \DevSet^{\SWAP}_{\PureStratSet_i}$ that maps the round $t$ to a weight.
The regret with respect to deviation $\dev$ and time selection function $\tSelectionFn \in \TSelectionSet(\dev)$ after $T$ rounds is
$\regret^{1:T}(\dev, \tSelectionFn) \as \sum_{t = 1}^T \tSelectionFn^t \regret(\dev; \strat^t)$.
The goal is to ensure that each of these regrets grow sublinearly, which can be accomplished by simply treating each $(\dev, \tSelectionFn)$-pair as a separate transformation (here called an \emph{expert}) and applying a no-regret algorithm\footnote{\
  The \emph{regret matching++} algorithm~\parencite{kash2019combining} could ostensibly be used to minimize regret with respect to all time selection functions simultaneously without using more than $\abs{\PureStratSet_i}$ computation and memory, however, there is an error in the proof of the regret bound.
  In Appendix F, we give an example where regret matching++ suffers linear regret and we show that no algorithm can have a sublinear bound on the sum of positive instantaneous regrets.}.

We introduce a $(\DevSet, f)$-regret matching~\parencite{Hart00,greenwald2006bounds} algorithm for the time selection setting with a regret bound that depends on the size of the largest time selection function set, $M^* = \max_{\dev \in \DevSet} M(\dev)$.\footnote{While we only present the bound for the \emph{rectified linear unit} (\emph{ReLU}) link function, $\cdot^+: x \mapsto \max\set{0, x}$, the arguments involved in proving \cref{thm:rm-for-heterogeneous-ts} apply to any link function; only the final bound would change.}
\begin{corollary}
  \label{thm:rm-for-heterogeneous-ts}
  Given deviation set $\DevSet \subseteq \DevSet^{\SWAP}_{\PureStratSet_i}$ and finite time selection sets
$\TSelectionSet(\dev) = \set{\tSelectionFn_j \in [0, 1]^T}_{j = 1}^{M(\dev)}$ for each deviation $\dev \in \DevSet$,
$(\DevSet, \cdot^+)$-regret matching chooses a strategy on each round $1 \le t \le T$ as the fixed point of
$\rmOperator^t: \strat_i \mapsto
      \nicefrac{1}{z^t}
      \sum_{\dev \in \DevSet} \dev(\strat_i) y^t_{\dev}$
or an arbitrary strategy when $z^t = 0$, where
\emph{link outputs} are generated from exact regrets
$y^t_{\dev}
  = \sum_{\tSelectionFn \in \TSelectionSet(\dev)}
    \tSelectionFn^t (
      \regret^{1:t - 1}(\dev, \tSelectionFn)
    )^+$
and $z^t = \sum_{\dev \in \DevSet} y^t_{\dev}$.
This algorithm ensures that
$\regret^{1:T}(\dev, \tSelectionFn)
  \le
    2 \maxReward \sqrt{
      M^* \maxActivation(\DevSet) T
    }$
for any deviation $\dev$ and time selection function $\tSelectionFn$, where
$\maxActivation(\DevSet) = \max_{\pureStrat_i \in \PureStratSet_i} \sum_{\dev \in \DevSet} \ind{\dev(\pureStrat_i) \ne \pureStrat_i}$
is the maximal activation of $\DevSet$~\parencite{greenwald2006bounds}.
 \end{corollary}
\def\rmExtraSection/{D}
This result is a consequence of two more general theorems presented in Appendix \rmExtraSection/, one that allows regret approximations \ala/ \textcite{dorazio2019frcfr} (motivating the use of function approximation) and another that allows predictions of future regret, \ie/, optimistic regret matching~\parencite{dorazioOptLH}.
Appendix \rmExtraSection/ also contains analogous results for the \emph{\rmPlus/}~\parencite{cfrPlus,Tammelin15CFRPlus} modification of regret matching.
 
\subsection{Memory Probabilities}
Just as we use the reach probability function to capture the frequency that a mixed strategy plays to reach a particular history, we define a \emph{memory probability function},
$\tSelectionFn_{\dev}$,
to capture the frequency that the deviation player, playing behavioral deviation $\dev$, reaches information set $\infoSet$ with memory state $g$ given mixed recommendations, $\strat_i$.
It is the product of the probabilities that $\strat_i$ plays each action in $g$, \ie/,
$\tSelectionFn_{\dev}(\infoSet, \emptyHistory; \strat_i) = 1$,
$\tSelectionFn_{\dev}(\infoSet', ga; \strat_i) = \tSelectionFn_{\dev}(\infoSet, g; \strat_i) \strat_i(a \given \infoSet)$,
and
$\tSelectionFn_{\dev}(\infoSet', g*; \strat_i) = \tSelectionFn_{\dev}(\infoSet, g; \strat_i)$, for all $\infoSet' \succ \infoSet$.
Under pure recommendations, the memory probability function expresses realizability.
We overload
\[G_i(\infoSet, \dev)
  = \set{g \in G_i \where \exists \pureStrat_i \in \PureStratSet_i, \tSelectionFn_{\dev}(\infoSet, g; \pureStrat_i) = 1}\]
as the set of memory states that $\dev$ can realize in $\infoSet$.

\subsection{EFR}
We define the immediate regret of behavioral deviation $\dev$ at information set $\infoSet$ and memory state $g$ as the immediate counterfactual regret for not applying action transformation $\dev_{\infoSet, g}$ weighted by the probability of $g$, \ie/,
$\tSelectionFn_{\dev}(\infoSet, g; \strat_i)
  \regret^{\CF}_{\infoSet}(\dev_{\infoSet, g}; \strat)$,
where we generalize counterfactual regret to action transformations as
$\regret^{\CF}_{\infoSet}(\dev_{\infoSet, g}; \strat)
  = \E_{a \sim \dev_{\infoSet, g}(\strat_i(\infoSet))}\subblock{ \regret^{\CF}_{\infoSet}(a; \strat) }$.
By treating
$t \mapsto \tSelectionFn_{\dev}(\infoSet, g; \strat^t_i)$
for each memory state $g$ in $\infoSet$ as a time selection function,
we reduce the problem of minimizing immediate regret to time selection regret minimization.

EFR is given a set of behavioral deviations, $\DevSet \subseteq \DevSet^{\IN}_{\InfoSets_i}$, and gathers all transformations at information set $\infoSet$ across realizable memory states into
$\DevSet_{\infoSet}
  = \set{\dev_{\infoSet, g}}_{
      \dev \in \DevSet, \,
      g \in G_i(\infoSet, \dev)}.$
Each action transformation, $\dev_{\infoSet} \in \DevSet_{\infoSet}$ is associated with the set of time selection functions
\[\TSelectionSet_{\infoSet}^{\DevSet}(\dev_{\infoSet}) = \set{
  t \mapsto
    \tSelectionFn_{\dev'}(\infoSet, g; \strat^t_i)
}_{\dev' \in \DevSet, \, g \in G_i(\infoSet, \dev'), \, \dev'_{\infoSet, g} = \dev_{\infoSet}}\]
derived from memory probabilities.
EFR then chooses its immediate strategy at $\infoSet$ according to a time selection regret minimizer.
Applying the same procedure at each information set, EFR minimizes immediate regret at all information sets and memory states simultaneously.

Hindsight rationality requires us to relate immediate regret to full regret.
The full regret of behavioral deviation $\dev$ at information set $\infoSet$ and memory state $g$, $\regret_{\infoSet, g}(\dev; \strat)$, is the expected value achieved by $\dev(\strat_i)$ from $\infoSet$ and $g$ minus that of $\strat_i$, weighted by the probability of $g$.
The full regret at the start of the game on any given round is then exactly the total performance difference between $\dev$ and the learner.
The full regret decomposes across successive information sets and memory states, \ie/,
\[\regret_{\infoSet, g}(\dev; \strat)
  =
    \underbrace{\regret_{\infoSet}(\dev_{\preceq \infoSet, \sqsubseteq g})}_{\text{Immediate regret.}}
    + \sum_{\substack{
        a' \in \Actions(\infoSet),\\
        \infoSet' \in \InfoSets_i(\infoSet, a'),\\
        b \in \set{*} \cup \Actions(\infoSet)}}
      \underbrace{\regret_{\infoSet', gb}(\dev; \strat),}_{\text{Full regret at successor.}}
\]
where
$\dev_{\preceq \infoSet, \sqsubseteq g}$
is the behavioral deviation that deploys $\dev$ at all $\bar{\infoSet} \preceq \infoSet$ and $\bar{g} \sqsubseteq g$ but the identity transformation otherwise.
Therefore, minimizing immediate regret at every information set and memory state also minimizes full regret at every information set and memory state, including those at the start of the game.
\correctionStart{Add OS qualification. The previous statement is still true but this statement is stronger.}
Finally, this implies that minimizing immediate regret with respect to any given set of behavioral deviations $\DevSet \subseteq \DevSet^{\IN}_{\InfoSets_i}$ ensures OS hindsight rationality with respect to $\DevSet$.
\correctionEnd
EFR's regret is bounded according to the following theorem:
\begin{theorem}
  \label{lem:efr}
  Instantiate EFR for player $i$ with exact regret matching and a set of behavioral deviations $\DevSet \subseteq \DevSet_{\InfoSets_i}^{\IN}$.
Let the maximum number of information sets along the same line of play where non-identity internal transformations are allowed before a non-identity transformation within any single deviation be $n_{\IN}$.
Let $D = \max_{\infoSet \in \InfoSets_i, \dev_{\infoSet} \in \DevSet_{\infoSet}} \abs{\TSelectionSet_{\infoSet}^{\DevSet}(\dev_{\infoSet})} \maxActivation(\DevSet_{\infoSet})$.
\correctionStart{Add a mention about how the full regret at each information set has this bound. The previous statement is still true but this statement is stronger.}
Then, EFR's cumulative full regret at each information set after $T$ rounds with respect to $\DevSet$ and the set of single-target deviations generated from $\DevSet$, $\DevSet_{\preceq \TARGET}$, is upper bounded by
$2^{n_{\IN} + 1} \maxReward \abs{\InfoSets_i} \sqrt{D T}$.
In addition, this implies that EFR is OS hindsight rational with respect to $\DevSet \cup \DevSet_{\preceq \TARGET}$.
\correctionEnd \end{theorem}
See Appendix \efrExtraSection/ for technical details.

\subsection{Discussion}
The variable $D$ in the EFR regret bound depends on the given behavioral deviations and is essentially the maximum number of realizable memory states times the number of action transformations across information sets.
See Table \efrExtraSection/.2 in Appendix \efrExtraSection/ for the $D$ value for each deviation type.

\cref{alg:efr-desc} provides an implementation of EFR with exact regret matching.
\begin{algorithm}[tb]
  \caption{EFR update for player $i$ with exact regret matching.}
  \label{alg:efr-desc}

  \begin{algorithmic}[1]
    \Input\ Strategy profile, $\strat^t \in \StrategySet$, $t \ge 1$, and \\
      \Indent behavioral deviations, $\DevSet \subseteq \DevSet^{\BHV}_{\InfoSets_i}$.
    \State \textbf{initialize} table $\regret^{1:0}_{\cdot, \cdot}(\cdot) = 0$.

    \LComment{Update cumulative immediate regrets:}
    \For{$\infoSet \in \InfoSets_i$,
        $\dev_{\infoSet} \in \DevSet_{\infoSet}$,
        $\tSelectionFn \in \TSelectionSet_{\infoSet}^{\DevSet}(\dev_{\infoSet})$}
      \State $\regret^{1:t}_{\infoSet, \tSelectionFn}(\dev_{\infoSet}) \gets
        \regret^{1:t-1}_{\infoSet, \tSelectionFn}(\dev_{\infoSet})
          + \tSelectionFn^t
            \regret^{\CF}_{\infoSet}(\dev_{\infoSet}; \strat^t)$
      \label{line:imm-phi-regret}
    \EndFor
    \LComment{Construct $\strat_i^{t + 1}$ with regret matching:}
    \For{$\infoSet \in \InfoSets_i$ from the start of the game to the end}
      \For{$\dev_{\infoSet} \in \DevSet_{\infoSet}$}
        \LComment{$\strat^{t + 1}_i$ need only be defined at $\bar{\infoSet} \prec \infoSet$.}
        \State $y^{t + 1}_{\dev_{\infoSet}} \gets \sum_{\tSelectionFn \in \TSelectionSet_{\infoSet}^{\DevSet}(\dev_{\infoSet})}
          \tSelectionFn^{t + 1}
            \big(\regret_{\infoSet, \tSelectionFn}^{1:t}(\dev_{\infoSet})\big)^+$
      \EndFor
      \State $z^{t + 1} \gets \sum_{\dev_{\infoSet} \in \DevSet_{\infoSet}} y^{t + 1}_{\dev_{\infoSet}}$
      \label{line:linkWeightSum}
      \If{$z^{t + 1} > 0$}
        \State $\strat^{t + 1}_i(\infoSet) \gets$ a fixed point of linear operator\\
          \Indent
            $\rmOperator^t :
              \simplex^{\abs{\Actions(\infoSet)}} \ni \immStrat
              \mapsto
              \frac{1}{z^{t + 1}}
                \sum\limits_{\dev_{\infoSet} \in \DevSet_{\infoSet}}
                  y^{t + 1}_{\dev_{\infoSet}} \dev_{\infoSet}(\immStrat)$
        \label{line:fixedPoint}
      \Else
        \State $\subblock{\strat^{t + 1}_i(a \given \infoSet) \gets \frac{1}{\Actions(\infoSet)}}_{a \in \Actions(\infoSet)}$ \hfill \CommentChar{} Arbitrary.
      \EndIf
    \EndFor
    \Return $\strat^{t + 1}_i$
  \end{algorithmic}
\end{algorithm}
 Notice that as a matter of practical implementation, EFR only requires $\DevSet_{\infoSet}$ and $\TSelectionSet_{\infoSet}^{\DevSet}$ for all information sets $\infoSet \in \InfoSets_i$, which are often easier to specify than $\DevSet$.
Table \efrExtraSection/.2 in Appendix \efrExtraSection/ shows the $\DevSet_{\infoSet}$ and $\TSelectionSet_{\infoSet}^{\DevSet}$ parameters corresponding to each deviation type, as well as the $D$ and $n_{\IN}$ values that determine each EFR instance's regret bound.
\correctionStart{Reuse single-target deviation terminology to improve clarity.}
Thanks to this feature, EFR always operates on the set of single-target deviations generated from $\DevSet$ that are additionally augmented with re-correlation.
\correctionEnd
This both potentially improves EFR's performance and ensures that learning is efficient even for some exponentially large deviation sets, like the external, blind causal, and informed causal deviations.

For example, it is equivalent to instantiate EFR with the blind causal deviations or the BPS deviations.
Likewise for the informed causal deviations and the CSPS deviations, where EFR reduces to a variation of ICFR~\parencite{celli2020noregret}.
To be precise, ICFR is pure EFR (analogous to pure CFR) instantiated with the CSPS deviations except that the external and internal action transformation learners at separate memory states within an information set are sampled and updated independently in ICFR.
EFR therefore improves on this algorithm (beyond its generality) because EFR's action transformation learners share all experience, potentially leading to faster learning, and EFR enjoys a deterministic finite time regret bound.

Crucially, EFR's generality does not come at a computational cost.
EFR reduces to the CFR algorithms previously described to handle counterfactual and action deviations~\parencite{cfr,hsr2020}.
Furthermore, EFR inherits CFR's flexibility as it can be used with Monte Carlo sampling~\parencite{mccfr,avgStratSampling,Gibson12probing,pcs}, function approximation~\parencite{waugh2015solving,morrill2016,dorazio2019frcfr,deepCFR,steinberger2020dream,dorazio2020}, variance reduction~\parencite{schmid2019variance,davis2020lowVarBaselines}, and predictions~\parencite{rakhlin2013optimization,farina2019stable,dorazioOptLH,farina2020faster}.
 
\section{Experiments}
\label{sec:experiments}

Our theoretical results show that EFR variants utilizing more powerful deviation types are pushed to accumulate higher payoffs during learning in worst-case environments.
Do these deviation types make a practical difference outside of the worst case?

We investigate the performance of EFR with different deviation types in nine benchmark game instances from \emph{OpenSpiel}~\parencite{LanctotEtAl2019OpenSpiel}.
We evaluate each EFR variant by the expected payoffs accumulated over the course of playing each game in each seat over 1000 rounds under two different regimes for selecting the other players.
In the ``fixed regime'', other players play their parts of the fixed sequence of strategy profiles generated with self-play before the start of the experiment using one of the EFR variants under evaluation.
In the ``simultaneous regime'', the other players are EFR instances themselves.
In games with more than two players, all other players share the same EFR variant and we only record the score for the solo EFR instance.
The fixed regime provides a test of how well each EFR variant adapts when the other players are gradually changing in an oblivious way where comparison is simple, while the simultaneous regime is a possibly more realistic test of dynamic adaptation where it is more difficult to draw definitive conclusions about relative effectiveness.

Since we evaluate expected payoff, use expected EFR updates, and use exact regret matching, all results are deterministic and hyperparameter-free.
To compute the regret matching fixed point when internal transformations are used, we solve a linear system with the Jacobi singular value algorithm implemented by the \texttt{jacobiSvd} method from the Eigen C++ library~\parencite{guennebaud2010eigen}.
Experimental data and code for generating both the data and final results are available on
\href{https://github.com/dmorrill10/hr_edl_experiments}{GitHub}.\footnote{\url{https://github.com/dmorrill10/hr_edl_experiments}}
Experiments took roughly 20 hours to complete on a 2.10GHz Intel\textregistered{} Xeon\textregistered{} CPU E5-2683 v4 processor with 10 GB of RAM.

\def\experimentsExtraSection/{G}
Appendix \experimentsExtraSection/
hosts the full set of results but a representative summary from two variants of imperfect information goofspiel~\parencite{Ross71Goofspiel,lanctot13phdthesis} (a two-player and a three-player version denoted as \gTwoFive/ and \gThreeFour/, respectively, both zero-sum) and Sheriff (two-player, non-zero-sum) is presented in \cref{tab:representativeResults}.
See Appendix \experimentsExtraSection/.1 for descriptions of all games.

\begin{table}[tb]
  \centering
  \begin{threeparttable}
  \caption{The payoff of each EFR instance averaged across both $1000$ rounds and each instance pairing (eight pairs in total) in two-player and three-player goofspiel (measured in win frequency between zero and one), and Sheriff (measured in points between $-6$ and $+6$).
    The top group of algorithms use weak deviation types (\textsc{act}\textsubscript{\INT} $\to$ informed action deviations, \textsc{cf} $\to$ blind counterfactual, and \textsc{cf}\textsubscript{\INT} $\to$ informed counterfactual) and the middle group use partial sequence deviation types.
    The \textsc{bhv} instance uses the full set of behavioral deviations.}
  \label{tab:representativeResults}
  \small
  \begin{tabularx}{\linewidth}{Xlcccccc}
    \toprule
    & \multicolumn{3}{c}{fixed} & \multicolumn{3}{c}{simultaneous}\\
    & \gTwoFive/ & \gThreeFour/ & Sheriff & \gTwoFive/ & \gThreeFour/\tnote{$\dagger$} & Sheriff\\
    \midrule
    \textsc{act}\textsubscript{\INT}              & 0.51                 & 0.48  & 0.28              & 0.45                 & 0.86  & 0.00\\
    \textsc{cf}                                   & 0.56                 & 0.51  & 0.48              & 0.50                 & 0.88  & 0.34\\
    \textsc{cf}\textsubscript{\INT}               & 0.57                 & 0.51  & 0.60              & 0.50                 & 0.92  & 0.37\\
    \midrule
    \textsc{bps}                                  & 0.58                 & 0.51  & 0.58              & 0.50                 & 0.85  & 0.34\\
    \textsc{cf}                                   & 0.58                 & 0.52  & 0.70              & 0.51                 & 0.84  & 0.37\\
    \textsc{csps}                                 & 0.59                 & 0.52  & 0.61              & 0.51                 & 0.91  & 0.37\\
    \textsc{tips}                                 & 0.60                 & 0.53  & 0.82              & 0.51                 & 0.87  & 0.38\\
    \midrule
    \textsc{bhv}                                  & 0.63                 & 0.53  & 0.91              & 0.51                 & 0.92  & 0.38\\
    \bottomrule
  \end{tabularx}
  \begin{tablenotes}
    \item[$\dagger$] In three-player goofspiel, players who tend to play the same actions perform worse.
    Since the game is symmetric across player seats, two players who use the same (deterministic) algorithm will always employ the same strategies and often play the same actions, giving the third player a substantial advantage.
    The win percentage for all variants in the simultaneous regime tends to be high because we only record the score for each variant when they are instantiated in a single seat.
    The relative comparison is still informative.
  \end{tablenotes}
  \end{threeparttable}
\end{table}

Stronger deviations consistently lead to better performance in both the fixed and the simultaneous regime.
The behavioral deviations (BHV) and the informed action deviations (ACT\textsubscript{\INT}) often lead to the best and worst performance, respectively, and this is true of each scenario in \cref{tab:representativeResults}.
In many cases however, TIPS or CSPS yield similar performance to BHV.
A notable outlier from the scenarios in \cref{tab:representativeResults} is three-player goofspiel with a descending point deck.
Here, blind counterfactual (CF) and BPS deviations lead to better performance in the first few rounds before all variants quickly converge to play that achieves essentially the same payoff (see Figures \experimentsExtraSection/.1-\experimentsExtraSection/.4 in Appendix \experimentsExtraSection/).

\section{Conclusions}

We introduced EFR, an algorithm that is hindsight rational for any given set of behavioral deviations. While the full set of behavioral deviations leads to generally intractable computational requirements, we identified four partial sequence deviation types that are both tractable and powerful in games with moderate lengths.

An important tradeoff within EFR is that using stronger deviation types generally leads to slower strategy updates, demonstrated by Figures \experimentsExtraSection/.5-\experimentsExtraSection/.6 in Appendix \experimentsExtraSection/ where learning curves are plotted according to runtime.
Often in a tournament setting, the number of rounds and computational budget may be fixed so that running faster cannot lead to more reward for the learner, but there may be reward to gain by running faster in other scenarios.
Quantifying the potential benefit of using a stronger deviation type in particular games could aid in navigating this tradeoff.

Alternatively, perhaps the learner can navigate this tradeoff on their own.
Algorithms like the fixed-share forecaster~\parencite{herbster1998fixedShare} or context tree weighting~\parencite{willems1993ctw} efficiently minimize regret across large structured sets of experts, effectively avoiding a similar tradeoff.
This approach could also address a second tradeoff, which is that stronger deviation types lead to EFR regret bounds with larger constant factors even if the best deviation is part of a ``simpler'' class, \eg/, the regret bound that TIPS EFR has with respect to counterfactual deviations is larger than that of CFR even though a TIPS EFR instance might often accumulate more reward in order to compete with the larger TIPS deviations.
Perhaps an EFR variant can be designed that would compete with large sets of behavioral deviations, but its regret bound would scale with the ``complexity'' (in a sense that has yet to be rigorously defined) of the best deviation rather than the size of the whole deviation set.
Ideally, its computational cost would be independent of the deviation set size or would at least scale with the complexity of the best deviation.

\section*{Acknowledgements}
Computation provided by WestGrid and Compute Canada.
Dustin Morrill, Michael Bowling, and James Wright are supported by the Alberta Machine Intelligence Institute (Amii) and NSERC.
Michael Bowling and James Wright hold Canada CIFAR AI Chairs at Amii.
Amy Greenwald is supported in part by NSF Award CMMI-1761546.
Thanks to Ian Gemp, Gabriele Farina, and anonymous reviewers for constructive comments and suggestions.
\correctionStart{Add acknowledgement.}
Thanks to Revan MacQueen for showing that the partial sequence deviations do not subsume the causal or external deviations without OSR.
\correctionEnd

\bibliography{references}
\bibliographystyle{icml2021}

\appendix
\renewcommand\thefigure{\thesection.\arabic{figure}}
\setcounter{figure}{0}
\renewcommand\thetable{\thesection.\arabic{table}}
\setcounter{table}{0}
\section{Notation and Symbols}

\nomenclature{$\simplex^d$}{The $d$-dimensional probability simplex.}
\nomenclature{$h \in \Histories$}{An action history.}
\nomenclature{$z \in \TerminalHistories \subseteq \Histories$}{A terminal history.}
\nomenclature{$\maxReward$}{The maximum magnitude of any payoff.}
\nomenclature{$\utility_i$}{Bounded utility function for player $i$.}
\nomenclature{$\pureStrat \in \PureStratSet$}{A pure strategy profile.}
\nomenclature{$\PureStratSet_{\chance}$}{The set of a game's random events or the set of pure strategies that could be assigned to chance.}
\nomenclature{$\reachProb$}{The reach probability function.}
\nomenclature{$\strat \in \StrategySet$}{A behavioral/mixed strategy profile.}
\nomenclature{$\dev \in \DevSet^{\cdot}_{\mathcal{X}}$}{A transformation from finite set $\mathcal{X}$ to itself. $\SWAP$ in the superscript of $\DevSet^{\cdot}_{\mathcal{X}}$ denotes the set of swap transformations, $\EXT$ denotes the external transformations, and $\INT$ denotes the internal transformations.}
\nomenclature{$\dev^1 : a \mapsto a$}{The identity transformation.}
\nomenclature{$\dev^{\to a}$}{The external transformation to $a$.}
\nomenclature{$\dev^{a \to a'}$}{The internal transformation from $a$ to $a'$.}
\nomenclature{$\regret$}{Regret.}
\nomenclature{$\regret^{\CF}$}{Counterfactual regret.}
\nomenclature{$\recDist \in \simplex^{\abs{\PureStratSet}}$}{A distribution over strategy profiles, often representing an empirical distribution of play and interpreted as a recommendation distribution.}
\nomenclature{$\playerChoice$}{The player choice function.}
\nomenclature{$\infoSet \in \InfoSets_i$}{One of player $i$'s information sets.}
\nomenclature{$\infoSetOf(h)$}{The information set containing history $h$.}
\nomenclature{$\arbHistory(\infoSet)$}{An arbitrary history in information set $\infoSet$.}
\nomenclature{$a \in \Actions(h) = \Actions(\infoSetOf(h))$}{An action from the set of legal actions at history $h$ in information set $\infoSetOf(h)$.}
\nomenclature{$n_{\Actions}$}{The maximum number of actions available at any history.}
\nomenclature{$\parent(\infoSet')$}{The unique parent (immediate predecessor) of information set $\infoSet'$.}
\nomenclature{$d_{\infoSet}$}{The depth of information set $\infoSet$.}
\nomenclature{$d_*$}{The depth of player $i$'s deepest information set.}
\nomenclature{$a_{h}^{\to \infoSet'}$ or $a_{\infoSetOf(h)}^{\to \infoSet'}$}{The unique action that would need to be taken in history $h$ or information set $\infoSet(h)$ to reach successor information set $\infoSet' \succ \infoSet(h)$.}
\nomenclature{$\Actions_*$}{The union of player $i$'s action sets.}
\nomenclature{$g \in G_i$}{A deviation player memory state string.}
\nomenclature{$\DevSet^{\INT}_{\InfoSets_i}$}{The set of player $i$'s behavioral deviations.}
\nomenclature{$\cfv_{\infoSet}$}{The counterfactual value function at information set $\infoSet$.}
\nomenclature{$\tSelectionFn \in \TSelectionSet(\dev)$}{A time selection function associated with transformation $\dev$.}
\nomenclature{$\TSelectionSet_{\infoSet}^{\DevSet}(\dev_{\infoSet})$}{The set of time selection functions associated with action transformation $\dev_{\infoSet}$ corresponding to the deviation player memory probabilities generated by the set of behavioral deviations $\DevSet$ in information set $\infoSet$.}
\nomenclature{$M(\dev)$}{The size of the time selection function set associated with transformation $\dev$.}
\nomenclature{$M^*$}{The size of the largest time selection function set.}
\printnomenclature

\correctionStart{Move extra EFG background from the EFR section up to here to share with the new OSR section.}
\section{Additional Extensive-Form Game Background}

\textbf{The parent action function.}
The action taken to reach a given information set from its parent is returned by $\parentAction : \infoSet' \mapsto a_{\parent(\infoSet')}^{\to \infoSet'}$ (``blackboard a'').

\textbf{Terminal successor histories.}
Let the histories that terminate without further input from player $i$ after taking action $a$ in $\infoSet$ be
\begin{align}
  \TerminalHistories_i(\infoSet, a) = \set*{
    z \in \TerminalHistories
    \, \Bigg| \,
    \begin{aligned}
      &\exists h \in \infoSet, \,
      z \sqsupseteq ha, \,\\
      &\nexists h' \in \Histories_i, ha \sqsubseteq h' \sqsubset z
    \end{aligned}
  }.
\end{align}

\textbf{Child information sets.}
Let the child information sets of information set $\infoSet$ after taking action $a$ be
\begin{align}
  \InfoSets_i(\infoSet, a) = \set*{
    \infoSet' \in \InfoSets_i
    \, \Bigg| \,
    \begin{aligned}
      &\forall h' \in \infoSet', \,
      \exists h \in \infoSet, \,
      h' \sqsupseteq ha, \,\\
      &\nexists h'' \in \Histories_i, ha \sqsubseteq h'' \sqsubset h'
    \end{aligned}
  }.
\end{align}

\textbf{Terminal payoffs.}
Let the counterfactual immediate payoff of action $a$ in information set $\infoSet$ to player $i$ be
$\termValue(\infoSet, a; \strat_{-i}) = \sum_{z \in \TerminalHistories_i(\infoSet, a)} \reachProb(z; \strat_{-i}) \utility_i(z)$
and overload its expected value as
$\termValue(\infoSet; \strat) = \E_{A \sim \strat_i(\infoSet)} \subblock*{ \termValue(\infoSet, A; \strat_{-i}) }$.

\correctionEnd

\correctionStart{Add proof of OSR strength elevation.}
\section{Observable Sequential Rationality}

The proof of \cref{thm:singleTargetDevElevationWithOsr} requires a general full regret decomposition.

\begin{lemma}
  The reach-probability-weighted counterfactual value of player $i$'s strategy, $\strat_i$, from information set $\infoSet$, under perfect recall, recursively decomposes as
  \begin{align*}
    \cfv_{\infoSet}(\strat_i; \strat_{-i})
      &=
        \reachProb\subex*{\arbHistory(\infoSet); \strat_i} \termValue(\infoSet; \strat)
        + \hspace{-2em} \sum_{
          \infoSet' \in \bigcup_{a \in \Actions(\infoSet)}
            \InfoSets_i(\infoSet, a)
        } \hspace{-2em}
          \cfv_{\infoSet'}(\strat_i; \strat_{-i}).
  \end{align*}
\end{lemma}
\begin{proof}
  Assuming reach-probability weight of $1$ for $\infoSet$
  ($\reachProb\subex*{\arbHistory(\infoSet); \strat_i}$),
  the counterfactual value decomposes as
  \begin{align}
    &\cfv^{\CF}_{\infoSet}(\strat_i; \strat_{-i})\nonumber\\
      &= \sum_{\substack{
        h \in \infoSet,\\
        z \in \TerminalHistories
      }} \reachProb(h; \strat_{-i}) \reachProb(h, z; \strat) \utility_i(z)\\
      &= \sum_{\substack{
        h \in \infoSet,\\
        z \in \TerminalHistories
      }} \reachProb(h, z; \strat_i) \underbrace{\reachProb(z; \strat_{-i}) \utility_i(z)}_{\text{Terminal counterfactual values.}}\\
      &= \underbrace{
        \E_{A \sim \strat_i(\infoSet)}
          \subblock*{
            \sum_{z \in \TerminalHistories_i(\infoSet, A)}
              \reachProb(z; \strat_{-i}) \utility_i(z)
          }
        }_{\text{Expected value from terminal histories.}}
        + \underbrace{
          \E_{A \sim \strat_i(\infoSet)}\subblock*{
            \sum_{\substack{
              h' \in \infoSet' \in \InfoSets_i(\infoSet, A)\\
              z \in \TerminalHistories
            }}
              \reachProb(h', z; \strat_i) \reachProb(z; \strat_{-i}) \utility_i(z)
          }
        }_{\text{Expected value from non-terminal histories.}}\\
      &= \termValue(\infoSet; \strat)
        + \E_{A \sim \strat_i(\infoSet)}\subblock*{
          \sum_{\infoSet' \in \InfoSets_i(\infoSet, A)}
            \underbrace{
              \sum_{a' \in \Actions(\infoSet')}
                \strat_i(a' \given \infoSet')
                \sum_{\substack{
                  h' \in \infoSet'\\
                  z \in \TerminalHistories
                }}
                  \reachProb(h'a', z; \strat_i) \reachProb(z; \strat_{-i}) \utility_i(z).
            }_{\cfv^{\CF}_{\infoSet'}(\strat_i; \strat_{-i})}
        }\\
      &= \underbrace{\termValue(\infoSet; \strat)}_{\text{Expected immediate value.}}
        + \underbrace{
          \E_{A \sim \strat_i(\infoSet)}\subblock*{
            \sum_{\infoSet' \in \InfoSets_i(\infoSet, A)}
              \cfv^{\CF}_{\infoSet'}(\strat_i; \strat_{-i})
          }
        }_{\text{Expected future value.}}.
  \end{align}

  Accounting for strategies $\strat_i$ with reach-probability weights that are not $1$, we simply multiply by the reach weight.
  \begin{align}
    &\cfv_{\infoSet}(\strat_i; \strat_{-i})\nonumber\\
      &= \reachProb\subex*{\arbHistory(\infoSet); \strat_i} \termValue(\infoSet; \strat)
        + \sum_{a \in \Actions(\infoSet)}
          \reachProb\subex*{\arbHistory(\infoSet); \strat_i} \strat_i(a \given \infoSet)
          \sum_{\infoSet' \in \InfoSets_i(\infoSet, a)}
            \cfv^{\CF}_{\infoSet'}(\strat_i; \strat_{-i})\\
      &= \reachProb\subex*{\arbHistory(\infoSet); \strat_i} \termValue(\infoSet; \strat)
        + \hspace{-2em} \sum_{
          \infoSet' \in \bigcup_{a \in \Actions(\infoSet)}
            \InfoSets_i(\infoSet, a)
        } \hspace{-2em}
          \cfv_{\infoSet'}(\strat_i; \strat_{-i}),
    \label{eq:rwerDecomposition}
  \end{align}
  which completes the proof.
\end{proof}
\begin{lemma}
  \label{lem:regretDecomposition}
  The full regret of $\dev \in \DevSet^{\SWAP}_{\PureStratSet_i}$, under strategy profile $\strat$ and perfect recall, from information set $\infoSet$, recursively decomposes as
  \begin{align*}
    \regret_{\infoSet}(\dev; \strat)
      &=
        \regret_{\infoSet}(\dev_{\preceq \infoSet}; \strat)
        + \hspace{-2.2em} \sum_{
          \infoSet' \in \bigcup_{a \in \Actions(\infoSet)}
            \InfoSets_i(\infoSet, a)
        } \hspace{-2.2em}
          \regret_{\infoSet'}(\dev; \strat).
  \end{align*}
\end{lemma}
\begin{proof}
  The result follows from simple algebra and the decomposition of the reach-probability-weighted counterfactual value.
  Let $\dev_{\preceq \infoSet}$ be the truncated deviation that only applies $\dev$ at predecessors $\bar{\infoSet} \preceq \infoSet$.
  \begin{align*}
    &\regret_{\infoSet}(\dev; \strat)\nonumber\\
      &=
        \cfv_{\infoSet}(\dev(\strat_i); \strat_{-i})
        \rlap{$\overbrace{\phantom{- \cfv_{\infoSet}(\dev_{\preceq \infoSet}(\strat_i); \strat_{-i})+ \cfv_{\infoSet}(\dev_{\preceq \infoSet}(\strat_i); \strat_{-i})\,}}^0$}
        - \cfv_{\infoSet}(\dev_{\preceq \infoSet}(\strat_i); \strat_{-i})
        + \underbrace{
          \cfv_{\infoSet}(\dev_{\preceq \infoSet}(\strat_i); \strat_{-i})
          - \cfv_{\infoSet}(\dev_{\prec \infoSet}(\strat_i); \strat_{-i})
        }_{\regret_{\infoSet}(\dev_{\preceq \infoSet}; \strat)}\\
      &=
        \regret_{\infoSet}(\dev_{\preceq \infoSet}; \strat)\\
        &\quad+ \underbrace{
          \reachProb\subex*{\arbHistory(\infoSet); \dev(\strat_i)} \termValue(\infoSet; \strat)
          - \reachProb\subex*{\arbHistory(\infoSet); \dev_{\preceq \infoSet}(\strat_i)}
            \termValue(\infoSet; \strat)
        }_{0}\\
        &\quad+ \hspace{-2em} \sum_{
          \infoSet' \in \bigcup_{a \in \Actions(\infoSet)}
            \InfoSets_i(\infoSet, a)
        } \hspace{-1em}\underbrace{
          \cfv_{\infoSet'}(\dev(\strat_i); \strat_{-i})
          - \cfv_{\infoSet'}(\dev_{\preceq \infoSet}(\strat_i); \strat_{-i})
        }_{\regret_{\infoSet'}(\dev; \strat)}\\
      &=
        \regret_{\infoSet}(\dev_{\preceq \infoSet}; \strat)
        + \hspace{-2.2em} \sum_{
          \infoSet' \in \bigcup_{a \in \Actions(\infoSet)}
            \InfoSets_i(\infoSet, a)
        } \hspace{-2.2em}
          \regret_{\infoSet'}(\dev; \strat).\qedhere
  \end{align*}
\end{proof}
\setcounter{theorem}{0}
\begin{theorem}
  If the full regret for player $i$ at each information set $\infoSet$ for the sequence of $T$ strategy profiles,
$\tuple{ \strat^t }_{t = 1}^T$,
with respect to each single-target deviation
$\dev \in \DevSet_{\preceq \TARGET}$,
is
$d_{\infoSet}(\dev)f(T) \ge 0$,
where
$d_{\infoSet}(\dev)$
is the number of non-identity action transformations
$\dev$
applies from $\infoSet$ to the end of the game, then $i$'s OSR gap with respect to
$\DevSet_{\preceq \TARGET}$
and
$\DevSet \subseteq \DevSet^{\SWAP}_{\PureStratSet_i}$
is no more than
$\abs{\InfoSets_i} f(T)$. \end{theorem}
\begin{proof}
  This proof has an identical inductive structure to the proof of Theorem 3 by \textcite{hsr2020arxiv}, the only difference is that the pair of deviation sets involved are generalized and the full regrets that we work with are weighted appropriately for any deviation, not just counterfactual or external deviations.

  First, establish a simple fact about single-target deviations.
  At each information set $\infoSet$, the full regret with respect to each single-target deviation $\dev_{\preceq \infoSet}$ that transforms the actions up to and at $\infoSet$, and then immediately re-correlates is at most $f(T)$ since $d_{\infoSet}(\dev_{\preceq \infoSet}) = 1$.
  Formally, denote this value as
  \begin{align}
    \regret_{\infoSet}^{1:T}(\dev_{\preceq \infoSet})
      = \sum_{t = 1}^T
        \regret_{\infoSet}(\dev_{\preceq \infoSet}; \strat^t)
      \le f(T).
    \label{eq:immediateSingleTargetRegretBound}
  \end{align}

  Next, consider the terminal or height $1$ information sets for player $i$, \ie/, those without successors.
  The maximum full regret with respect to $\DevSet_{\preceq \TARGET}$ is the positive part of the maximum full regret with respect to $\DevSet$ at each terminal information set $\infoSet$; either the single-target deviation can change the action at $\infoSet$ or it can re-correlate.
  Therefore, the theorem is proved if all information sets are terminal as $d < \abs{\InfoSets_i}$.
  This serves as the base case of a proof by induction.

  For the induction step, assume that the maximum full regret of a deviation in $\DevSet$ and a single-target deviation in $\DevSet_{\preceq \TARGET}$ at information set $\infoSet$ is upper bounded at each immediate successor $\infoSet' \in \InfoSets_i(\infoSet, a)$ by $(d - 1)f(T)$, where $d$ is the height of $\infoSet$ ($d - 1$ player $i$ actions leads to a terminal information set).
  The full regret of deviation $\dev \in \DevSet$ decomposes as
  \begin{align}
    \regret_{\infoSet}^{1:T}(\dev)
      &=
        \regret^{1:T}_{\infoSet}(\dev_{\preceq \infoSet})
        + \hspace{-2.2em} \sum_{
          \infoSet' \in \bigcup_{a \in \Actions(\infoSet)}
            \InfoSets_i(\infoSet, a)
        } \hspace{-2.2em}
          \regret^{1:T}_{\infoSet'}(\dev).
    \shortintertext{according to \cref{lem:regretDecomposition}.
      We can bound the full regret at each $\infoSet'$ by the induction assumption,}
    \regret_{\infoSet}^{1:T}(\dev)
      &\le
        \regret^{1:T}_{\infoSet}(\dev_{\preceq \infoSet})
        + \abs{\InfoSets_i(\infoSet, a)}
          (d - 1)f(T).
    \shortintertext{We can then bound the maximum immediate regret at $\infoSet$ by \cref{eq:immediateSingleTargetRegretBound},}
    \regret_{\infoSet}^{1:T}(\dev)
      &\le
        f(T)
        + \abs{\InfoSets_i(\infoSet, a)}
          (d - 1)f(T)\\
      &\le
        \abs{\InfoSets_i}f(T),
      \label{eq:singleTargetElevationCompletion}
  \end{align}
  where the last inequality follows from the fact that
  $(d - 1)\abs{\InfoSets_i(\infoSet, a)} \le \abs{\InfoSets_i} - 1$.
  \Cref{eq:singleTargetElevationCompletion} completes the proof.
\end{proof}
\correctionEnd

\section{Regret Matching for Time Selection}
In an online decision problem (also called a \emph{prediction with expert advice} problem), \emph{regret matching} is a learning algorithm that accumulates a vector of regrets, $\regret^{1:t-1}$---one for each deviation or ``expert'', $\dev \in \DevSet \subseteq \DevSet^{\SWAP}_{\PureStratSet_i}$---and chooses its mixed strategy, $\strat^t$, on each round as the fixed point of a linear operator.
We generalize this algorithm and three extensions---\rmPlus/, regret approximation, and predictions---to the time selection setting.

\subsection{Background}

\subsubsection{Regret Matching}

The regret matching operator is constructed from a vector of non-negative \emph{link outputs},
$y^t \in \reals_+^{\abs{\DevSet}}$,
generated by applying a link function,
$f : \reals^{\abs{\PureStratSet_i}} \to \reals^{\abs{\PureStratSet_i}}_+$,
to the cumulative regrets, \ie/,
$y^t = f(\regret^{1:t-1})$.
The operator is defined as
\begin{align}
  \rmOperator^t :
    \strat_i
    \mapsto
    \frac{1}{z^t}
      \sum_{\dev \in \DevSet} \dev(\strat_i) y^t_{\dev},
\end{align}
where
$z^t = \sum_{\dev \in \DevSet} y^t_{\dev}$
is the sum of the link outputs,
and $\strat^t_i$ is chosen arbitrarily if $z^t = 0$.

Regret bounds are generally derived for regret matching algorithms by choosing $f = \alpha g$ for some $\alpha > 0$, where $g$ is part of a Gordon triple~\parencite{gordon2005no}, $\tuple{G, g, \gamma}$.
A Gordon triple is a triple consisting of a potential function, $G : \reals^n \to \reals$, a scaled link function $g: \reals^n \to \reals^n_+$, and a size function, $\gamma : \reals^n \to \reals_+$, where they satisfy the generalized smoothness condition
$G(x + x') \leq G(x) + x' \cdot g(x) + \gamma(x')$
for any $x,x' \in \reals^n$.
By applying the potential function to the cumulative regret, we can unroll the recursive bound to get a simple bound on the cumulative regret itself.

While its bounds are not quite optimal, \textcite{Hart00}'s original regret matching algorithm, defined with the \emph{rectified linear unit} (\emph{ReLU}) link function, $\cdot^+ = \max\set{\cdot, 0}$, is often exceptionally effective in practice (see, \eg/, \textcite{waugh2015unified,burch2017time}).
We focus our analysis on this link function, but our arguments readily apply to other link functions.
Only the final regret bounds will change.
We follow the typical convention for analyzing \textcite{Hart00}'s regret matching with $\gamma(x) = \frac{1}{2} \norm{x}_2^2$, $G(x) = \gamma(x^+)$, and $g = f$.

\subsubsection{Regret Matching\textsuperscript{+}}

Instead of the cumulative regrets, \rmPlus/ updates a vector of pseudo regrets (sometimes called ``q-regrets''), $q^{1:t} = (q^{1:t-1} + \regret^t)^+ \ge \regret^{1:t}$~\parencite{cfrPlus,solvingHulhe}.
If we assume a \emph{positive invariant} potential function where $G((x + x')^+) \le G(x + x')$, then the same regret bounds follow from the same arguments used in the analysis of regret matching~\textcite{dorazio2020}.
Note that this condition is satisfied with equality for the quadratic potential $G(x) = \frac{1}{2} \norm{x^+}_2^2$.

\subsubsection{Regret Approximation}

Approximate regret matching is regret matching with approximated cumulative regrets, $\like{\regret^{1:t - 1}} \approx \regret^{1:t - 1}$~\parencite{waugh2015solving,dorazio2019frcfr}
or q-regrets, $\like{q^{1:t - 1}} \approx q^{1:t - 1}$~\parencite{morrill2016,dorazio2020}.
The regret of approximate regret matching depends on its approximation accuracy and motivates the use of function approximation when it is impractical to store and update the regret for each deviation individually.
While it requires an extra assumption, we derive simpler approximate regret matching bounds than those derived by \textcite{dorazio2019frcfr,dorazio2020} through an analysis of regret matching with predictions.

\subsubsection{Optimism via Predictions}

Optimistic regret matching augments its link inputs by adding a prediction of the instantaneous regret on the next round, \ie/, $m^t \sim \regret^t$.
If the predictions are accurate then the algorithm's cumulative regret will be very small.
This is a direct application of optimistic Lagrangian Hedging~\parencite{dorazioOptLH} to $\DevSet$-regret.
The general approach of adding predictions to improve the performance of regret minimizers originates with \textcite{rakhlin2013optimization,syrgkanis2015fast}.

\textcite{dorazioOptLH}'s analysis requires that $G$ and $g$ satisfy
$G(x') \ge G(x) + \ip{g(x)}{x' - x}$,
which is achieved, for example, if $G$ is convex and $g$ is a subgradient of $G$.
Note that this is achieved for \textcite{Hart00}'s regret matching because \textcite{greenwald2006bounds} shows that the ReLU function is the gradient of the convex quadratic potential $G(x) = \frac{1}{2} \norm{x^+}_2^2$.

\subsection{Time Selection}

To adapt regret matching to the time selection framework, we treat each deviation--time selection function pair as a separate expert and sum over the link outputs corresponding to a given deviation to construct the regret matching operator.
Our goal is then to ensure that each element of the cumulative regret matrix, $\regret^{1:T}$, grows sublinearly, where each index in the second dimension corresponds to a time selection function.
Each deviation $\dev \in \DevSet$ is assigned a finite set of time selection functions, $\tSelectionFn \in \TSelectionSet(\dev)$, so the regret matrix entries corresponding to $(\dev, \tSelectionFn)$-pairings where
$\tSelectionFn \notin \TSelectionSet(\dev)$,
are always zero.

To facilitate a unified analysis, we assume a general optimistic regret matching algorithm that, after $t - 1$ rounds, uses link outputs
$y^t_{\dev}
  = \sum_{\tSelectionFn \in \TSelectionSet(\dev)}
    \tSelectionFn^t (x^t_{\dev, \tSelectionFn} + m^t_{\dev, \tSelectionFn})^+$,
where either $x^t = \regret^{1:t - 1}$ or $x^t = q^{1:t-1}$ with $x^1 = \zeros$, and $m^t$ is a matrix of arbitrary predictions or approximation errors.
Notice that this means that $x^t + m^t$ can be generated from a function approximator instead of storing either term in a table.
Denoting the weighted sum of the link outputs as
$z^t
  = \sum_{\dev \in \DevSet}
    y^t_{\dev}$,
the regret matching operator has the same form as initially defined, \ie/,
\begin{align}
  \rmOperator^t :
    \strat_i
    \mapsto
    \frac{1}{z^t}
      \sum_{\dev \in \DevSet}
        \dev(\strat_i) y^t_{\dev}.
\end{align}

With this, we can bound the regret of optimistic regret matching thusly:
\setcounter{theorem}{2}
\begin{theorem}\label{thm:rm-opt-for-ts}
  Establish deviation set $\DevSet \subseteq \DevSet^{\SWAP}_{\PureStratSet_i}$ and finite time selection sets
  $\TSelectionSet(\dev) = \set{\tSelectionFn \in [0, 1]^T}_{j = 1}^{M(\dev)}$ for each deviation $\dev \in \DevSet$.
  On each round $1 \le t \le T$, $(\DevSet, \cdot^+)$-regret matching with respect to matrix $x^t$ (equal to either $\regret^{1:t - 1}$ or $q^{1:t-1}$) and predictions $m^t$ chooses its strategy,
  $\strat^t_i \in \StrategySet_i$,
  to be the fixed point of
  $\rmOperator^t: \strat_i \mapsto
    \nicefrac{1}{z^t}
    \sum_{\dev \in \DevSet} \dev(\strat_i) y^t_{\dev}$
  or an arbitrary strategy when $z^t = 0$, where
  link outputs are generated from
  $y^t_{\dev}
    = \sum_{\tSelectionFn \in \TSelectionSet(\dev)}
      \tSelectionFn^t (
        x^t_{\dev, \tSelectionFn} + m^t_{\dev, \tSelectionFn}
      )^+$
  and $z^t = \sum_{\dev \in \DevSet} y^t_{\dev}$.
  This algorithm ensures that
  \[\regret^{1:T}(\dev, \tSelectionFn) \le \sqrt{
    \sum_{t=1}^T
      \sum_{\substack{
        \dev' \in \DevSet,\\
        \bar{\tSelectionFn} \in \TSelectionSet(\dev')
      }}
        \subex*{\bar{\tSelectionFn}^t \regret(\dev'; \strat^t) - m^t_{\dev', \bar{\tSelectionFn}}}^2}\]
  for every deviation $\dev$ and time selection function $\tSelectionFn$.
\end{theorem}
\begin{proof}
  Let us overload $\TSelectionSet = \bigcup_{\dev \in \DevSet} \TSelectionSet(\dev)$ and let
  $a_{\cdot, \tSelectionFn} = [a_{\dev, \tSelectionFn}]_{\dev \in \DevSet}$
  for any matrix
  $a \in \reals^{\abs{\DevSet} \times \abs{\TSelectionSet}}$.
  Then, for any time selection function,
  $\tSelectionFn \in \TSelectionSet$,
  the quadratic potential function,
  $G(x) = \frac{1}{2} \norm{x^+}_2^2$
  is convex, positive invariant (with equality), has the ReLU function as its gradient~\parencite{greenwald2006bounds}, and is smooth with respect to
  $\gamma(x) = \frac{1}{2}\norm{x}^2_2$.
  Altogether, these properties imply that
  \begin{align}
    G\subex*{\subex*{x^t_{\cdot, \tSelectionFn} + \tSelectionFn^t \regret^t}^+}
      &=
        G\subex*{
          \subex*{x^t_{\cdot, \tSelectionFn}
          + m^t_{\cdot, \tSelectionFn}
          + \tSelectionFn^t \regret^t
          - m^t_{\cdot, \tSelectionFn}}^+}\\
      &=
        G\subex*{
          x^t_{\cdot, \tSelectionFn}
          + m^t_{\cdot, \tSelectionFn}
          + \tSelectionFn^t \regret^t
          - m^t_{\cdot, \tSelectionFn}}\\
      &\le
        G\subex*{
          x^t_{\cdot, \tSelectionFn}
          + m^t_{\cdot, \tSelectionFn}}
        + \ip{
            \tSelectionFn^t \regret^t - m^t_{\cdot, \tSelectionFn}
          }{
            \subex*{x^t_{\cdot, \tSelectionFn} + m^t_{\cdot, \tSelectionFn}}^+
          }
        + \gamma\subex*{
          \tSelectionFn^t \regret^t - m^t_{\cdot, \tSelectionFn}
        },
  \end{align}
  where $\regret^t = [\regret(\dev; \strat^t)]_{\dev \in \DevSet}$ is the vector of instantaneous regrets on round $t$.

  By convexity,
  $G(a) - G(b) \le \ip{\grad G(a)}{a - b}$,
  for any vectors $a$ and $b$, so we substitute $a = x^t_{\cdot, \tSelectionFn} + m^t_{\cdot, \tSelectionFn}$ and $b = x^t_{\cdot, \tSelectionFn}$ to bound
  $G(x^t_{\cdot, \tSelectionFn}
      + m^t_{\cdot, \tSelectionFn})
      - \ip{m^t_{\cdot, \tSelectionFn}}{(x^t_{\cdot, \tSelectionFn}
      + m^t_{\cdot, \tSelectionFn})^+}
    \le
      G(x^t_{\cdot, \tSelectionFn})$.
  Therefore,
  \begin{align}
    G\subex*{
        \subex*{x^t_{\cdot, \tSelectionFn} + \tSelectionFn^t \regret^t}^+
    }
      &\le
        G(x^t_{\cdot, \tSelectionFn})
        + \tSelectionFn^t \ip{
            \regret^t
          }{
            \subex*{x^t_{\cdot, \tSelectionFn} + m^t_{\cdot, \tSelectionFn}}^+
          }
        + \gamma\subex*{\tSelectionFn^t \regret^t - m^t_{\cdot, \tSelectionFn}}\\
      &=
        G\subex*{
          \subex*{x^t_{\cdot, \tSelectionFn}}^+
        }
        + \tSelectionFn^t \ip{
            \regret^t
          }{
            \subex*{x^t_{\cdot, \tSelectionFn} + m^t_{\cdot, \tSelectionFn}}^+
          }
        + \gamma\subex*{\tSelectionFn^t \regret^t - m^t_{\cdot, \tSelectionFn}}.
  \end{align}

  Summing the potentials across time selection functions,
  \begin{align}
    \sum_{\tSelectionFn \in \TSelectionSet}
        G\subex*{\subex*{x^t_{\cdot, \tSelectionFn} + \tSelectionFn^t \regret^t}^+}
      &\le
        \sum_{\tSelectionFn \in \TSelectionSet}
          G\subex*{\subex*{x^t_{\cdot, \tSelectionFn}}^+}
          + \tSelectionFn^t \ip{
                \regret^t
              }{
                \subex*{x^t_{\cdot, \tSelectionFn} + m^t_{\cdot, \tSelectionFn}}^+
              }
          + \gamma\subex*{\tSelectionFn^t \regret^t - m^t_{\cdot, \tSelectionFn}}.
  \end{align}

  With some algebra, we can rewrite the sum of inner products:
  \begin{align}
    \sum_{\tSelectionFn \in \TSelectionSet}
        \tSelectionFn^t \ip{\regret^t}{(x^t_{\cdot, \tSelectionFn} + m^t_{\cdot, \tSelectionFn})^+}
      &=
        \sum_{\tSelectionFn \in \TSelectionSet}
          \sum_{\dev \in \DevSet}
            \tSelectionFn^t
            \regret(\dev; \strat^t)
            \subex*{x^t_{\dev, \tSelectionFn} + m^t_{\dev, \tSelectionFn}}^+\\
      &=
        \sum_{\dev \in \DevSet}
          \regret(\dev; \strat^t)
          \sum_{\tSelectionFn \in \TSelectionSet(\dev)}
            \tSelectionFn^t
            \subex*{x^t_{\dev, \tSelectionFn} + m^t_{\dev, \tSelectionFn}}^+\\
      &=
        \sum_{\dev \in \DevSet} \regret(\dev; \strat^t) y^t_{\dev}\\
      &=
        \ip{\regret^t}{y^t}.
  \end{align}

  Since the strategy $\strat_i^t$ is the fixed point of $\rmOperator^t$ generated from link outputs $y^t$, the Blackwell condition $\ip{\regret^t}{y^t} \le 0$ is satisfied with equality.
  For proof, see, for example, \textcite{greenwald2006bounds}.
  The sum of potential functions after $T$ rounds are then bounded as
  \begin{align}
    \sum_{\tSelectionFn \in \TSelectionSet}
        G\subex*{\subex*{x^T_{\cdot, \tSelectionFn} + \tSelectionFn^T \regret^T}^+}
      &\le
        \sum_{\tSelectionFn \in \TSelectionSet}
          G\subex*{\subex*{x^T_{\cdot, \tSelectionFn}}^+}
          + \gamma(\tSelectionFn^T \regret^T - m^T_{\cdot, \tSelectionFn}).
  \end{align}

  Expanding the definition of $\gamma$,
  \begin{align}
    \sum_{\tSelectionFn \in \TSelectionSet}
        G\subex*{\subex*{x^T_{\cdot, \tSelectionFn} + \tSelectionFn^T \regret^T}^+}
      &\le
        \sum_{\tSelectionFn \in \TSelectionSet}
          G\subex*{\subex*{x^T_{\cdot, \tSelectionFn}}^+}
        + \dfrac{1}{2} \sum_{\tSelectionFn \in \TSelectionSet}
          \sum_{\dev \in \DevSet}
            \subex*{\tSelectionFn^T \regret(\dev; \strat^T) - m^T_{\dev, \tSelectionFn}}^2\\
      &=
        \sum_{\tSelectionFn \in \TSelectionSet}
          G\subex*{\subex*{x^T_{\cdot, \tSelectionFn}}^+}
        + \dfrac{1}{2} \sum_{\substack{
          \dev \in \DevSet,\\
          \tSelectionFn \in \TSelectionSet(\dev)}}
            \subex*{\tSelectionFn^T \regret(\dev; \strat^T) - m^T_{\dev, \tSelectionFn}}^2.\\
      \shortintertext{Unrolling the recursion accross time,}
      &=
        \dfrac{1}{2}
        \sum_{t = 1}^T
          \sum_{\substack{
              \dev \in \DevSet,\\
              \tSelectionFn \in \TSelectionSet(\dev)}}
            \subex*{\tSelectionFn^t \regret(\dev; \strat^t) - m^t_{\dev, \tSelectionFn}}^2.
  \end{align}

  We lower bound
  \begin{align}
    \sum_{\tSelectionFn \in \TSelectionSet}
        G\subex*{\subex*{x^{T + 1}_{\cdot, \tSelectionFn}}^+}
      &=
        \frac{1}{2}
        \sum_{\tSelectionFn \in \TSelectionSet}
          \sum_{\dev \in \DevSet}
            \subex*{\subex*{ x^{T + 1}_{\dev, \tSelectionFn} }^+}^2\\
      &\geq
        \frac{1}{2}
        \max_{\substack{
          \dev \in \DevSet,\\
          \tSelectionFn \in \TSelectionSet(\dev)
        }}
        \subex*{
          \subex*{x^{T + 1}_{\dev, \tSelectionFn}}^+
        }^2\\
    \shortintertext{so that}
    \frac{1}{2}
    \max_{\substack{
      \dev \in \DevSet,\\
      \tSelectionFn \in \TSelectionSet(\dev)
    }}
    \subex*{
      \subex*{x^{T + 1}_{\dev, \tSelectionFn}}^+
    }^2
      &\leq
        \dfrac{1}{2}
        \sum_{t = 1}^T
          \sum_{\substack{
              \dev \in \DevSet,\\
              \tSelectionFn \in \TSelectionSet(\dev)}}
            \subex*{\tSelectionFn^t \regret(\dev; \strat^t) - m^t_{\dev, \tSelectionFn}}^2.
    \shortintertext{Multiplying both sides by two, taking the square root, and applying $\regret^{1:T}(\dev, \tSelectionFn) \le \subex*{x^{T + 1}_{\dev, \tSelectionFn}}^+$, we arrive at the final bound,}
    \max_{\substack{
      \dev \in \DevSet,\\
      \tSelectionFn \in \TSelectionSet(\dev)
    }} \regret^{1:T}(\dev, \tSelectionFn)
      &\leq
        \sqrt{
          \sum_{t = 1}^T
            \sum_{\substack{
                \dev' \in \DevSet,\\
                \bar{\tSelectionFn} \in \TSelectionSet(\dev')}}
              \subex*{\bar{\tSelectionFn}^t \regret(\dev'; \strat_i^t) - m^t_{\dev', \bar{\tSelectionFn}}}^2
        }.
  \end{align}
  Since the bound is true of the worst-case $\dev \in \DevSet$ and $\tSelectionFn \in \TSelectionSet$, it is true of each pair, thereby proving the claim.
\end{proof}

If all of the predictions $m^t$ are zero, then we arrive at a simple bound for exact regret matching.
We only prove the bound for ordinary regret matching for simplicity but the result and arguments are identical for exact \rmPlus/.
\setcounter{corollary}{0}
\begin{corollary}
  \label{cor:exact-rm}
  Given deviation set $\DevSet \subseteq \DevSet^{\SWAP}_{\PureStratSet_i}$ and finite time selection sets
$\TSelectionSet(\dev) = \set{\tSelectionFn_j \in [0, 1]^T}_{j = 1}^{M(\dev)}$ for each deviation $\dev \in \DevSet$,
$(\DevSet, \cdot^+)$-regret matching chooses a strategy on each round $1 \le t \le T$ as the fixed point of
$\rmOperator^t: \strat_i \mapsto
      \nicefrac{1}{z^t}
      \sum_{\dev \in \DevSet} \dev(\strat_i) y^t_{\dev}$
or an arbitrary strategy when $z^t = 0$, where
\emph{link outputs} are generated from exact regrets
$y^t_{\dev}
  = \sum_{\tSelectionFn \in \TSelectionSet(\dev)}
    \tSelectionFn^t (
      \regret^{1:t - 1}(\dev, \tSelectionFn)
    )^+$
and $z^t = \sum_{\dev \in \DevSet} y^t_{\dev}$.
This algorithm ensures that
$\regret^{1:T}(\dev, \tSelectionFn)
  \le
    2 \maxReward \sqrt{
      M^* \maxActivation(\DevSet) T
    }$
for any deviation $\dev$ and time selection function $\tSelectionFn$, where
$\maxActivation(\DevSet) = \max_{\pureStrat_i \in \PureStratSet_i} \sum_{\dev \in \DevSet} \ind{\dev(\pureStrat_i) \ne \pureStrat_i}$
is the maximal activation of $\DevSet$~\parencite{greenwald2006bounds}.
 \end{corollary}
\begin{proof}
  Since $m^t = \zeros$ on every round $t$, we know from \cref{thm:rm-opt-for-ts} that
  \begin{align}
    \regret^{1:T}(\dev, \tSelectionFn)
      &\leq
        \sqrt{
          \sum_{t = 1}^T
            \sum_{\substack{
                \dev' \in \DevSet,\\
                \bar{\tSelectionFn} \in \TSelectionSet(\dev')}}
              \subex*{\bar{\tSelectionFn}^t \regret(\dev'; \strat_i^t)}^2
        }\\
      &=
        \sqrt{
          \sum_{t = 1}^T
            \sum_{\dev' \in \DevSet}
              \subex*{\regret(\dev'; \strat_i^t)}^2
              \sum_{\bar{\tSelectionFn} \in \TSelectionSet(\dev')}
                \subex*{\bar{\tSelectionFn}^t}^2}.\\
      \shortintertext{Since $0 \le \bar{\tSelectionFn}^t \le 1$,}
      &\le
        \sqrt{
          M^* \sum_{t = 1}^T
            \sum_{\dev' \in \DevSet}
              \subex*{\regret(\dev'; \strat_i^t)}^2}.\\
      \shortintertext{Since
        $\sum_{\dev' \in \DevSet} \subex*{\regret(\dev'; \strat_i^t)}^2
          \le
            (2\maxReward)^2 \maxActivation(\DevSet)$
        (see \textcite{greenwald2006bounds}),}
      &\le
        \sqrt{M^* (2\maxReward)^2 \maxActivation(\DevSet) T}\\
      &=
        2\maxReward \sqrt{M^* \maxActivation(\DevSet) T}.\\
  \end{align}
  This completes the argument.
\end{proof}

If $x^t + m^t$ is generated from a function attempting to approximate $x^t + [\tSelectionFn^t \regret(\dev; \strat^t)]_{\dev \in \DevSet, \tSelectionFn \in \TSelectionSet}$, then we can rewrite \cref{thm:rm-opt-for-ts} in terms of its approximation error.
\begin{corollary}
  Establish deviation set $\DevSet \subseteq \DevSet^{\SWAP}_{\PureStratSet_i}$ and finite time selection sets
  $\TSelectionSet(\dev) = \set{\tSelectionFn \in [0, 1]^T}_{j = 1}^{M(\dev)}$ for each deviation $\dev \in \DevSet$.
  On each round $1 \le t \le T$, approximate $(\DevSet, \cdot^+)$-regret matching with respect to matrix $x^t$ (equal to either $\regret^{1:t - 1}$ or $q^{1:t-1}$) chooses its strategy,
  $\strat^t_i \in \StrategySet_i$,
  to be the fixed point of
  $\rmOperator^t: \strat_i \mapsto
    \nicefrac{1}{z^t}
    \sum_{\dev \in \DevSet} \dev(\strat_i) y^t_{\dev}$
  or an arbitrary strategy when $z^t = 0$, where
  link outputs are generated from approximation matrix $\like{y}^t \in \reals^{\abs{\DevSet} \times \abs{\TSelectionSet}}$ as
  $y^t_{\dev}
    = \sum_{\tSelectionFn \in \TSelectionSet(\dev)}
      \tSelectionFn^t (
        \like{y}^t_{\dev, \tSelectionFn}
      )^+$
  and $z^t = \sum_{\dev \in \DevSet} y^t_{\dev}$.
  This algorithm ensures that
  \[\regret^{1:T}(\dev, \tSelectionFn) \le \sqrt{
    \sum_{t=1}^T
      \sum_{\substack{
        \dev' \in \DevSet,\\
        \bar{\tSelectionFn} \in \TSelectionSet(\dev')
      }}
        \subex*{x^t_{\dev', \bar{\tSelectionFn}} + \bar{\tSelectionFn}^t \regret(\dev'; \strat^t) - \like{y}^t_{\dev', \bar{\tSelectionFn}}}^2}\]
  for every deviation $\dev$ and time selection function $\tSelectionFn$.
\end{corollary}
\begin{proof}
  Since the predictions $m^t$ are arbitrary, we can set $\like{y}^t = x^t + m^t$, which implies that $m^t = \like{y}^t - x^t$.
  Substituting this into the bound of \cref{thm:rm-opt-for-ts}, we arrive at the desired result.
\end{proof}

\section{EFR}
EFR's regret decomposition is a straightforward generalization of CFR's by \textcite{CFR-TR} but it requires us to build up some mathematical machinery.

\subsection{Preliminaries}

\textbf{Deviation player observations.}
Let
\[
  o: \Actions(\infoSet) \ni a; \DevSet^{\SWAP}_{\Actions(\infoSet)} \ni \dev = \begin{cases}
    * &\mbox{if } \dev \in \DevSet^{\EXT}_{\Actions(\infoSet)}\\
    a &\mbox{o.w.}
  \end{cases}
\]
return the observation that the deviation player makes when they apply action transformation $\dev$ and action $a$ is recommended.
Now, we can characterize how memory probabilities evolve in general as the deviation player chooses actions.
For any $\infoSet' \in \InfoSets_i(\infoSet, a')$ child of $\infoSet$ following action $a'$ and observation $b \in \set{*} \cup \Actions(\infoSet)$, the probability of memory $gb$ under behavioral deviation $\dev$ is
\begin{align}
  \tSelectionFn_{\dev}(\infoSet', gb; \strat_i)
    = \tSelectionFn_{\dev}(\infoSet, g; \strat_i) \sum_{a \in \Actions(\infoSet)} \ind{\dev_{\infoSet, g}(a) = a' \land o(a; \dev_{\infoSet, g}) = b} \strat_i(a \given \infoSet).
\end{align}

\subsection{Counterfactual Value}

\textbf{The counterfactual value of an action.}
The counterfactual value function is
\begin{align*}
  \cfv_{\infoSet} : a; \strat \mapsto
    \sum_{\substack{
      h \in \infoSet,\\
      z \in \TerminalHistories}}
        \reachProb(h; \strat_{-i})
        \reachProb(ha, z; \strat_i, \strat_{-i}) \utility_i(z).
\end{align*}
We overload
$\cfv_{\infoSet}(\strat'_i(\infoSet); \strat) = \E_{a' \sim \strat'_i(\infoSet)} \cfv_{\infoSet}(a'; \strat)$.
If the same strategy is used in $\infoSet$ as well as the following information sets, we overload
$\cfv_{\infoSet}(\strat) = \cfv_{\infoSet}(\strat_i(\infoSet); \strat)$.
By splitting up the histories that lead out of $\infoSet$ into those that terminate without further input from $i$ and those that lead to to child information sets, we can decompose counterfactual values recursively:
\begin{align}
    \cfv_{\infoSet}(a; \strat)
      &= \sum_{\substack{
        h \in \infoSet,\\
        z \in \TerminalHistories
      }} \reachProb(h; \strat_{-i}) \reachProb(ha, z; \strat) \utility_i(z)\\
      &= \sum_{\substack{
        h \in \infoSet,\\
        z \in \TerminalHistories
      }} \reachProb(ha, z; \strat_i) \underbrace{\reachProb(z; \strat_{-i}) \utility_i(z)}_{\text{Terminal counterfactual values.}}\\
      &= \underbrace{
        \sum_{z \in \TerminalHistories_i(\infoSet, a)}
          \reachProb(z; \strat_{-i}) \utility_i(z)
        }_{\text{Expected value from terminal histories.}}
        + \underbrace{
          \sum_{\substack{
            h' \in \infoSet' \in \InfoSets_i(\infoSet, a)\\
            z \in \TerminalHistories
          }}
            \reachProb(h', z; \strat_i) \reachProb(z; \strat_{-i}) \utility_i(z)
        }_{\text{Expected value from non-terminal histories.}}\\
      &= \termValue(\infoSet, a; \strat_{-i})
        + \sum_{\infoSet' \in \InfoSets_i(\infoSet, a)}
          \underbrace{
            \sum_{a' \in \Actions(\infoSet')}
              \strat_i(a' \given \infoSet')
              \sum_{\substack{
                h' \in \infoSet'\\
                z \in \TerminalHistories
              }}
                \reachProb(h'a', z; \strat_i) \reachProb(z; \strat_{-i}) \utility_i(z).
          }_{\cfv_{\infoSet'}(\strat)}\\
      &= \underbrace{\termValue(\infoSet, a; \strat_{-i})}_{\text{Expected immediate value.}}
        + \underbrace{
          \sum_{\infoSet' \in \InfoSets_i(\infoSet, a)}
            \cfv_{\infoSet'}(\strat)
        }_{\text{Expected future value.}}.
  \label{eq:bellman}
\end{align}

\textbf{The values for a behavioral deviation.}
To account for the deviation player's memory when evaluating behavioral deviations, we must define a new variation of counterfactual value.
\begin{definition}
  The counterfactual value of behavioral deviation $\dev \in \DevSet^{\IN}_{\InfoSets_i}$ from information set $\infoSet$ and memory state $g \in G_i$, given strategy profile $\strat \in \StrategySet$, is
  \begin{align*}
    \hat{\cfv}_{\infoSet, g}(\dev; \strat)
      &= \sum_{a \in \Actions(\infoSet)}
        \strat_i(a \given \infoSet)
          \subex*{
            \termValue(\infoSet, \dev_{\infoSet, g}(a); \strat_{-i})
            + \sum_{\infoSet' \in \InfoSets_i(\infoSet, \dev_{\infoSet, g}(a))}
              \hat{\cfv}_{\infoSet', go(a; \dev_{\infoSet, g})}(\dev; \strat)
          }.
  \end{align*}
\end{definition}
The memory probability--weighted counterfactual value, $\tSelectionFn_{\dev}(\infoSet, g; \strat_i) \hat{\cfv}_{\infoSet, g}(\dev; \strat)$, represents the value that deviation $\dev$ achieves by playing through $\infoSet$ and $g$ given $\strat$, has a similar recursive form as the counterfactual value:
\begin{align}
  &\tSelectionFn_{\dev}(\infoSet, g; \strat_i) \hat{\cfv}_{\infoSet, g}(\dev; \strat)\\
    &=
      \tSelectionFn_{\dev}(\infoSet, g; \strat_i)
      \sum_{a \in \Actions(\infoSet)}
        \strat_i(a \given \infoSet)
          \subex*{
            \termValue(\infoSet, \dev_{\infoSet, g}(a); \strat_{-i})
            + \sum_{\infoSet' \in \InfoSets_i(\infoSet, \dev_{\infoSet, g}(a))}
              \hat{\cfv}_{\infoSet', go(a; \dev_{\infoSet, g})}(\dev; \strat)
          }\\
    &=
      \tSelectionFn_{\dev}(\infoSet, g; \strat_i)
      \sum_{a \in \Actions(\infoSet)}
        \strat_i(a \given \infoSet) \termValue(\infoSet, \dev_{\infoSet, g}(a); \strat_{-i})\\
      &\quad+ \tSelectionFn_{\dev}(\infoSet, g; \strat_i)
      \sum_{a \in \Actions(\infoSet)}
        \strat_i(a \given \infoSet)
        \sum_{\infoSet' \in \InfoSets_i(\infoSet, \dev_{\infoSet, g}(a))}
            \hat{\cfv}_{\infoSet', go(a; \dev_{\infoSet, g})}(\dev; \strat)\\
    &=
      \sum_{a \in \Actions(\infoSet)}
        \tSelectionFn_{\dev}(\infoSet, g; \strat_i) \strat_i(a \given \infoSet) \termValue(\infoSet, \dev_{\infoSet, g}(a); \strat_{-i})\\
      &\quad+ \sum_{\substack{
          a' \in \Actions(\infoSet),\\
          b \in \set{*} \cup \Actions(\infoSet)}}
        \ind{\dev_{\infoSet, g}(a) = a' \land o(a; \dev_{\infoSet, g}) = b}
        \tSelectionFn_{\dev}(\infoSet, g; \strat_i)
        \sum_{a \in \Actions(\infoSet)}
          \strat_i(a \given \infoSet)
          \sum_{\infoSet' \in \InfoSets_i(\infoSet, a')}
            \hat{\cfv}_{\infoSet', gb}(\dev; \strat)\\
    &=
      \sum_{a \in \Actions(\infoSet)}
        \tSelectionFn_{\dev}(\infoSet, g; \strat_i) \strat_i(a \given \infoSet) \termValue(\infoSet, \dev_{\infoSet, g}(a); \strat_{-i})\\
      &\quad+ \sum_{\substack{
          a' \in \Actions(\infoSet),\\
          b \in \set{*} \cup \Actions(\infoSet)}}
        \sum_{\infoSet' \in \InfoSets_i(\infoSet, a')}
          \hat{\cfv}_{\infoSet', gb}(\dev; \strat)
        \underbrace{
          \tSelectionFn_{\dev}(\infoSet, g; \strat_i)
          \sum_{a \in \Actions(\infoSet)}
            \ind{\dev_{\infoSet, g}(a) = a' \land o(a; \dev_{\infoSet, g}) = b}
            \strat_i(a \given \infoSet)
        }_{\tSelectionFn_{\dev}(\infoSet', gb; \strat_i)}\\
    &=
      \sum_{a \in \Actions(\infoSet)}
        \tSelectionFn_{\dev}(\infoSet, g; \strat_i) \strat_i(a \given \infoSet) \termValue(\infoSet, \dev_{\infoSet, g}(a); \strat_{-i})
      + \sum_{\substack{
          a \in \Actions(\infoSet),\\
          \infoSet' \in \InfoSets_i(\infoSet, a),\\
          b \in \set{*} \cup \Actions(\infoSet)
        }}
          \tSelectionFn_{\dev}(\infoSet', gb; \strat_i) \hat{\cfv}_{\infoSet', gb}(\dev; \strat).
\end{align}

At any information set $\infoSet_0$ at the start of the game, both the counterfactual value and weighted counterfactual value of behavioral deviation $\dev$ matches the expected value of $\dev(\strat_i)$, \ie/,
\[\hat{\cfv}_{\infoSet_0, \emptyHistory}(\dev; \strat)
  = \tSelectionFn_{\dev}(\infoSet_0, \emptyHistory; \strat_i) \hat{\cfv}_{\infoSet_0, \emptyHistory}(\dev; \strat)
  = \E_{\pureStrat_i \sim \dev(\strat_i)} [\utility_i(\pureStrat_i, \strat_{-i})].\]
If there are no non-identity internal transformations in $\dev$, then $\hat{\cfv}$ reduces to the usual counterfactual value function under strategy profile $(\dev(\strat_i), \strat_{-i})$.
If all of the transformations following information set $\infoSet$ are identity transformations, then the counterfactual value of $\dev$ is just the counterfactual value of applying the transformation at $\infoSet$ and memory state $g$, \ie/,
\begin{align}
  \hat{\cfv}_{\infoSet, g}(\dev_{\preceq \infoSet, \sqsubseteq g}; \strat)
    &=
      \sum_{a \in \Actions(\infoSet)}
        \strat_i(a \given \infoSet)
          \subex*{
            \termValue(\infoSet, \dev_{\infoSet, g}(a); \strat_{-i})
            + \sum_{\infoSet' \in \InfoSets_i(\infoSet, \dev_{\infoSet, g}(a))}
              \hat{\cfv}_{\infoSet', go(a; \dev_{\infoSet, g})}(\dev^1; \strat)
          }\\
    &=
      \sum_{a \in \Actions(\infoSet)}
        \strat_i(a \given \infoSet)
          \subex*{
            \termValue(\infoSet, \dev_{\infoSet, g}(a); \strat_{-i})
            + \sum_{\infoSet' \in \InfoSets_i(\infoSet, \dev_{\infoSet, g}(a))}
              \cfv_{\infoSet'}(\strat)
          }\\
    &=
      \cfv_{\infoSet}(\dev_{\infoSet, g}(\strat_i(\infoSet)); \strat).
\end{align}

\subsection{Regret}

Next, we generalize the idea of immediate and full regret to behavioral deviations.

\textbf{Immediate regret.}
The immediate regret of behavioral deviation $\dev$ from information set $\infoSet$ and memory state $g$ is the immediate counterfactual regret weighted by the probability of $g$, \ie/,
\begin{align}
  \regret_{\infoSet}(\dev_{\preceq \infoSet, \sqsubseteq g}; \strat)
    &= \tSelectionFn_{\dev}(\infoSet, g; \strat_i)
      \subex*{
        \cfv_{\infoSet}(\dev_{\infoSet, g}(\strat_i(\infoSet)); \strat)
        - \cfv_{\infoSet}(\strat)
      }\\
    &= \tSelectionFn_{\dev}(\infoSet, g; \strat_i)
      \regret_{\infoSet}^{\CF}(\dev_{\infoSet, g}; \strat).
\end{align}

\textbf{Full regret.}
The full regret of behavioral deviation $\dev$ at information set $\infoSet$ and memory state $g$ is
\begin{align}
  \regret_{\infoSet, g}(\dev; \strat)
    &=
      \tSelectionFn_{\dev}(\infoSet, g; \strat_i)
      \subex*{
        \hat{\cfv}_{\infoSet, g}(\dev; \strat)
        - \cfv_{\infoSet}(\strat)
      }.
    \label{eq:full-phi-regret}
\end{align}

Some intuition can be gained about generalized immediate and full regret by making a formal connection between memory probabilities and reach probabilities.
If zero non-identity internal transformations are used in behavioral deviation $\dev$, then there is a unique memory state $g$ that $\dev$ realizes in information set $\infoSet$ and its probability coincides with the reach probability of the deviation strategy, $\dev(\strat_i)$, \ie/,
$\tSelectionFn_{\dev}(\infoSet, g; \strat_i) = \reachProb(\arbHistory(\infoSet); \dev(\strat_i))$.
After internal transformation $\dev^{a^{\TRIGGER} \to a^{\TARGET}}$ that outputs $a^{\TARGET} \ne a^{\TRIGGER}$ leading to a set of successors $\set{\infoSet' \succ \infoSet}$, there are two possible memory states, $ga^{\TARGET}$ and $ga^{\TRIGGER}$.
The reach probability at any $\infoSet'$ is then the sum of memory probabilities across these states, \ie/,
$\reachProb(\arbHistory(\infoSet'); \dev(\strat_i))
  = \tSelectionFn_{\dev}(\infoSet', ga^{\TARGET}; \strat_i) + \tSelectionFn_{\dev}(\infoSet', ga^{\TRIGGER}; \strat_i)$.
Thus, the full counterfactual regret at $\infoSet'$ weighted by a memory probability is nearly the difference in expected payoff from $\infoSet'$ between $\dev(\strat_i)$ and $\strat_i$ assuming that both play to $\infoSet'$ according to $\dev(\strat_i)$.
Minimizing regret with respect to each memory state is a stronger property than minimizing regret with respect to each deviation strategy reach probability because memory states distinguish between the different ways an information set could be reached.

At the start of the game, there is only one memory state, $\emptyHistory$, and
$\tSelectionFn_{\dev}(\infoSet, \emptyHistory; \strat_i) = 1$,
which means that $\regret_{\infoSet, \emptyHistory}(\dev; \strat)$ reduces to the difference in expected value achieved by $\dev(\strat_i)$ and $\strat_i$, \ie/,
\[\regret_{\infoSet, \emptyHistory}(\dev; \strat)
  = \E_{\pureStrat_i \sim \dev(\strat_i)} [\utility_i(\pureStrat_i, \strat_{-i})] - \utility_i(\strat).\]
Thus, bounding full regret at the start of the game ensures hindsight rationality.
We achieve this by showing a recursive decomposition between full and immediate regret so that an algorithm must only minimize immediate regret at each information set to be hindsight rational.
\cref{lem:recursivCfvDiff} is the key observation required for this decomposition.
\begin{lemma}
  \label{lem:recursivCfvDiff}
  Given strategy profile $\strat$ and behavioral deviation $\dev$, consider $\dev_{\infoSet, g}(\strat_i)$, the strategy for player $i$ that applies $\dev$ at information set $\infoSet$ assuming memory state $g$ and thereafter follows $\strat_i$.
  The regret for re-correlating after $\infoSet$ and $g$---that is, the difference between the counterfactual value of $\dev(\strat_i)$ and $\dev_{\infoSet, g}(\strat_i)$ from $\infoSet$ and $g$, weighted by the probability of $g$---is equal to the sum of full regrets at $\infoSet$'s and $g$'s children, \ie/,
  \[
    \tSelectionFn_{\dev}(\infoSet, g; \strat_i) \subex*{
      \hat{\cfv}_{\infoSet, g}(\dev; \strat) - \hat{\cfv}_{\infoSet, g}(\dev_{\preceq \infoSet, \sqsubseteq g}; \strat)
    }
      =
        \sum_{\substack{
            a' \in \Actions(\infoSet),\\
            \infoSet' \in \InfoSets_i(\infoSet, a'),\\
            b \in \set{*} \cup \Actions(\infoSet)}}
          \regret_{\infoSet', gb}(\dev; \strat).
  \]
\begin{proof}
  \begin{align}
    &\tSelectionFn_{\dev}(\infoSet, g; \strat_i) \subex*{
      \hat{\cfv}_{\infoSet, g}(\dev; \strat) - \hat{\cfv}_{\infoSet, g}(\dev_{\preceq \infoSet, \sqsubseteq g}; \strat)
    }\\
      &= \tSelectionFn_{\dev}(\infoSet, g; \strat_i)\hat{\cfv}_{\infoSet, g}(\dev; \strat)
        - \tSelectionFn_{\dev}(\infoSet, g; \strat_i)\hat{\cfv}_{\infoSet, g}(\dev_{\preceq \infoSet, \sqsubseteq g}; \strat)\\
      &=
        \sum_{a \in \Actions(\infoSet)}
          \tSelectionFn_{\dev}(\infoSet, g; \strat_i) \strat_i(a \given \infoSet) \termValue(\infoSet, \dev_{\infoSet, g}(a); \strat_{-i})
        + \sum_{\substack{
            a \in \Actions(\infoSet),\\
            \infoSet' \in \InfoSets_i(\infoSet, a),\\
            b \in \set{*} \cup \Actions(\infoSet)
          }}
            \tSelectionFn_{\dev}(\infoSet', gb; \strat_i) \hat{\cfv}_{\infoSet', gb}(\dev; \strat)\\
        &\quad- \sum_{a \in \Actions(\infoSet)}
          \tSelectionFn_{\dev}(\infoSet, g; \strat_i) \strat_i(a \given \infoSet) \termValue(\infoSet, \dev_{\infoSet, g}(a); \strat_{-i})
        - \sum_{\substack{
            a \in \Actions(\infoSet),\\
            \infoSet' \in \InfoSets_i(\infoSet, a),\\
            b \in \set{*} \cup \Actions(\infoSet)
          }}
            \tSelectionFn_{\dev}(\infoSet', gb; \strat_i) \cfv_{\infoSet', gb}(\strat)\\
      &=
        \sum_{\substack{
          a \in \Actions(\infoSet),\\
          \infoSet' \in \InfoSets_i(\infoSet, a),\\
          b \in \set{*} \cup \Actions(\infoSet)
        }}
          \tSelectionFn_{\dev}(\infoSet', gb; \strat_i) \hat{\cfv}_{\infoSet', gb}(\dev; \strat)
        - \sum_{\substack{
          a \in \Actions(\infoSet),\\
          \infoSet' \in \InfoSets_i(\infoSet, a),\\
          b \in \set{*} \cup \Actions(\infoSet)
        }}
          \tSelectionFn_{\dev}(\infoSet', gb; \strat_i) \cfv_{\infoSet', gb}(\strat)\\
      &=
        \sum_{\substack{
          a \in \Actions(\infoSet),\\
          \infoSet' \in \InfoSets_i(\infoSet, a),\\
          b \in \set{*} \cup \Actions(\infoSet)
        }}
          \tSelectionFn_{\dev}(\infoSet', gb; \strat_i) \subex*{
            \hat{\cfv}_{\infoSet', gb}(\dev; \strat) - \cfv_{\infoSet', gb}(\strat)}\\
      &=
        \sum_{\substack{
          a \in \Actions(\infoSet),\\
          \infoSet' \in \InfoSets_i(\infoSet, a),\\
          b \in \set{*} \cup \Actions(\infoSet)
        }}
          \regret_{\infoSet', gb}(\dev; \strat),
  \end{align}
  as required.
\end{proof}
\end{lemma}

The following two corollaries will help us to present a simple regret bound for EFR.
\begin{corollary}
  \label{cor:recursivCfvDiffCases}
  \cref{lem:recursivCfvDiff} has three cases:
  \begin{align*}
    &\tSelectionFn_{\dev}(\infoSet, g; \strat_i) \subex*{
      \hat{\cfv}_{\infoSet, g}(\dev; \strat) - \hat{\cfv}_{\infoSet, g}(\dev_{\preceq \infoSet, \sqsubseteq g}; \strat)
    }\\
      &= \begin{cases}
        \sum_{\substack{
          a \in \Actions(\infoSet),\\
          \infoSet' \in \InfoSets_i(\infoSet, a)
        }}
          \regret_{\infoSet', ga}(\dev; \strat)
            &\mbox{if } \dev_{\infoSet, g} = \dev^1\\
        \sum_{\infoSet' \in \InfoSets_i(\infoSet, a^{\TARGET})}
          \regret_{\infoSet', g*}(\dev; \strat)
            &\mbox{if } \exists a^{\TARGET} \in \Actions(\infoSet), \dev_{\infoSet, g} = \dev^{\to a^{\TARGET}}\\
          \sum_{\infoSet' \in \InfoSets_i(\infoSet, a^{\TARGET})}
            \regret_{\infoSet', ga^{\TRIGGER}}(\dev; \strat)\\
          + \sum_{\substack{
            a \in \Actions(\infoSet) \setminus \set{a^{\TRIGGER}},\\
            \infoSet' \in \InfoSets_i(\infoSet, a)
          }}
            \regret_{\infoSet', ga}(\dev; \strat)
          &\mbox{if }
            \exists a^{\TRIGGER} \ne a^{\TARGET} \in \Actions(\infoSet), \dev_{\infoSet, g} = \dev^{a^{\TRIGGER} \to a^{\TARGET}}.
      \end{cases}
  \end{align*}
  \begin{proof}
  Case 1: assume that $\dev_{\infoSet, g} = \dev^1$ (the identity transformation), then $\dev_{\infoSet, g}(a) = a$ and $o(a; \dev_{\infoSet, g}) = a$ for all $a \in \Actions(\infoSet)$.
  Thus,
  \begin{align}
    \tSelectionFn_{\dev}(\infoSet, g; \strat_i) \subex*{
      \hat{\cfv}_{\infoSet, g}(\dev; \strat) - \hat{\cfv}_{\infoSet, g}(\dev_{\preceq \infoSet, \sqsubseteq g}; \strat)
    }
      &=
        \sum_{\substack{
          a \in \Actions(\infoSet),\\
          \infoSet' \in \InfoSets_i(\infoSet, a)
        }}
          \regret_{\infoSet', ga}(\dev; \strat),
  \end{align}
  as required.

  Case 2: assume that $\dev_{\infoSet, g} = \dev^{\to a^{\TARGET}}$ (an external transformation), then $\dev_{\infoSet, g}(a) = a^{\TARGET}$ and $o(a; \dev_{\infoSet, g}) = *$ for all $a \in \Actions(\infoSet)$.
  Thus,
  \begin{align}
    &\tSelectionFn_{\dev}(\infoSet, g; \strat_i) \subex*{
      \hat{\cfv}_{\infoSet, g}(\dev; \strat) - \hat{\cfv}_{\infoSet, g}(\dev_{\preceq \infoSet, \sqsubseteq g}; \strat)
    }
      =
        \sum_{\infoSet' \in \InfoSets_i(\infoSet, a^{\TARGET})}
          \regret_{\infoSet', g*}(\dev; \strat),
  \end{align}
  as required.

  Case 3: assume that $\dev_{\infoSet, g} = \dev^{a^{\TRIGGER} \to a^{\TARGET}}$, $a^{\TRIGGER} \ne a^{\TARGET}$ (a non-identity internal transformation), then
  \begin{align}
    \tSelectionFn_{\dev}(\infoSet, g; \strat_i) \subex*{
      \hat{\cfv}_{\infoSet, g}(\dev; \strat) - \hat{\cfv}_{\infoSet, g}(\dev_{\preceq \infoSet, \sqsubseteq g}; \strat)
    }
      &=
        \sum_{\infoSet' \in \InfoSets_i(\infoSet, a^{\TARGET})}
          \regret_{\infoSet, ga^{\TRIGGER}}(\dev; \strat)
        + \sum_{\substack{
          a \in \Actions(\infoSet) \setminus \set{a^{\TRIGGER}},\\
          \infoSet' \in \InfoSets_i(\infoSet, a)
        }}
          \regret_{\infoSet, ga}(\dev; \strat),
  \end{align}
  as required.
\end{proof}
\end{corollary}
\begin{corollary}
  \label{cor:recursivCfvDiffSimple}
  If the full regret following information set $\infoSet$ and memory state $g$ is always bounded by $C \ge 0$, then the regret for re-correlating after $\infoSet$ and $g$ is bounded as
  \begin{align*}
    \tSelectionFn_{\dev}(\infoSet, g; \strat_i) \subex*{
      \hat{\cfv}_{\infoSet, g}(\dev; \strat) - \hat{\cfv}_{\infoSet, g}(\dev_{\preceq \infoSet, \sqsubseteq g}; \strat)
    }
      &\le
        \subex*{1 + \ind{\dev_{\infoSet, g} \in \DevSet_{\Actions(\infoSet)}^{\IN} \setminus \set{\dev^1}}}
        \abs*{\bigcup_{a \in \Actions(\infoSet)} \InfoSets_i(\infoSet, a)} C
  \end{align*}
\begin{proof}
  Case 1: assume that $\dev_{\infoSet, g} = \dev^1$ (the identity transformation), then
  \begin{align}
    \tSelectionFn_{\dev}(\infoSet, g; \strat_i) \subex*{
      \hat{\cfv}_{\infoSet, g}(\dev; \strat) - \hat{\cfv}_{\infoSet, g}(\dev_{\preceq \infoSet, \sqsubseteq g}; \strat)
    }
      &=
        \sum_{\substack{
          a \in \Actions(\infoSet),\\
          \infoSet' \in \InfoSets_i(\infoSet, a)
        }}
          \regret_{\infoSet', ga}(\dev; \strat)\\
      &\le
        \abs*{\bigcup_{a \in \Actions(\infoSet)} \InfoSets_i(\infoSet, a)} C,
  \end{align}
  as required.

  Case 2: assume that $\dev_{\infoSet, g} = \dev^{\to a^{\TARGET}}$ (an external transformation), then
  \begin{align}
    \tSelectionFn_{\dev}(\infoSet, g; \strat_i) \subex*{
      \hat{\cfv}_{\infoSet, g}(\dev; \strat) - \hat{\cfv}_{\infoSet, g}(\dev_{\preceq \infoSet, \sqsubseteq g}; \strat)
    }
      &=
        \sum_{\infoSet' \in \InfoSets_i(\infoSet, a^{\TARGET})}
          \regret_{\infoSet', g*}(\dev; \strat)\\
      &\le
        \abs*{\bigcup_{a \in \Actions(\infoSet)} \InfoSets_i(\infoSet, a)} C,
  \end{align}
  as required.

  Case 3: assume that $\dev_{\infoSet, g} = \dev^{a^{\TRIGGER} \to a^{\TARGET}}$, $a^{\TRIGGER} \ne a^{\TARGET}$ (an internal transformation), then
  \begin{align}
    \tSelectionFn_{\dev}(\infoSet, g; \strat_i) \subex*{
      \hat{\cfv}_{\infoSet, g}(\dev; \strat) - \hat{\cfv}_{\infoSet, g}(\dev_{\preceq \infoSet, \sqsubseteq g}; \strat)
    }
      &=
        \sum_{\infoSet' \in \InfoSets_i(\infoSet, a^{\TARGET})}
          \regret_{\infoSet, ga^{\TRIGGER}}(\dev; \strat)
        + \sum_{\substack{
          a \in \Actions(\infoSet) \setminus \set{a^{\TRIGGER}},\\
          \infoSet' \in \InfoSets_i(\infoSet, a)
        }}
          \regret_{\infoSet, ga}(\dev; \strat)\\
      &\le
        \abs*{\InfoSets_i(\infoSet, a^{\TARGET})} C
        + \abs*{\bigcup_{a \ne a^{\TRIGGER}} \InfoSets_i(\infoSet, a)} C\\
      &\le
        2 \abs*{\bigcup_{a \in \Actions(\infoSet)} \InfoSets_i(\infoSet, a)} C.
  \end{align}
  as required.
\end{proof}
\end{corollary}

We can now state our decomposition result:
\begin{lemma}
  \label{lem:uni-decomp}
  The full regret of behavioral deviation $\dev$ from information set $\infoSet$ and memory state $g$ is bounded by the immediate regret at $\infoSet$ plus the full regret at each of $\infoSet$'s and $g$'s children.
 \begin{proof}
  \begin{align}
  &\regret_{\infoSet, g}(\dev; \strat)\\
    &=
      \tSelectionFn_{\dev}(\infoSet, g; \strat_i)\subex*{
        \hat{\cfv}_{\infoSet, g}(\dev; \strat)
        - \cfv_{\infoSet}(\strat)}.\\
  \shortintertext{Adding and subtracting $\tSelectionFn_{\dev}(\infoSet, g; \strat_i) \hat{\cfv}_{\infoSet, g}(\dev_{\preceq \infoSet, \sqsubseteq g}; \strat)$,}
    &=
      \tSelectionFn_{\dev}(\infoSet, g; \strat_i) \subex*{
        \hat{\cfv}_{\infoSet, g}(\dev_{\preceq \infoSet, \sqsubseteq g}; \strat)
          - \cfv_{\infoSet}(\strat)
        }
      + \tSelectionFn_{\dev}(\infoSet, g; \strat_i) \subex*{
          \hat{\cfv}_{\infoSet, g}(\dev; \strat)
          - \hat{\cfv}_{\infoSet, g}(\dev_{\preceq \infoSet, \sqsubseteq g}; \strat)
        }\\
    &=
      \underbrace{
          \tSelectionFn_{\dev}(\infoSet, g; \strat_i) \subex*{
            \cfv_{\infoSet}\subex*{\dev_{\infoSet, g}\subex*{\strat_i(\infoSet)}; \strat}
            - \cfv_{\infoSet}(\strat)
          }
        }_{\text{Immediate regret, $\regret_{\infoSet}(\dev_{\preceq \infoSet, \sqsubseteq g}; \strat)$.}}
      + \underbrace{
          \tSelectionFn_{\dev}(\infoSet, g; \strat_i) \subex*{
            \hat{\cfv}_{\infoSet, g}(\dev; \strat)
            - \hat{\cfv}_{\infoSet, g}(\dev_{\preceq \infoSet, \sqsubseteq g}; \strat)
          }
        }_{
          \text{Regret for re-correlating after $\infoSet$ and $g$.}
        }\\
  \shortintertext{Applying \cref{lem:recursivCfvDiff},}
    &=
      \regret_{\infoSet}(\dev_{\preceq \infoSet, \sqsubseteq g}; \strat)
      + \sum_{\substack{
          a \in \Actions(\infoSet),\\
          \infoSet' \in \InfoSets_i(\infoSet, a),\\
          b \in \set{*} \cup \Actions(\infoSet)}}
        \regret_{\infoSet', gb}(\dev; \strat),
\end{align}
as required.
 \end{proof}
\end{lemma}
\begin{corollary}
  \label{cor:uni-decomp-simple}
  If the full regret of each child of information set $\infoSet$ is bounded by $C \ge 0$, then
  \begin{align*}
    \regret_{\infoSet, g}(\dev; \strat)
      \le
        \regret_{\infoSet}(\dev_{\preceq \infoSet, \sqsubseteq g})
        + \subex*{
          1
          + \ind{\dev_{\infoSet, g} \in \DevSet_{\Actions(\infoSet)}^{\IN} \setminus \set{\dev^1}}
        } \abs*{\bigcup_{a \in \Actions(\infoSet)} \InfoSets_i(\infoSet, a)} C.
  \end{align*}
\begin{proof}
  \cref{lem:uni-decomp,cor:recursivCfvDiffSimple}.
\end{proof}
\end{corollary}

Using \cref{cor:uni-decomp-simple} and instantiating EFR with exact regret matching, we can derive a simple regret bound that depends on the number of immediate regrets associated with a given subset of behavioral deviations.
\setcounter{resetCounter}{\arabic{theorem}}
\setcounter{theorem}{1}
\begin{theorem}
  \label{lem:efr}
  Instantiate EFR for player $i$ with exact regret matching and a set of behavioral deviations $\DevSet \subseteq \DevSet_{\InfoSets_i}^{\IN}$.
Let the maximum number of information sets along the same line of play where non-identity internal transformations are allowed before a non-identity transformation within any single deviation be $n_{\IN}$.
Let $D = \max_{\infoSet \in \InfoSets_i, \dev_{\infoSet} \in \DevSet_{\infoSet}} \abs{\TSelectionSet_{\infoSet}^{\DevSet}(\dev_{\infoSet})} \maxActivation(\DevSet_{\infoSet})$.
\correctionStart{Add a mention about how the full regret at each information set has this bound. The previous statement is still true but this statement is stronger.}
Then, EFR's cumulative full regret at each information set after $T$ rounds with respect to $\DevSet$ and the set of single-target deviations generated from $\DevSet$, $\DevSet_{\preceq \TARGET}$, is upper bounded by
$2^{n_{\IN} + 1} \maxReward \abs{\InfoSets_i} \sqrt{D T}$.
In addition, this implies that EFR is OS hindsight rational with respect to $\DevSet \cup \DevSet_{\preceq \TARGET}$.
\correctionEnd \begin{proof}
  EFR keeps track of each immediate regret for each transformation associated with each realizable memory state $g$ in each information set $\infoSet$.
  EFR's immediate strategies at each $\infoSet$ on each round are chosen according to time selection regret matching on the cumulative immediate regrets and memory state probabilities there.
  Therefore, the cumulative immediate regret at each information set and memory state is bounded as
  $\sum_{t = 1}^T \regret_{\infoSet}(\dev_{\preceq \infoSet, \sqsubseteq g}) \le 2 \maxReward \sqrt{D T}$
  according to \cref{cor:exact-rm}.
  Working from the leaves of the information set tree to the roots, we recursively bound the full regret according to \cref{cor:uni-decomp-simple}.
  The full regret at each information set is then bounded as
  $\sum_{t = 1}^T \regret_{\infoSet, g}(\dev; \strat^t) \le 2^{n_{\IN} + 1} \maxReward \abs{\InfoSets_i} \sqrt{D T}$.
  EFR's cumulative regret with respect to $\DevSet$ is equal to its cumulative full regret at the start of the game, so the former is bounded by $2^{n_{\IN} + 1} \maxReward \abs{\InfoSets_i} \sqrt{D T}$ as well, which concludes the argument.
\end{proof}
\end{theorem}
\setcounter{theorem}{\arabic{resetCounter}}

See \cref{tab:efrParams} for EFR parameters each deviation type.
\def\behavioralDevDef{
  \stretchBox{
    $\forall \infoSet', \, \exists a^{\TRIGGER}_{\infoSet'}, a'_{\infoSet'}$,\\
    $\begin{cases}
      a'_{\infoSet} &\mbox{\scriptsize if }
        {
          \forall \bar{\infoSet} \preceq \infoSet,
          \pureStrat_i(\bar{\infoSet}) = a^{\TRIGGER}_{\bar{\infoSet}}}\\
      \pureStrat_i(\infoSet) &\mbox{\scriptsize o.w.}
    \end{cases}$}
}
\def\tipsDef{
  \stretchBox{
    $\exists \infoSet^{\TRIGGER}, a^{\TRIGGER}, \infoSet^{\TARGET}, a^{\TARGET}, a^{\TARGET \TRIGGER}$,\\
    $\begin{cases}
      a^{\TARGET} &\mbox{\scriptsize if }
        { \infoSet = \infoSet^{\TARGET},}
        { \pureStrat_i(\infoSet^{\TARGET}) = a^{\TARGET \TRIGGER},}\\
        &{ \pureStrat_i(\infoSet^{\TRIGGER}) = a^{\TRIGGER}}\\
      a_{\infoSet}^{\to \infoSet^{\TARGET}} &\mbox{\scriptsize if }
        { \infoSet \succeq \infoSet^{\TRIGGER}, \pureStrat_i(\infoSet^{\TRIGGER}) = a^{\TRIGGER}}\\
      \pureStrat_i(\infoSet) &\mbox{\scriptsize o.w.}
    \end{cases}$}
}
\def\cspsDef{
  \stretchBox{
    $\exists \infoSet^{\TRIGGER}, a^{\TRIGGER}, \infoSet^{\TARGET}, a^{\TARGET}$,\\
    $\begin{cases}
      a^{\TARGET} &\mbox{\scriptsize if }
        { \infoSet = \infoSet^{\TARGET},}
        { \pureStrat_i(\infoSet^{\TRIGGER}) = a^{\TRIGGER}}\\
      a_{\infoSet}^{\to \infoSet^{\TARGET}} &\mbox{\scriptsize if }
        { \infoSet \succeq \infoSet^{\TRIGGER}, \pureStrat_i(\infoSet^{\TRIGGER}) = a^{\TRIGGER}}\\
      \pureStrat_i(\infoSet) &\mbox{\scriptsize o.w.}
    \end{cases}$}
}
\def\cfpsDef{
  \stretchBox{
    $\exists \infoSet^{\TRIGGER}, \infoSet^{\TARGET}, a^{\TARGET}, a^{\TARGET \TRIGGER}$,\\
    $\begin{cases}
      a^{\TARGET} &\mbox{\scriptsize if }
        { \infoSet = \infoSet^{\TARGET},}
        { \pureStrat_i(\infoSet^{\TARGET}) = a^{\TARGET \TRIGGER}}\\
      a_{\infoSet}^{\to \infoSet^{\TARGET}} &\mbox{\scriptsize if }
        { \infoSet \succeq \infoSet^{\TRIGGER}}\\
      \pureStrat_i(\infoSet) &\mbox{\scriptsize o.w.}
    \end{cases}$}
}
\def\bpsDef{
  \stretchBox{
    $\exists \infoSet^{\TRIGGER}, \infoSet^{\TARGET}, a^{\TARGET}$,\\
    $\begin{cases}
      a^{\TARGET} &\mbox{\scriptsize if }
        { \infoSet = \infoSet^{\TARGET}}\\
      a_{\infoSet}^{\to \infoSet^{\TARGET}} &\mbox{\scriptsize if }
        { \infoSet \succeq \infoSet^{\TRIGGER}}\\
      \pureStrat_i(\infoSet) &\mbox{\scriptsize o.w.}
    \end{cases}$}
}
\def\informedCfDef{
  \stretchBox{
    $\exists \infoSet^{\TARGET}, a^{\TARGET}, a^{\TARGET \TRIGGER}$,\\
    $\begin{cases}
      a^{\TARGET} &\mbox{\scriptsize if }
        {\scriptstyle \infoSet = \infoSet^{\TARGET},}
        {\scriptstyle \pureStrat_i(\infoSet^{\TARGET}) = a^{\TARGET \TRIGGER}}\\
      a_{\infoSet}^{\to \infoSet^{\TARGET}} &\mbox{\scriptsize if }
        {\scriptstyle \infoSet \preceq \infoSet^{\TARGET}}\\
      \pureStrat_i(\infoSet) &\mbox{\scriptsize o.w.}
    \end{cases}$}
}
\def\blindCfDef{
  \stretchBox{
    $\exists \infoSet^{\TARGET}, a^{\TARGET}$,\\
    $\begin{cases}
      a^{\TARGET} &\mbox{\scriptsize if }
        {\scriptstyle \infoSet = \infoSet^{\TARGET}}\\
      a_{\infoSet}^{\to \infoSet^{\TARGET}} &\mbox{\scriptsize if }
        {\scriptstyle \infoSet \preceq \infoSet^{\TARGET}}\\
      \pureStrat_i(\infoSet) &\mbox{\scriptsize o.w.}
    \end{cases}$}
}
\def\informedActionDef{
  \stretchBox{
    $\exists \infoSet^{\TRIGGER}, a^{\TARGET}, a^{\TRIGGER}$,\\
    $\begin{cases}
      a^{\TARGET} &\mbox{\scriptsize if }
        {\scriptstyle \infoSet = \infoSet^{\TRIGGER},}
        {\scriptstyle \pureStrat_i(\infoSet^{\TRIGGER}) = a^{\TRIGGER}}\\
      \pureStrat_i(\infoSet) &\mbox{\scriptsize o.w.}
    \end{cases}$}
}
\def\blindActionDef{
  \stretchBox{
    $\exists \infoSet^{\TRIGGER}, a^{\TARGET}$,\\
    $\begin{cases}
      a^{\TARGET} &\mbox{\scriptsize if }
        {\scriptstyle \infoSet = \infoSet^{\TRIGGER}}\\
      \pureStrat_i(\infoSet) &\mbox{\scriptsize o.w.}
    \end{cases}$}
}
\def\informedCausalDef{
  \stretchBox{
    $\exists \infoSet^{\TRIGGER}, a^{\TRIGGER}, \pureStrat'$,\\
    $\begin{cases}
      \pureStrat'(\infoSet) &\mbox{\scriptsize if }
        {\scriptstyle \infoSet \succeq \infoSet^{\TRIGGER},}
        {\scriptstyle \pureStrat_i(\infoSet^{\TRIGGER}) = a^{\TRIGGER}}\\
      \pureStrat_i(\infoSet) &\mbox{\scriptsize o.w.}
    \end{cases}$}
}
\def\blindCausalDef{
  \stretchBox{
    $\exists \infoSet^{\TRIGGER}, \pureStrat'$,\\
    $\begin{cases}
      \pureStrat'(\infoSet) &\mbox{\scriptsize if }
        {\scriptstyle \infoSet \succeq \infoSet^{\TRIGGER}}\\
      \pureStrat_i(\infoSet) &\mbox{\scriptsize o.w.}
    \end{cases}$}
}

\def\lightrule{\specialrule{0.1pt}{2mm}{2mm}}

{\renewcommand{\arraystretch}{5.2}
  \begin{table}[tb]
  \centering
  \caption{Formal definition of the strategy generated by a deviation of each type given pure strategy ${\pureStrat_i \in \PureStratSet_i}$ at each information set $\infoSet \in \InfoSets_i$.
  }
  \label{tab:devDefs}
  \begin{threeparttable}
    \begin{tabularx}{\linewidth}{lXlX}
      \stretchBox{
        single-target\\
        behavioral
      } & \behavioralDevDef & in.\ causal & \informedCausalDef\\
      TIPS & \tipsDef & in.\ action & \informedActionDef\\
      CSPS & \cspsDef & in.\ CF & \informedCfDef\\
      CFPS & \cfpsDef & blind causal & \blindCausalDef\\
      BPS & \bpsDef & blind action & \blindActionDef\\
      && blind CF & \blindCfDef
    \end{tabularx}
\end{threeparttable}
\end{table}}

\def\numActions{n_{\Actions}}
\def\DevSetInMinusI{\DevSet^{\INT}_{\Actions(\infoSet)} \setminus \set{\dev^1}}
\def\DevSetI{\DevSet^{\INT}_{\Actions(\infoSet)}}
\def\DevSetE{\DevSet^{\EXT}_{\Actions(\infoSet)}}
\def\numDevInMinusI{\numActions^2 - \numActions}

\def\behavioralWeights{
  $\underset{
    \hfill \bar{\infoSet}' \prec \infoSet, \,
    \forall \bar{\infoSet} \preceq \bar{\infoSet}',
      a_{\bar{\infoSet}} \in \Actions(\bar{\infoSet})
  }{
    \set{t \mapsto 1} \cup \set*{
      t \mapsto
        \prod_{\bar{\infoSet} \preceq \bar{\infoSet}'}
          \strat^t_i(a_{\bar{\infoSet}} \given \bar{\infoSet})
    }
  }$
}
\def\tipsWeights{
  $\underset{
    \hfill \infoSet^{\TRIGGER} \prec \infoSet, a^{\TRIGGER} \in \Actions(\infoSet^{\TRIGGER})
  }{
    \set{t \mapsto 1} \cup \set*{
      t \mapsto
        \reachProb(\arbHistory(\infoSet^{\TRIGGER}); \strat^t_i)
        \strat^t_i(a^{\TRIGGER} \given \infoSet^{\TRIGGER})
    }
  }$
}
\def\cspsWeights{
  $\begin{cases}
    \underset{
      \hfill \infoSet^{\TRIGGER} \prec \infoSet, a^{\TRIGGER} \in \Actions(\infoSet^{\TRIGGER})
    }{
      \begin{aligned}[b]
        &\set{t \mapsto 1} \cup\\
        &\set*{
          t \mapsto
            \reachProb(\arbHistory(\infoSet^{\TRIGGER}); \strat^t_i)
            \strat^t_i(a^{\TRIGGER} \given \infoSet^{\TRIGGER})
        }
      \end{aligned}
    }
      &\mbox{if } \dev_{\infoSet} \in \DevSet_{\Actions(\infoSet)}^{\EXT}\\
    \set*{
      t \mapsto \reachProb(\arbHistory(\infoSet); \strat^t_i)
    }
      &\mbox{o.w.}
  \end{cases}$
}
\def\cfpsWeights{
  $\underset{\hfill \infoSet^{\TRIGGER} \preceq \infoSet}{
    \set{t \mapsto 1} \cup \set*{
      t \mapsto \reachProb(\arbHistory(\infoSet^{\TRIGGER}); \strat^t_i)
    }
  }$
}
\def\bpsWeights{\cfpsWeights}
\def\actionDevWeights{
  $\set*{t \mapsto \reachProb(\arbHistory(\infoSet); \strat^t_i)}$
}
\def\cfWeights{$\set*{t \mapsto 1}$}

{\renewcommand{\arraystretch}{3}
\begin{table}[tb]
  \centering
  \centerline{\begin{threeparttable}
  \caption{EFR parameters and regret bound constants for different deviation types.}
  \label{tab:efrParams}
  \small
  \begin{tabular}{@{}l|@{\hskip 1em}c@{\hskip 1em}c@{\hskip -2em}c@{}c@{\hskip 1em}c@{}}
    \toprule
    type &
    \stretchBox{
      $\DevSet_{\infoSet}$
for all ${\scriptstyle \infoSet \in \InfoSets_i}$} &
    \stretchBox{
      $\TSelectionSet^{\DevSet}_{\infoSet}$ as function of $\dev_{\infoSet}$
for all ${\scriptstyle \infoSet \in \InfoSets_i}$} &
    \stretchBox{
      $\max\limits_{\infoSet \in \InfoSets_i, \dev_{\infoSet} \in \DevSet_{\infoSet}} \abs*{\TSelectionSet_{\infoSet}^{\DevSet}(\dev_{\infoSet})}$} &
    $D$ &
    $n_{\IN}$\\
    \midrule
      behavioral &
        $\DevSetInMinusI$ &
        \behavioralWeights &
        $\numActions^{d_*}$ &
        $\numActions^{d_*}
          \subex{ \numActions^2 - \numActions }$ &
        $d_*$\\
      TIPS &
        $\DevSetInMinusI$ &
        \tipsWeights &
        $d_* \numActions + 1$ &
        $\subex*{ d_* \numActions + 1 }
          \subex{ \numActions^2 - \numActions }$ &
        1\\
      \stretchBox{CSPS /\\in.\ causal} &
        $\DevSet^{\EXT}_{\Actions(\infoSet)} \cup \DevSetInMinusI$ &
        \cspsWeights &
        $d_* \numActions$ &
        $d_*\numActions (\numActions^2 - 2)$ &
        1\\
      CFPS &
        $\DevSetInMinusI$ &
        \cfpsWeights &
        $d_* + 1$ &
        $\subex*{ d_* + 1 }
          \subex{ \numActions^2 - \numActions }$ &
        0\\
      \stretchBox{BPS /\\blind causal} &
        $\DevSet^{\EXT}_{\Actions(\infoSet)}$ &
        \bpsWeights &
        $d_* + 1$ &
        $\subex*{ d_* + 1 }
          \subex{\numActions - 1}$ &
        0\\
      in.\ action &
        $\DevSetInMinusI$ &
        \actionDevWeights &
        $1$ &
        $\numActions^2 - \numActions$ &
        0\\
      in.\ CF &
        $\DevSetInMinusI$ &
        \cfWeights &
        $1$ &
        $\numActions^2 - \numActions$ &
        0\\
      blind action &
        $\DevSet^{\EXT}_{\Actions(\infoSet)}$ &
        \actionDevWeights &
        $1$ &
        $\numActions - 1$ &
        0\\
      \stretchBox{
        blind CF /\\
        external
      } &
        $\DevSet^{\EXT}_{\Actions(\infoSet)}$ &
        \cfWeights &
        $1$ &
        $\numActions - 1$ &
        0\\
    \bottomrule
  \end{tabular}
  \end{threeparttable}}
\end{table}}
 
The variable $D$ in the EFR regret bound that depends on the particular behavioral deviation subset with which it is instantiated is often the number of immediate regrets generated by that subset divided by the number of information sets.
$D$ is slightly larger for CSPS because it uses the union of internal and external transformations for its action transformation set, $\DevSet_{\infoSet}$, at all information sets except those at the beginning of the game.
Since our bound depends on the maximum number of memory states associated with any $\dev \in \DevSet_{\infoSet}$ and $\maxActivation(\DevSet_{\infoSet})$ counts the maximum number of non-trivial ways an action can be transformed according to the transformations in $\DevSet_{\infoSet}$, their product ends up being larger for CSPS than the number of valid combinations between memory states and action transformations.
See \cref{tab:efrParams} for $D$ values corresponding to each partial sequence deviation type.

\section{Regret Matching++}
\textcite{kash2019combining} presents the regret matching++ algorithm and claims that it is no-external-regret.
This algorithm's proposed regret bound implies a sublinear bound on cumulative
positive regret, which would further imply that it has the same bound with respect to \emph{all possible} time selection functions.
The surprising aspect of this result is that the algorithm does not require any information about any of the possible time selection functions and requires no more computation or storage than basic regret matching.
The following result, \cref{thm:lin-lower-bound-sum-pos-regrets}, shows that there is actually no algorithm that can achieve a sublinear bound on cumulative positive regret.
This result proves that regret matching++ cannot be no-external-regret as claimed.
\cref{sec:regret-matching-regret-bound-bug} identifies the mistake in the regret matching++ bound proof.

\subsection{Linear Lower Bound on the Sum of Positive
Regrets}\label{sec:linear-lower-bound-on-the-sum-of-positive-regrets}

\begin{theorem}
  \label{thm:lin-lower-bound-sum-pos-regrets}
  The worst-case maximum cumulative positive regret,
  $Q^T = \max_{a \in \Actions} \sum_{t = 1}^T \subex*{\reward^t(a) - \ip{\policy^t}{\reward^t}}^+$, under a sequence of reward functions chosen from the class of bounded reward functions, $(\reward^t \in \set{\reward : \reward \in \reals^{\abs{\Actions}}, \, \norm{\reward}_{\infty} \le 1})_{t = 1}^T$, of any algorithm that chooses policies $\policy^t \in \simplex^{\abs{\Actions}}$ over a finite set of actions, $\Actions$, in an online fashion over $T$ rounds, is at least $T/4$.
\begin{proof}
Without loss of generality, consider a two action environment, $\Actions = \tuple{a, a'}$, and any learning algorithm that deterministically chooses a distribution, $\policy^t \in \simplex^2$, over them on each round $t$. The environment gets to see the learner's policy before presenting a reward function. If the learner weights one action more than the other, the environment gives a reward of zero for the action with the larger weight and one to the action with the smaller weight.
Formally, if $\policy^t(a) \ge 0.5$, then $\reward^t(a) = 0$, $\reward^t(a') = 1$, and vice-versa otherwise.

Let $a_{\text{low}} = a'$ if $\policy^t(a) \ge 0.5$ and $a_{\text{low}} = a$ otherwise. The positive regrets on any round $t$ are
$(1 - \policy^t(a_{\text{low}}))^+ \ge 0.5$ and
$\subex*{0 - (1 - \policy^t(a_{\text{low}}))}^+ = 0$.
So the learner is forced to suffer at least 0.5 positive regret on each round for one of the actions. Since there are only two actions, then over $T$ rounds one of the actions must have accumulated a regret of 0.5 on at least $T/2$ rounds. The cumulative positive regret for this action must then be $T/4$.
Therefore, the maximum cumulative positive regret of any deterministic algorithm in this environment must be at least $T/4$.

To extend this result to include algorithms that stochastically choose $\policy^t$, we simply need to consider the expected cumulative positive regret and notice that the rectified linear unit function ($\cdot^+$) is convex.
By Jensen's inequality and the fact that the max of an expectation is no larger than the expectation of the max, the expected cumulative positive regret is lower bounded by the cumulative positive regret under the learner's expected distributions,
$\E[\policy^t]$,
\ie/,
\[
  \E[Q^T] \ge \max_{a \in \Actions} \sum_{t = 1}^T \subex*{\reward^t(a) - \ip{\E[\policy^t]}{\reward^t}}^+ \ge T/4.
\]
Since $\E[\policy^t]$ is a single distribution, we have reduced the stochastic case to the deterministic case, thereby showing they have the same regret lower bound.
\end{proof}
\end{theorem}

\subsection{Regret Matching++ Regret Bound Bug}
\label{sec:regret-matching-regret-bound-bug}

Define the cumulative positive regret of action \(a \in \Actions\) as
\(Q^T_a = \sum_{t = 1}^T (\regret^t_a)^+ = \sum_{t = 1}^T \subex{\reward^t_a - \ip{\policy^t}{\reward^t}}^+\),
where \(\reward^t\) is the reward function on round \(t\).
\textcite{kash2019combining} bounds
\((\max_a Q^T_a)^2 \leq \sum_a (Q^T_a)^2 = \sum_a \subex*{Q^{T - 1}_a + (\regret^T_a)^+}^2\).
They then state that
\(\subex*{Q^{T - 1}_a + (\regret^T_a)^+}^2 \le \subex*{Q^{T - 1}_a + \regret^T_a}^2 + \Delta^2\),
where \(\Delta\) is the diameter of the reward range. This is false in
general: \begin{align}
  \subex*{Q^{T - 1}_a + (\regret^T_a)^+}^2
    &= \subex*{Q^{T - 1}_a}^2 + \subex*{\regret^T_a}^2 + 2 Q^{T - 1}_a (\regret^T_a)^+ \\
    &\leq \subex*{Q^{T - 1}_a}^2 + \subex*{\regret^T_a}^2 + 2 Q^{T - 1}_a \regret^T_a + 2 Q^{T - 1}_a \Delta \\
    \label{eq:rmpp-inequality}
    &= \subex*{Q^{T - 1}_a + \regret^T_a}^2 + 2 Q^{T - 1}_a \Delta,
\end{align} where \(2 Q^{T - 1}_a \Delta > \Delta^2\) if
\(Q^{T - 1}_a > \Delta / 2\).
There are scenarios where \cref{eq:rmpp-inequality} is tight so it is unclear how this bound could be improved.
Attempting the rest of the proof, we get
\begin{align*}
  \sum_a \subex*{Q^{T - 1}_a + (\regret^T_a)^+}^2
    &\le \abs{\Actions} \Delta^2 + \sum_a \subex*{Q^{T - 1}_a}^2 + 2 \Delta \sum_a Q^{T - 1}_a.
\end{align*} Unrolling the recursion exactly is messy, but the extra
\(2 \Delta \sum_a Q^{T - 1}_a\) term ensures that the bound will be no
smaller than
\(\sum_a Q^{T - 1}_a + (\regret^T_a)^+ \le 2^{\frac{T}{2}} \abs{\Actions}^{\frac{T - 1}{2}} \Delta^T\).

\section{Experiments}
\subsection{Games}

\subsubsection{Leduc Hold'em Poker}
Leduc hold'em poker~\citep{southey05} is a two-player poker game with a deck of six cards (two suits and three ranks).
At the start of the game, both players ante one chip and receive one private card.
There are two betting rounds and there is a maximum of two raises on each round.
Bet sizes are limited to two chips in the first round and four in the second.
If one player folds, the other wins.
At the start of the second round, a public card is revealed.
A showdown occurs at the end of the second round if no player folds.
The strongest hand in a showdown is a pair (using the public card), and if no player pairs, players compare the ranks of their private cards.
The player with the stronger hand takes all chips in the pot or players split the pot if their hands have the same strength.
Payoffs are reported in milli-big blinds (mbb) (where the ante is considered a big blind) for consistency with the way performance is reported in other poker games.

\subsubsection{Imperfect Information Goofspiel}
Imperfect information goofspiel~\cite{Ross71Goofspiel,lanctot13phdthesis} is a bidding game for $N$ players.
Each player is given a hand of $n$ ranks that they play to bid on $n$ point cards.
On each round, one point card is revealed and each player simultaneously bids on the point card.
The point cards might be sorted in ascending order ($\uparrow$), descending order ($\downarrow$), or they might be shuffled ($R$).
If there is one bid that is greater than all the others, the player who made that bid wins the point card.
If there is a draw, the bid card is instead discarded.
The player with the most points wins so payoffs are reported in win percentage.
We use five goofspiel variants:
\begin{itemize}
  \item two-player, 5-ranks, ascending (goofspiel$(5, \uparrow, N = 2)$, denoted as \gTwoFive/ in the main paper),
  \item two-player, 5-ranks, descending (goofspiel$(5, \downarrow, N = 2)$),
  \item two-player, 4-ranks, random (goofspiel$(4, R, N = 2)$),
  \item three-player, 4-ranks, ascending (goofspiel$(4, \uparrow, N = 3)$, denoted as \gThreeFour/ in the main paper), and
  \item three-player, 4-ranks, descending (goofspiel$(4, \downarrow, N = 3)$).
\end{itemize}

\subsubsection{Sheriff}
Sheriff is a two-player, non-zero-sum negotiation game resembling the Sheriff of Nottingham board game and it was introduced by \textcite{Farina19:Correlation}.
At the beginning of the game, the ``smuggler'' player chooses zero or more illegal items (maximum of three) to add to their cargo.
The rest of the game proceeds over four rounds.

At the beginning of each round, the smuggler signals how much they would be willing to pay the ``sheriff'' player to bribe them into not inspecting the smuggler's cargo, between zero and three.
The sheriff responds by signalling whether or not they would inspect the cargo.
On the last round, the bribe amount chosen by the smuggler and the sheriff's decision about whether or not to inspect the cargo are binding.

If the cargo is not inspected, then the smuggler receives a payoff equal to the number of illegal items included within, minus their bribe amount, and the sheriff receives the bribe amount.
Otherwise, the sheriff inspects the cargo.
If the sheriff finds an illegal item, then the sheriff forces the smuggler to pay them two times the number of illegal items.
Otherwise, the sheriff compensates the smuggler by paying them three.

\subsubsection{Tiny Bridge}
A miniature version of bridge created by Edward Lockhart, inspired by a research project at University of Alberta by Michael Bowling, Kate Davison, and Nathan Sturtevant.
We use the smaller two-player rather than the full four-player version.
See the implementation from \textcite{LanctotEtAl2019OpenSpiel} for more details.

\subsubsection{Tiny Hanabi}
A miniature two-player version of Hanabi described by \textcite{foerster2019bayesian}.
The game is fully cooperative and the optimal score is ten.
Both players take only one action so all EFR instances collapse except when they differ in their choice of $\DevSet_{\infoSet}$.

\subsection{Alternative $\DevSet_{\infoSet}$ Choices}
When implementing EFR for deviations that set the action transformations at each information set to the internal transformations, we have the option of implementing these variants by using the union of the internal and external transformations without substantially changing the variant's theoretical properties.
We test how this impacts practical performance within EFR variants for informed counterfactual deviations, CFPS deviations, and TIPS deviations.
These variants have an ``\EXT + \INT'' subscript.

\subsection{Results}
We present four sets of figures to summarize the performance of each EFR variant in the fixed and simultaneous regimes described in Section 7.

The first three sets of figures illustrate how each variant performs on average in each round individually.
\cref{fig:lcAvgFixed,fig:lcAvgSim} show the running average expected payoff of each variant over rounds, averaged over play with all EFR variants (including itself).
These figures summarize the progress that each variant makes over rounds to adapt to and correlate with its companion variant, on average.
\cref{fig:lcInstFixed,fig:lcInstSim} show the instantaneous expected payoff of each variant over rounds, averaged over play with all EFR variants.
\cref{fig:runtimeLcAvgFixed,fig:runtimeLcAvgSim} show the same data as in \cref{fig:lcAvgFixed,fig:lcAvgSim} except according to runtime rather than rounds.
Tiny Hanabi is omitted because it is too small to make meaningful runtime comparisons between EFR variants.

\cref{fig:heatmaps} show the average expected payoff of each variant paired with each other variant (including itself) after 1000 rounds.
These figures summarize how well each variant works with each other variant.

\begin{figure}[htbp]
\centering
  \includegraphics[width=0.9\linewidth]{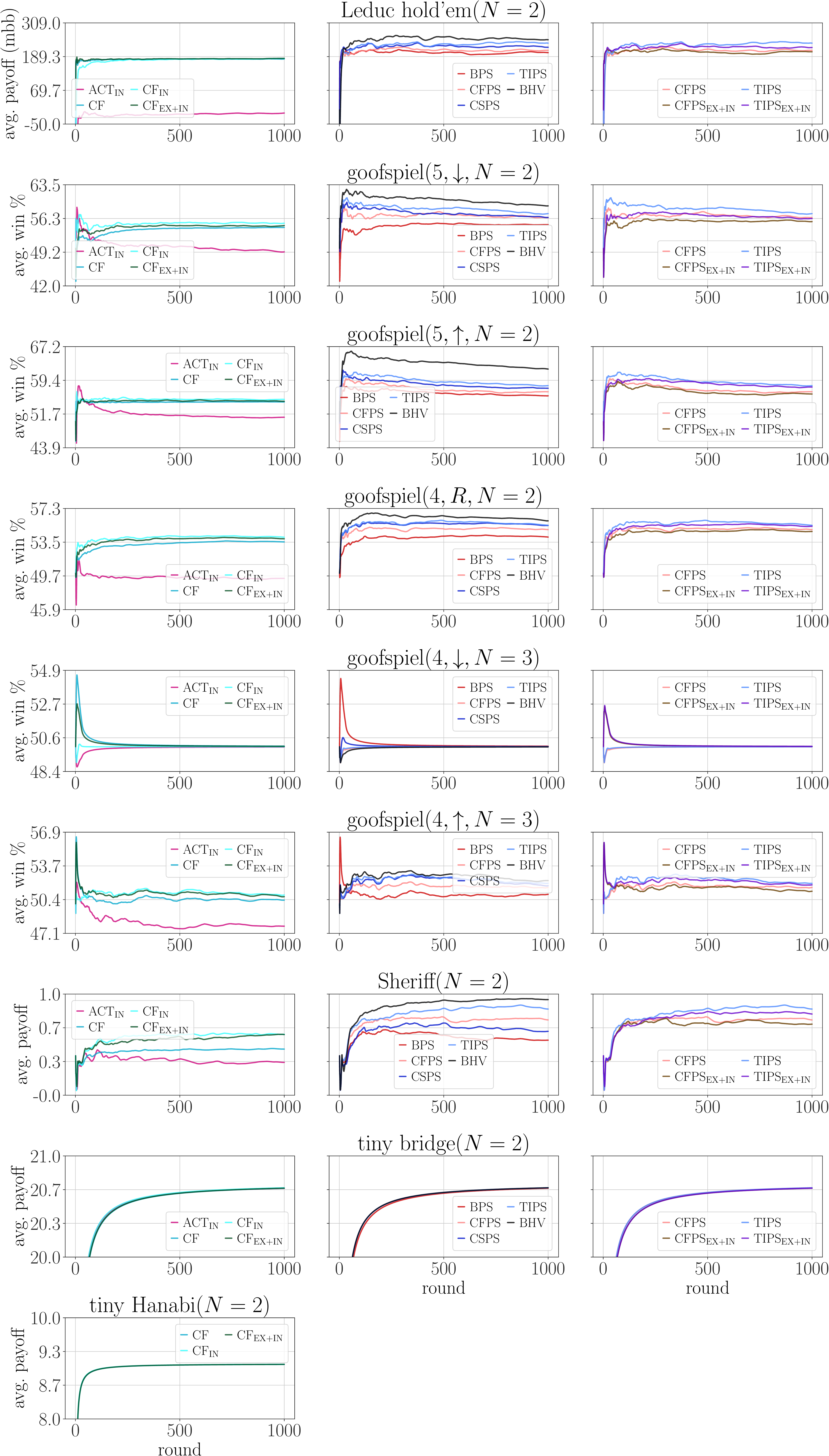}
  \caption{The expected payoff accumulated by each EFR variant over rounds averaged over play with all EFR variants in each game in the fixed regime.}
  \label{fig:lcAvgFixed}
\end{figure}

\begin{figure}[htbp]
\centering
  \includegraphics[width=0.9\linewidth]{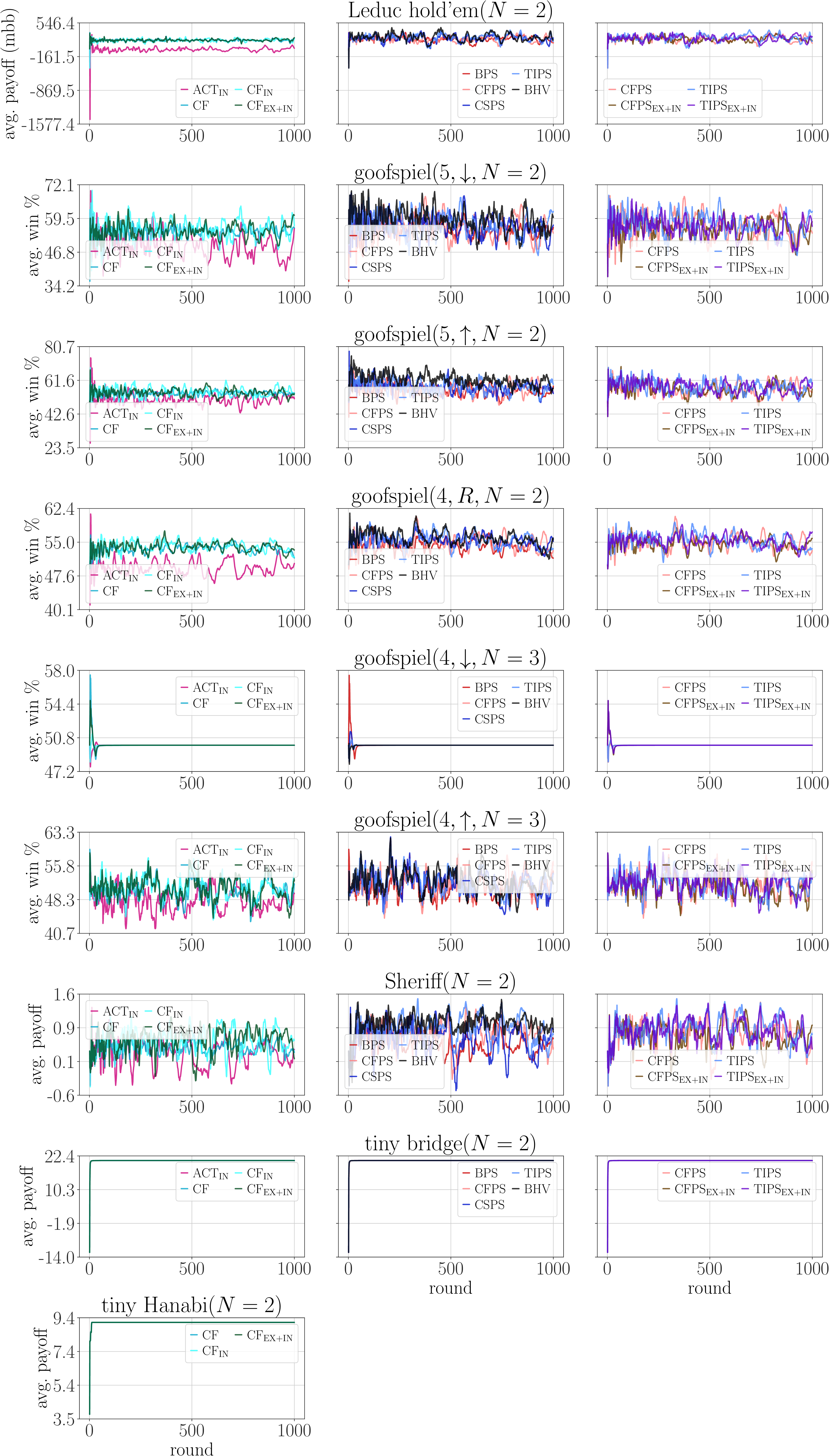}
  \caption{The instantaneous payoff achieved by each EFR variant on each round averaged over play with all EFR variants in each game in the fixed regime.}
  \label{fig:lcInstFixed}
\end{figure}

\begin{figure}[htbp]
\centering
  \includegraphics[width=0.9\linewidth]{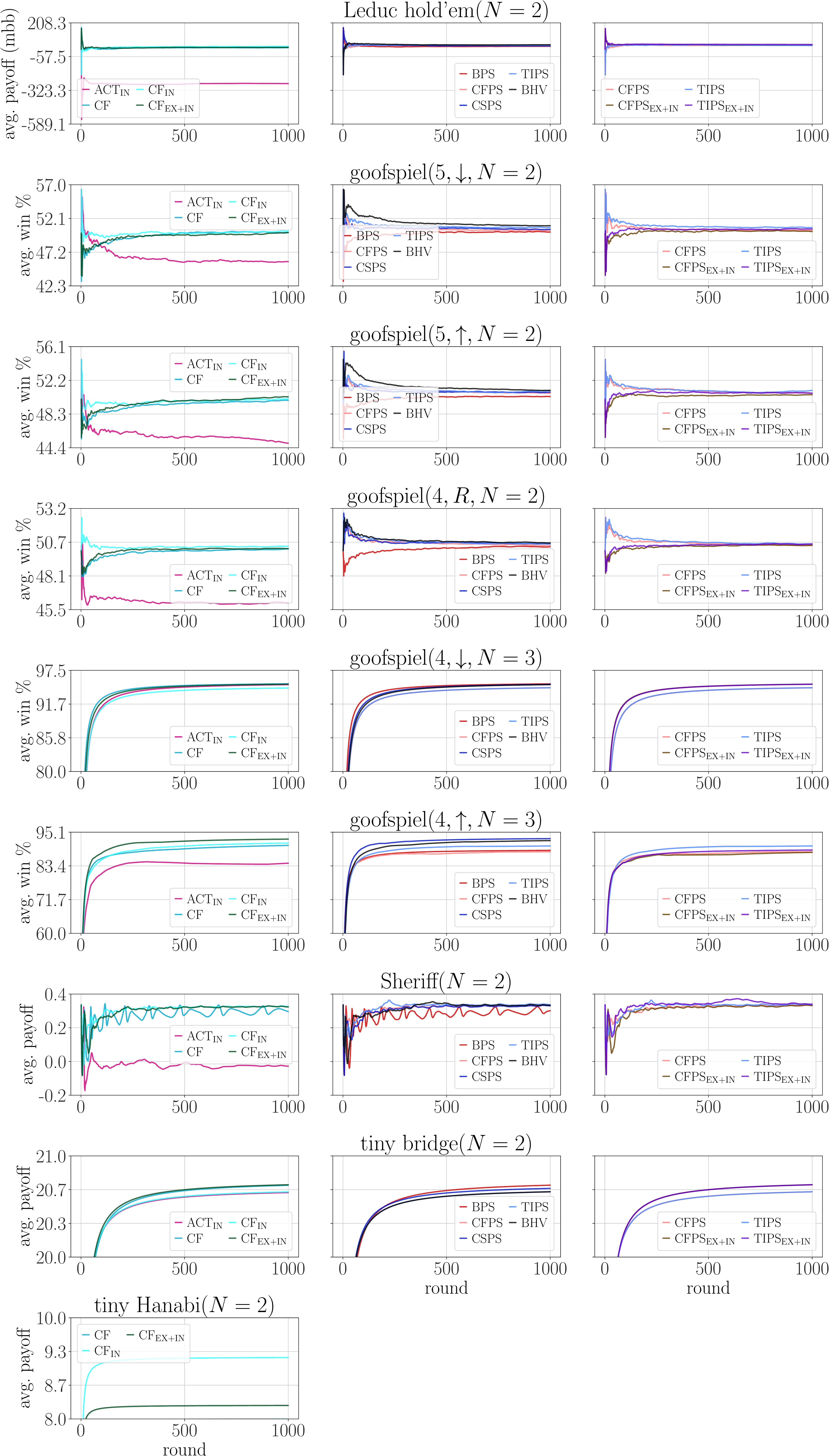}
  \caption{The expected payoff accumulated by each EFR variant over rounds averaged over play with all EFR variants in each game in the simultaneous regime.}
  \label{fig:lcAvgSim}
\end{figure}

\begin{figure}[htbp]
\centering
  \includegraphics[width=0.9\linewidth]{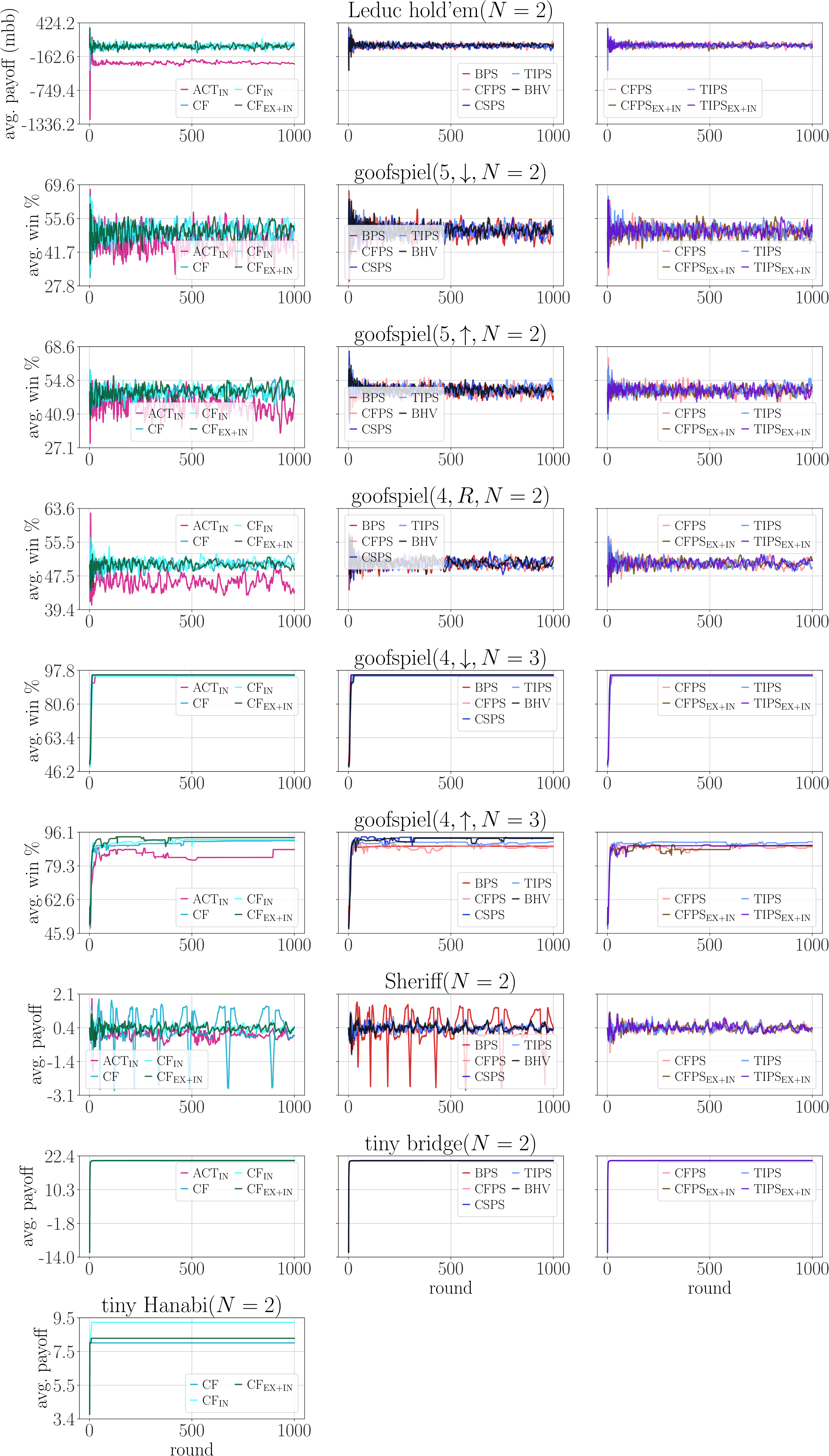}
  \caption{The instantaneous payoff achieved by each EFR variant on each round averaged over play with all EFR variants in each game in the simultaneous regime.}
  \label{fig:lcInstSim}
\end{figure}

\begin{figure}[htbp]
\centering
  \includegraphics[width=0.9\linewidth]{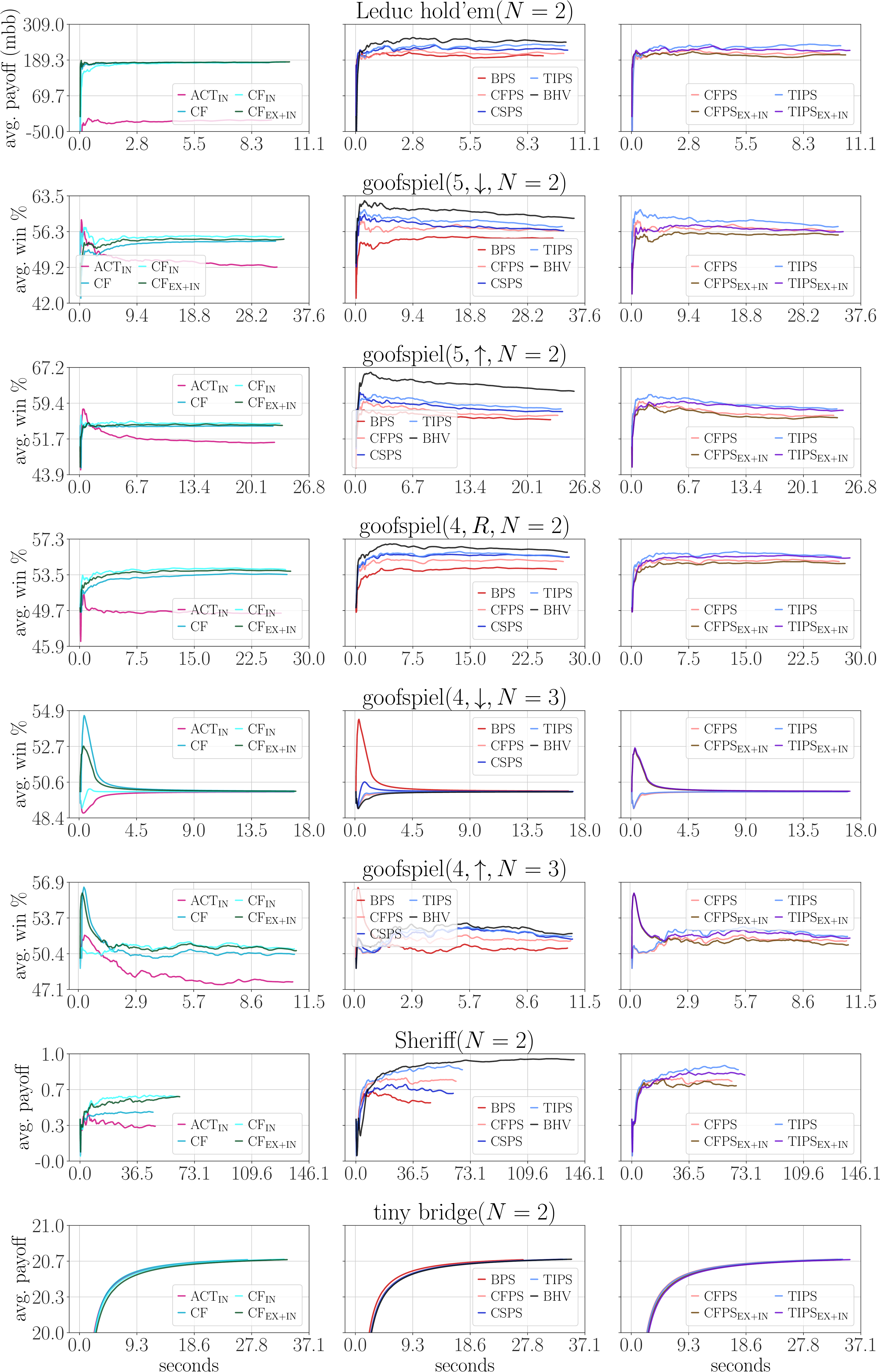}
  \caption{The expected payoff accumulated by each EFR variant over runtime averaged over play with all EFR variants in each game in the fixed regime.}
  \label{fig:runtimeLcAvgFixed}
\end{figure}

\begin{figure}[htbp]
\centering
  \includegraphics[width=0.9\linewidth]{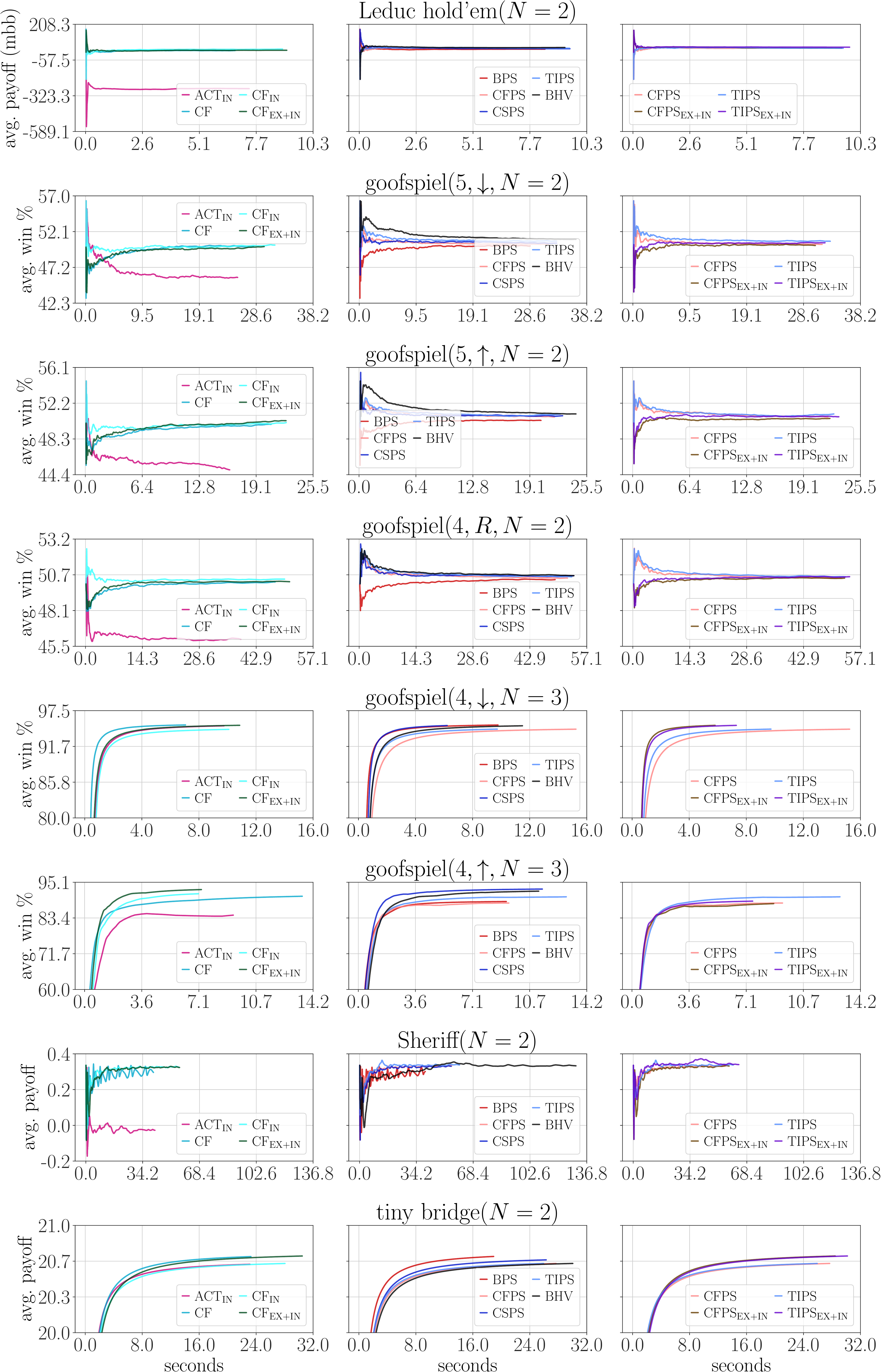}
  \caption{The expected payoff accumulated by each EFR variant over runtime averaged over play with all EFR variants in each game in the simultaneous regime.}
  \label{fig:runtimeLcAvgSim}
\end{figure}

\begin{figure}[htbp]
\centering
  \centerline{\includegraphics[width=1.3\linewidth]{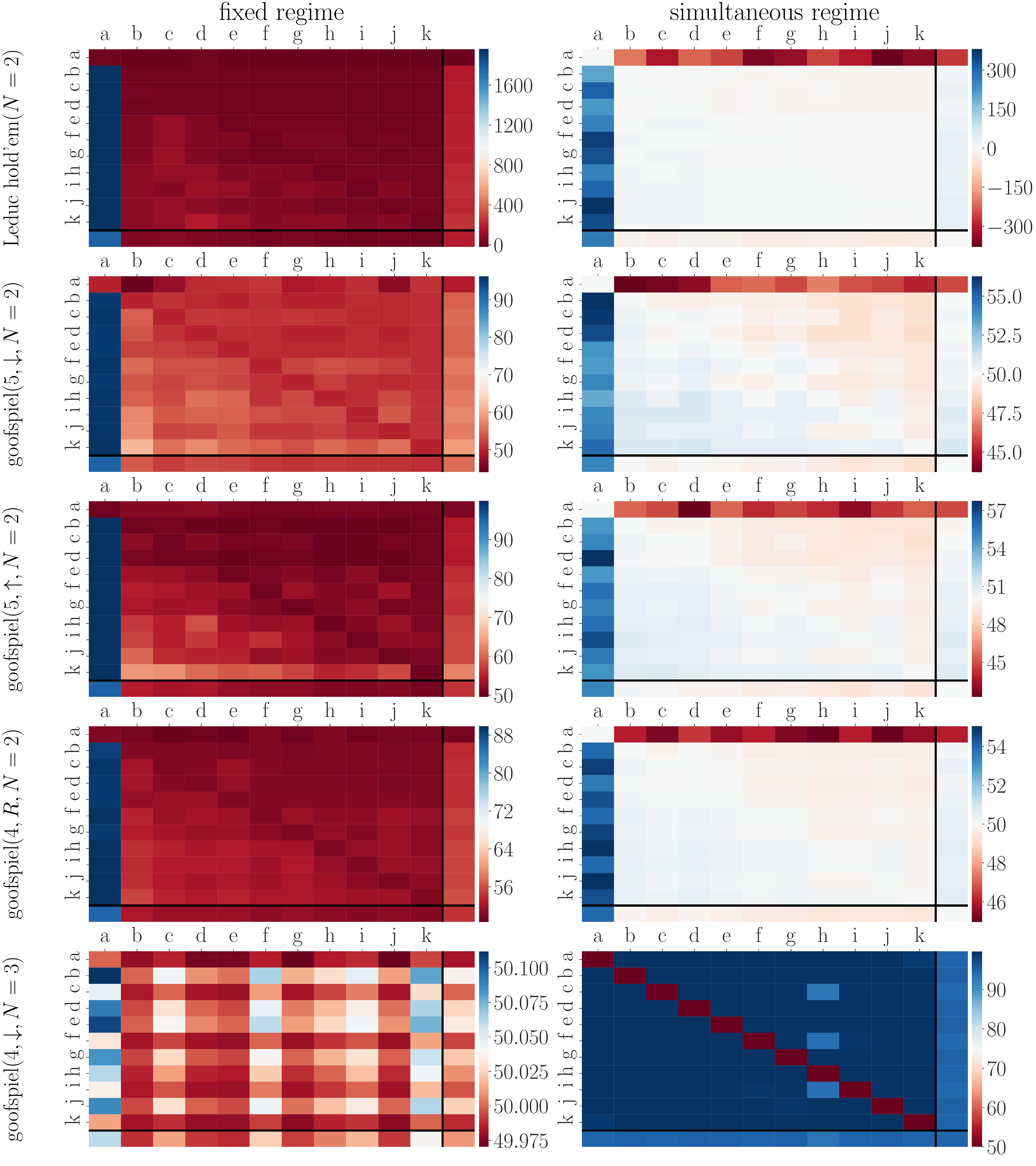}}
  \caption{(1 / 2) The average expected payoff accumulated by each EFR variant (listed by row) from playing with each other EFR variant (listed by column) in each game after 1000 rounds where
    a $\to$ ACT\textsubscript{\IN}, b $\to$ CF, c $\to$ CF\textsubscript{\IN}, d $\to$ CF\textsubscript{\EXT+\IN}, e $\to$ BPS, f $\to$ CFPS, g $\to$ CFPS\textsubscript{\EXT+\IN},
    h $\to$ CSPS, i $\to$ TIPS, j $\to$ TIPS\textsubscript{\EXT+\IN}, k $\to$ BHV. The bottom rows and farthest right columns represent the column and row averages, respectively.}
  \label{fig:heatmaps}
\end{figure}

\begin{figure}[htbp]
\centering
  \centerline{\includegraphics[width=1.3\linewidth]{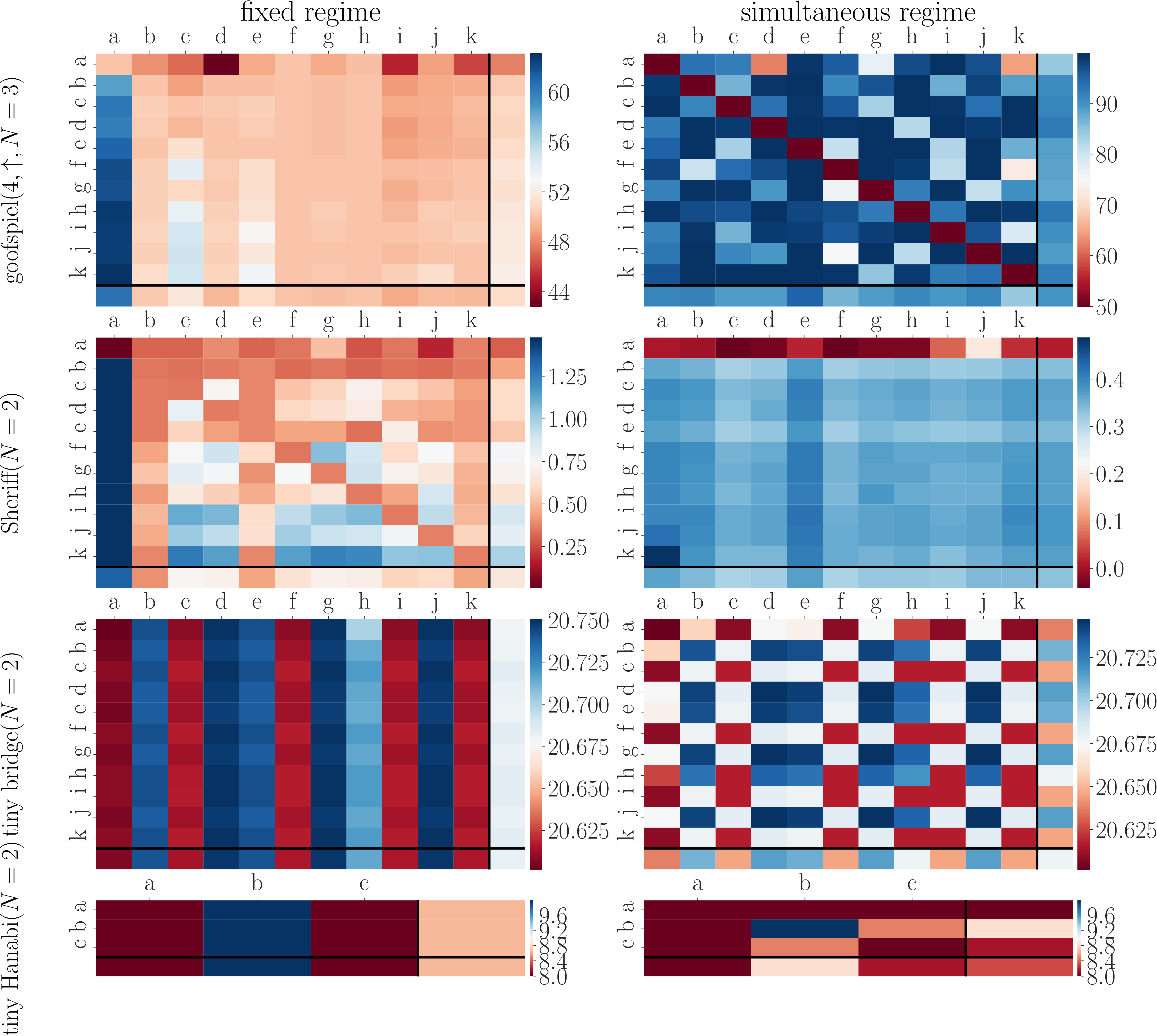}}
  \caption{(2 / 2) The average expected payoff accumulated by each EFR variant (listed by row) from playing with each other EFR variant (listed by column) in each game after 1000 rounds where
    a $\to$ ACT\textsubscript{\IN}, b $\to$ CF, c $\to$ CF\textsubscript{\IN}, d $\to$ CF\textsubscript{\EXT+\IN}, e $\to$ BPS, f $\to$ CFPS, g $\to$ CFPS\textsubscript{\EXT+\IN},
    h $\to$ CSPS, i $\to$ TIPS, j $\to$ TIPS\textsubscript{\EXT+\IN}, k $\to$ BHV. The bottom rows and farthest right columns represent the column and row averages, respectively.}
  \label{fig:heatmaps}
\end{figure}
  
\end{document}